\documentclass{article}
\usepackage{arxiv}

\pdfoutput=1

\usepackage{amssymb}
\usepackage{amsmath}
\usepackage{amsthm}

\usepackage[utf8]{inputenc} % allow utf-8 input
\usepackage[T1]{fontenc}    % use 8-bit T1 fonts
\usepackage{hyperref}       % hyperlinks
\usepackage{url}            % simple URL typesetting
\usepackage{booktabs}       % professional-quality tables
\usepackage{amsfonts}       % blackboard math symbols
\usepackage{nicefrac}       % compact symbols for 1/2, etc.
\usepackage{microtype}      % microtypography
\usepackage{cleveref}       % smart cross-referencing
\usepackage{graphicx}
\usepackage{doi}
\usepackage[ruled,linesnumbered]{algorithm2e}
\usepackage{xspace}
\usepackage{array}

\usepackage[center,scriptsize]{subfigure}
\usepackage{multirow,amssymb,microtype}
\usepackage{array,bm}
\usepackage{tikz}
\usepackage{thmtools, thm-restate}
\usepackage{enumitem}
\usetikzlibrary{decorations.pathmorphing}
\newcolumntype{H}{>{\setbox0=\hbox\bgroup}c<{\egroup}@{}}
\newtheorem{theorem}{Theorem}
\newtheorem{definition}{Definition}
\newtheorem{invariant}{Invariant}
\newtheorem{observation}{Observation}
\newtheorem{corollary}{Corollary}
\newtheorem{lemma}{Lemma}

\newcommand {\NODES}{\mathcal{N}}
\newcommand {\fullyupdate}{{\sc fully-update} }
\newcommand {\fullyupdateend}{{\sc fully-update}}
\newcommand {\fullydynamic}{{\sc fully-dynamic} }
\newcommand {\fullydynamicend}{{\sc fully-dynamic}}
\newcommand {\fullycleanup}{{\sc fully-cleanup} }
\newcommand {\fullycleanupend}{{\sc fully-cleanup}}
\newcommand {\fullyfixup}{{\sc fully-fixup} }
\newcommand {\fullyfixupend}{{\sc fully-fixup}}
\newcommand {\cb}{\beta}
\newcommand {\vstar} {{\nu^*}}
\newcommand {\fresh} {fresh }
\newcommand {\freshend} {fresh}
\newcommand {\system}{tuple system }
\newcommand {\systemend}{tuple system}

\newcommand {\weight} {{\bf {w}}}

\newcommand {\dict} {dict}
\newcommand {\MT} {Marked-Tuples }
\newcommand {\MTend} {Marked-Tuples}

\newcommand {\new} {new}

\newcommand{\vone}{\vspace{.1in}}
\newcommand{\vhalf}{\vspace{.05in}}
\newcommand {\levelgraph}{\Gamma}

\newcommand {\CA}{C}
\newcommand {\CH}{CH}

\mathchardef\mh="2D

\newcommand {\setbit}{set\mh bit}
\newcommand {\priortimes}{Prior\mh times}

\newcommand{\hide}[1]{}
\newcommand{\Xomit}[1]{}

\newcommand{\DI}{\texttt{DI }}

\newcommand{\NPR}{\texttt{NPRdec }}
\newcommand{\Tho}{\texttt{Thorup }}
\newcommand{\FFD}{{\sc ffd }}
\newcommand{\DIe}{\texttt{DI}}

\newcommand{\NPRe}{\texttt{NPRdec}}
\newcommand{\Thoe}{\texttt{Thorup}}
\newcommand{\FFDe}{{\sc ffd}}
\newcommand{\BDAGe}{{\sc build-dag}}
\newcommand {\ffullyupdate}{{\sc ff-update } }
\newcommand {\ffullyupdateend}{{\sc ff-update}}
\newcommand {\ffullydynamic}{{\sc ffd} }
\newcommand {\ffullydynamicend}{{\sc ffd}}
\newcommand {\ffullycleanup}{{\sc ff-cleanup} }
\newcommand {\ffullycleanupend}{{\sc ff-cleanup}}
\newcommand {\ffullyfixup}{{\sc ff-fixup} }
\newcommand {\ffullyfixupend}{{\sc ff-fixup}}
\newcommand {\ffullypopulate}{{\sc ff-populate-heap}}
\newcommand {\ffullygetnew}{{\sc ff-new-paths}}
\newcommand {\ffullyccenters}{{\sc ff-cleanup-centers }}

\newcommand {\ffullyccentersend}{{\sc ff-cleanup-centers}}
\newcommand {\ffullyfcentersend}{{\sc ff-fixup-centers}}
\newcommand {\LHT}{LHT }
\newcommand {\LHTe}{LHT}
\newcommand {\HT}{HT }

\newcommand {\DMs}{DL}
\newcommand {\LN}{LN}
\newcommand {\RN}{RN}

\newcommand {\incremental}{decrease-only\xspace}
\newcommand {\decremental}{increase-only\xspace}

\title{Fully Dynamic Algorithms for All Pairs All Shortest Paths and Betweenness Centrality}

\date{}

\author{ Matteo Pontecorvi\\
	Nokia Bell Labs\\
	Paris-Saclay, France \\
	\texttt{matteo.pontecorvi@nokia.com} \\
	\And
	Vijaya Ramachandran \\
	Department of Computer Science\\
	University of Texas at Austin\\
	Texas, USA \\
	\texttt{vlr@cs.utexas.edu} \\
}

\renewcommand{\headeright}{}
\renewcommand{\undertitle}{}
\renewcommand{\shorttitle}{}

\begin{document}
\maketitle

\begin{abstract}
We consider the all pairs all shortest paths (APASP)
problem, which maintains
all of the multiple shortest paths for every vertex pair
in a directed graph $G=(V,E)$ with a positive real weight on each edge. We present 
fully dynamic algorithms for this problem in which an update supports either weight increases or
weight decreases on a subset of edges incident to a vertex.

Our basic algorithm runs in amortized
$O(\vstar^2 \cdot \log^3 n)$ time per update, where
$n = |V| $, and $\vstar$ 
bounds  the
number of edges
that lie on shortest paths through any single vertex.
Our method is a generalization and a variant of the fully dynamic algorithm of Demetrescu and Italiano~\cite{DI04}
for unique shortest paths (APSP), and it builds on recent work on \decremental APASP~\cite{NPR14b}.
Our algorithm matches the
fully dynamic amortized 
bound in~\cite{DI04} (which works only for single shortest paths) for graphs with $O(n)$-size shortest path dags, though our method, and especially
its  analysis, are different.
We also present a faster APASP algorithm where we achieve an amortized $O(\vstar^2 \cdot \log^2 n)$ time per update.
% This result improves on the amortized bound for the basic algorithm by a logarithmic factor.
 Our faster algorithm uses new data structures and techniques that are extensions of the method in the fully dynamic algorithm in Thorup \cite{Thorup04} for APSP in graphs with unique shortest paths.
%For graphs with $\nu^* = O(n)$, our faster algorithm matches the fully dynamic APSP
%bound in Thorup~\cite{Thorup04}, which holds for graphs with $\nu^* = n-1$, since it assumes unique shortest paths.

Our APASP algorithms  lead to fully dynamic 
computation, with the same respective amortized bounds, for betweenness centrality,
which is a parameter widely used in the analysis of
large complex networks.
Moreover, it provides an alternate amortized analysis for the fully dynamic APSP algorithm of Demetrescu and Italiano, and new algorithms for computing APSP under fully dynamic updates.
\end{abstract}

\section{Introduction}

Given a directed graph $G = (V, E)$ with positive edge weights, we consider
the problem of maintaining  {\it AP\underline{A}SP (all pairs \underline{all} shortest paths)}~\cite{NPR14b}, i.e., 
the set of all shortest paths (SPs) between all pairs
of vertices. This is a fundamental graph property, and it also enables the efficient computation of
{\em betweenness centrality} 
for every vertex
in the graph.

Betweenness centrality (BC) is a widely
used measure for the analysis of social networks.
Shortest paths starting from a node $s$ can be represented by the single source shortest paths (SSSP) \emph{out-dag} (paths from the root $s$ to the leaves of the dag) rooted at $s$, similarly the shortest paths ending in $t$ can be represented by the SSSP \emph{in-dag} (paths from the leaves to the root $t$ of the dag) rooted at $t$. In the following, we will refer to these as SP in-dag and SP out-dag.

In this paper we present a fully dynamic 
algorithm for the APASP problem, where each update in $G$ 
is either {\it \incremental} or {\it \decremental}. A \incremental update 
either
inserts a new vertex along with incident edges of finite weight,
 or decreases the weights of some existing edges incident on a vertex.
An \decremental update
 deletes an existing vertex, or increases the
weights of some edges incident on a vertex.

Recently, we  gave (with Nasre)  a simple \incremental APASP algorithm~\cite{NPR14}, and a
more involved \decremental  APASP algorithm~\cite{NPR14b}.
Neither of these algorithms is correct for the fully dynamic case.
The fully dynamic methods that we present in this paper 
  build on~\cite{NPR14b}, and also incorporate a variant of the fully dynamic methods developed by
Demetrescu and Italiano~\cite{DI04} (the \DI method) and Thorup~\cite{Thorup04} (the \Tho method) for APSP  
 where only one shortest path is maintained for each
pair of vertices.
The \DI algorithm runs in $O(n^2 \cdot \log^3 n)$ amortized
time per update,
where $n=|V|$.
The \Tho algorithm is faster by a logarithmic factor but is considerably more complicated, even for the unique shortest paths
case. 
The algorithms in both~\cite{DI04} and \cite{Thorup04} use the unique shortest paths 
assumption crucially, and need considerable enhancements even to maintain a small number of
multiple shortest paths correctly.
In \cite{Thorup05} the author presented a worst case fully dynamic algorithm for APSP which recomputes the complete distance matrix in $\tilde{O}(n^{2.75})$. Recently, the worst case time bound was improved in \cite{Abraham} with a randomized algorithm which runs in $O(n^{2+2/3}\log^{4/3}n)$.

In this paper we present two fully dynamic algorithms for APASP, which maintain all of the multiple
shortest paths for every pair of vertices. Our basic algorithm is
as simple as the \DI method (though somewhat different) when specialized to unique shortest
paths.
In fact, it matches the \DI bound  for graphs with  a constant number of (or unique) shortest paths, while 
being applicable to the more general APASP problem.
Moreover, it provides a new amortized analysis for the fully dynamic \DI algorithm if our `dummy sequence' replaces the one used in \DIe.
Our second algorithm
improves the amortized time bound by a logarithmic factor
using data structures and techniques that are considerably more complicated.

Both APASP methods give fully dynamic algorithms for the betweenness centrality (BC) problem, described 
below.

\vhalf
\noindent
{\bf Betweenness Centrality (BC).}
Betweenness centrality is a widely-used measure
in the analysis of large complex networks, and is defined as follows.
For any pair $x, y$ in $V$, let $\sigma_{xy}$ denote the number of shortest paths
from $x$ to $y$ in $G$, and
let
 $\sigma_{xy}(v)$ denote the number of shortest
paths from $x$ to $y$ in $G$ that pass through $v$. Then, $BC(v) = \sum_{s \neq v, t \neq v} \frac{\sigma_{st}(v)}{\sigma_{st}}$.
This measure is often used as an index that  determines the relative importance of
$v$ in the network.
Some applications of
BC
include analyzing social interaction networks \cite{KA12},
identifying lethality in biological networks \cite{BCBN},
and identifying key actors in terrorist networks \cite{Coffman,Krebs02}.

Heuristics for dynamic betweenness centrality with good experimental
performance are given in~\cite{GreenMB12,Lee12,SinghGIS13},
but none of these heuristics
provably improve on the widely used static algorithm by
Brandes~\cite{Brandes01}, which runs in $O(mn + n^2 \log n)$ time (where $m=|E|)$
on any class of graphs. The only earlier results that provably improve on
Brandes on 
some 
classes of graphs are the recent separate \incremental and \decremental
algorithms in~\cite{NPR14,NPR14b}.
Table I contains a summary of these results.

The results in this paper are mainly for fully-dynamic APASP, which is an important graph-theoretic property of independent interest. However, we show that our
fully dynamic APASP algorithms
give fully-dynamic BC algorithms with the same bounds.

\begin{table}[ht]
	\begin{center}
		{%\scriptsize
			\begin{tabular}{c | H c | c | c | c | c  }
				\hline
				\textbf{Paper} & \textbf{Year} & \textbf{Time} & \textbf{Weights} & \textbf{Update Type} & \textbf{DR/UN} & \textbf{Result} \\ 
				\hline
				\hline
				Brandes~\cite{Brandes01} & 2001 & $O(mn)$ & NO & Static Alg. & Both & Exact \\
				Brandes~\cite{Brandes01} & 2001 & $O(mn+n^2\log n)$ & YES & Static Alg. & Both & Exact \\
				Geisberger~et~al.~\cite{Geisb} & 2007 & Heuristic & YES & Static Alg. & Both & Approx. \\
				Riondato~et~al.~\cite{Riondato} & 2014 & depends on $\epsilon$ & YES & Static Alg. & Both & $\epsilon$-Approx. \\
				\hline
				\multicolumn{7}{c}{\textbf{Semi Dynamic}} \\
				\hline
				Green~et~al.~\cite{GreenMB12} & 2012 & $O(mn)$ & NO & Edge Inc. & Both & Exact \\ 
				Kas~et~al.~\cite{KasWCC13} & 2013 & Heuristic & YES & Edge Inc. & Both & Exact \\ 
				NPR~\cite{NPR14} & 2014 & $O(\vstar \cdot n)$ & YES & Vertex Inc. & Both & Exact \\ 
				NPRdec~\cite{NPR14b} & 2014 & $O(\vstar^2 \cdot \log n)$ & YES & Vertex Dec. & Both & Exact \\ 
				Bergamini~et~al.~\cite{Bergamini1} & 2015 & depends on $\epsilon$ & YES & Batch (edges) Inc. & Both & $\epsilon$-Approx. \\ 
				\hline 
				\multicolumn{7}{c}{\textbf{Fully Dynamic}} \\
				\hline
				Lee~et~al.~\cite{Lee12} & 2012 & Heuristic & NO & Edge Update & UN & Exact \\
				Singh~et~al.~\cite{SinghGIS13} & 2013 & Heuristic & NO & Vertex Update & UN & Exact \\ 
				Kourtellis+~\cite{Kourtellis} & 2014 & $O(mn)$ & NO & Edge Update & Both & Exact \\
				Bergamini~et~al.~\cite{Bergamini2} & 2015 & depends on $\epsilon$ & YES & Batch (edges) & UN & $\epsilon$-Approx. \\
				\textbf{This paper (Basic)}  & \textbf{2015} & ${O(\vstar^2 \cdot \log^3 n)}$ & \textbf{YES} & \textbf{Vertex Update} & \textbf{Both} & \textbf{Exact} \\ 
				\textbf{This paper (\FFDe)} & \textbf{2015} & $O(\vstar^2 \cdot \log^2 n)$ & \textbf{YES} & \textbf{Vertex Update} & \textbf{Both} & \textbf{Exact} \\ 
				\hline 
			\end{tabular} 
	}
		\caption{Related results (DR stands for Directed and UN for Undirected)}
	\end{center}
	\label{table1}
\end{table}

\noindent
{\bf Our Results.}
 Let $\vstar$ be the maximum
number of edges that lie on shortest paths through any given vertex in $G$.
Our main results are the following theorems, where we assume
for convenience
 $\vstar=\Omega (n)$. We present two algorithms: \fullydynamic (basic) and \FFD (faster fully dynamic) with bounds stated in the following theorem.

\begin{theorem}\label{th:main}
Let $\Sigma$ be a sequence of $\Omega (n)$ 
fully dynamic APASP updates on an $n$-node graph $G=(V,E)$. Then,
\begin{enumerate}
	\item algorithm \fullydynamic mantains APASP
	and all BC scores in amortized time $O( \vstar^2 \cdot \log^3 {n})$ per update,
	\item algorithm \FFD mantains APASP
	and all BC scores in amortized time $O( \vstar^2 \cdot \log^2 {n})$ per update,
\end{enumerate}
 where $\nu^*$ bounds the
number of distinct edges that 
lie on shortest paths through any given vertex in any of the updated graphs or their vertex induced subgraphs.
\end{theorem}

Our algorithms are provably faster than 
Brandes on dense graphs with succinct single-source SP dags. 
Our techniques rely on recomputing BC scores using certain data structures related to shortest paths extensions (see Section \ref{sec:fdbc}). These are generalizations of structures introduced by Demetrescu and Italiano  \cite{DI04} for fully dynamic APSP 
and \Thoe, where only one SP is maintained for each pair of vertices. 
Our generalizations build on the
\system introduced in~\cite{NPR14b} for \decremental APASP (see Section \ref{sec:NPR}), which is a method to
succinctly  represent all of the multiple SPs for every pair of vertices. Additionally, in algorithm \FFDe, one of the main challenges we address is to  generalize the `level graphs' of \Tho (see Section \ref{sec:improved}) to the case when different SPs for a given vertex pair can be distributed across multiple levels.

Of independent interest is a new amortized analysis for the fully dynamic \DI algorithm which we obtain using a new `dummy updates' sequence in our \fullydynamic algorithm. 

\vhalf
\noindent
{\bf Discussion of the Parameter $\bm{\vstar}$.}
Let $m^*$ be the number of edges in $G$ that lie on shortest paths.
As  noted in \cite{KKP93}, it
is well-known that $m^* = O(n \log{n})$
with high probability in a complete graph where edge weights are chosen from
a large class of probability distributions.
Since $\vstar \leq m^*$,
our algorithms will have an amortized bound of 
 $O(n^2 \cdot polylog (n))$ on
such graphs.  
Also,
$\vstar =O(n)$ in any graph with only $k$ shortest paths, where $k$ is a constant, between every
pair of vertices.
These graphs are called $k$-geodetic \cite{kgeo}, and are well studied in graph theory \cite{bigeo,bandelt,mulder}.
In fact $\vstar = O(n)$ even in some graphs that have an exponential number of SPs between some pairs of vertices.
In contrast, $m^*$ can be $\Theta (n^2)$  even in some graphs with unique SPs, for example the complete unweighted
graph $K_n$.
Another type of graph with  $\vstar \ll m^*$ is one with large clusters of nodes, as described by the \emph{planted $\ell$-partition model} \cite{Condon,clusters}:
consider a graph $H$ with $k$ clusters of size  $n/k$ (for some constant $k \geq 1$) with
$\delta < w(e) \leq 2\delta$, for some constant $\delta>0$,  for each edge $e$ in a cluster;
between the clusters is a sparse
interconnect. Then
$m^* = \Omega(n^2)$ but $\vstar= O(n)$. 

\noindent
In all such cases our improved algorithm \FFD
will run in amortized $O(n^2 \log^2 n)$
time per update.
Thus we have:

\begin{theorem}\label{th:main2} 
Let $\Sigma$ be a sequence of $\Omega (n)$ updates on graphs with $O(n)$ distinct edges on shortest paths through any single vertex in any vertex-induced subgraph.
Then, APASP
and all  BC scores can be
maintained in amortized time $O( n^2 \cdot \log^2 {n})$ per update.
\end{theorem}

\begin{corollary}\label{cor:main}
If the number of shortest paths through any single vertex 
is bounded by a constant, then
fully dynamic APASP and BC have
amortized cost $O(n^2 \cdot \log^2 n)$  per update if the update sequence has length $\Omega(n)$.
\end{corollary}

\noindent
Both our algorithms use $\widetilde{O}(m \cdot \vstar)$ space, extending the $\widetilde{O}(mn)$ 
\DI result for APSP. Brandes uses only linear space, but all
known
dynamic algorithms require at least 
$\Omega (n^2)$ space.

\vhalf
\noindent
{\bf Overview of the Paper.} 
In Section \ref{sec:fdbc} we describe a very simple method to obtain a fully dynamic BC algorithm using a fully dynamic  APASP algorithm. The rest of the paper will focus only on obtaining efficient algorithms for APASP.
In Section \ref{sec:DI}  we briefly review the \NPR \decremental APASP algorithm and the fully dynamic \DI algorithm.
 In Section \ref{sec:basic} we describe our basic fully dynamic APASP algorithm \fullydynamicend, its properties and complexity analysis. Pseudocode and correctness are in Section \ref{sec:fd}.
In Section~\ref{sec:improved} we briefly review the \Tho fully dynamic APSP algorithm and then we present our improved algorithm \FFDe. The complexity of \FFD is discussed in Section \ref{sec:proof}, while the pseudocode and correctness are available in the Appendix \ref{sec:algo}. 
The paper ends with a conclusion in Section~\ref{sec:extn}.
\section{Fully Dynamic Betweenness Centrality} \label{sec:fdbc}

The static Brandes algorithm~\cite{Brandes01} computes BC scores in a two phase process. The first phase computes the SP out-dag
for every source through $n$ applications of Dijkstra's algorithm.  The second phase
uses an `accumulation' technique that computes all BC scores  using these SP dags in  $O(n \cdot \nu^*)$ time.

In our fully dynamic algorithm, we will leave the second phase unchanged. In the first phase, we 
implicitly maintain the SP dags with some of the structures maintained by our algorithms.
In contrast, a different approach was previously used in the \incremental BC algorithm in~\cite{NPR14}, where the SP dags were explicitly maintained. 
In the \NPR algorithm the SP dags are also implicitly maintained, but no technique to rebuild them is given in~\cite{NPR14b}.

We now describe a very simple method to construct the SP dags from the following structures.
For every vertex pair $x,y$, consider the following sets $R^*(x,y)$, $L^*(x,y)$:

\vhalf
\noindent
- $R^*(x,y)$ contains all nodes $y'$ such that every shortest path $x \rightsquigarrow y$ in $G$ can be extended with the edge $(y,y')$ to generate another shortest path $x \rightsquigarrow y \rightarrow y'$. \\
- $L^*(x,y)$ contains all nodes $x'$ such that every shortest path $x \rightsquigarrow y$ in $G$ can be extended with the edge $(x',x)$ to generate another shortest path $x' \rightarrow x \rightsquigarrow y$.\\

\noindent
These sets were introduced in \DI and allow us to construct the SP dag for each source $s$ using the following simple algorithm \BDAGe.

\vspace{-0.1in}
\begin{algorithm}
  \SetAlgoLined
  %\KwData{}
  %\KwResult{}
		\For {each $t \in V$}{
		\For {each $u \in R^*(s,t)$}{
			{\bf if} {$D(s,t)+\weight(t,u) = D(s,u)$} {\bf then} add the edge $(t,u)$ to dag($s$)\; \label{buildDAG:check}
	}
	}
	\caption{\BDAGe($G, s, \weight, D)$ \scriptsize ($\weight$ is the weight function; $D$ is the distance matrix)}
	\label{algo:buildDAG}
\end{algorithm}
\vspace{-0.1in}

Our fully dynamic algorithms will maintain the $R^*$ and $L^*$ sets. 
More precisely, in our algorithms $R^*$ and $L^*$ will be supersets of the exact collections of nodes defined above, but the check in 
Step~\ref{buildDAG:check} will ensure that only the correct SP dag edges are included. The combined sizes of these $R^*$ and
$L^*$ sets is $\tilde{O}(n \cdot \nu^*)$ in our algorithms, hence the
amortized time bound for the overall fully dynamic BC algorithm is dominated by the time bound for computing fully dynamic APASP.
In the rest of this paper, we will present our fully dynamic APASP algorithms.
\section{Background}\label{sec:DI}

 Our fully dynamic APASP algorithms build on the approach used in the elegant fully dynamic APSP algorithm
 of Demetresu and Italiano~\cite{DI04} for unique shortest paths (the `\DIe' method). 
 Initially, the \DI method solves the \decremental APSP problem. Then, the authors extend it to a fully dynamic algorithm.
 For APASP, in~\cite{NPR14b} the authors present an \decremental algorithm (the `\NPRe' method). 
 However, its extension to a fully dynamic algorithm for APASP presents several challenges that we address in our results.
 We now  briefly review these two methods (\DI and \NPRe), and then we
 give an overview of our methods for fully dynamic APASP. 

\subsection{The NPRdec Increase-Only APASP Algorithm}
 \label{sec:NPR}
  In the APASP problem, we need to maintain all shortest paths,
  and $G$ can have an exponential (in $n$) number of SPs.
  An \decremental algorithm for maintaining APASP was introduced in \cite{NPR14b}
  where the authors use a compact \system to handle the above case, which is described below (we briefly review it, referring the reader to~\cite{NPR14b} for more details).
 Let $d(x,y)$ denote the shortest path length from $x$ to $y$.

 A {\it tuple}, $\tau = (xa, by)$, represents a set of paths in $G$, all with the same weight, and all
 of which  use
 the same first edge $(x,a)$ and the same last edge $(b,y)$.
 In addition, if $d(x, y) = \weight(x, a) + d(a, b) + \weight(b, y)$, then 
 $\tau$ is a {\it shortest path tuple (ST)}.
 
 A path $p$
 in $G$ is a  {\it locally shortest path (LSP)} if the path $p'$ obtained by removing the first edge from $p$ and the path $p''$ obtained
 by removing the last edge from $p$ are both SPs in $G$.
 If the paths in $\tau$ are LSPs, then $\tau$ is
 an LST (locally shortest tuple), and the
 weight of every path
 in  $\tau$ is $\weight(x, a)$ + $d(a, b) + \weight(b,y)$. 
 The key concept of LSP was introduced in the \DI method~\cite{DI04} but such method is not feasible for APASP since it maintains
 each shortest path separately.  
 
 A {\it triple} $\gamma=(\tau, wt, count)$, represents the tuple $\tau=(xa,by)$ that contains
 $count$ number of paths from $x$ to $y$, each with weight
 $wt$. 
 We use triples to succinctly store 
 all LSPs and SPs
 for each vertex pair
 in $G$. For $x, y \in V$,
 we define:
 \begin{eqnarray*}
 	P(x,y) &=& \{\mbox{$((xa, by), wt, count)$: $(xa, by)$ is an LST from $x$ to $y$ in $G$} \} \\
 	P^{*}(x,y) &=& \{\mbox{$((xa, by), wt, count)$: $(xa, by)$ is an ST from $x$ to $y$ in $G$} \}.
 \end{eqnarray*}

 \noindent
 {\bf Left Tuple and Right Tuple.}
 A left tuple (or $\ell$-tuple), 
 $\tau_{\ell} = (xa, y)$,
 represents the
 set of LSPs from $x$ to $y$, all of which use 
 the same first edge $(x,a)$.
 A right tuple ($r$-tuple) $\tau_r = (x, by)$ is defined analogously.
 For a shortest path $r$-tuple $\tau_r = (x, by)$,  $L({\tau_r})$ is
 the set of vertices which can be used as pre-extensions to create  LSTs
 in~$G$, and
 for a shortest path $\ell$-tuple $\tau_{\ell} = (xa, y)$,
 $R(\tau_{\ell})$ is the set of vertices which can be used as post-extensions to
 create LSTs in $G$.  Hence:
 \begin{eqnarray*}
 	L(x, by) &=& \{x': \mbox {$(x', x) \in E(G)$ and $(x'x, by)$ is an LST in $G$}\} \\
 	R(xa, y) &=& \{y': \mbox {$(y, y') \in E(G)$ and $(xa, yy')$ is an LST in $G$}\}.
 \end{eqnarray*}
 
 \begin{figure}
 	\centering
 	\subfigure[tuple  $\tau = (xa,by)$]{
 		\makebox[.3\textwidth]{
 			\begin{tikzpicture}[every node/.style={circle, draw, inner sep=0pt, minimum width=5pt}]
 			\node (x)[label=above:$x$] at (0,0)  {};
 			\node (a)[label=left:$a$] at (0,-0.8) {};
 			\node (b)[label=below left:$b$] at (0,-2.2) {};
 			\node (y)[label=below:$y$] at (0,-3) {}; 
 			\draw[->] (x) -- (a);
 			\path[->,decoration={snake}] { (a) edge[decorate] (b)};
 			\path[->,decoration={snake}] { (a) edge[bend right=60,decorate] (b.west)};
 			\path[->,decoration={snake}] { (a) edge[bend left=60,decorate] (b.east)};
 			\draw[->] (b) -- (y);
 			\end{tikzpicture}} \label{fig:tau}} 
 	\subfigure[$\ell$-tuple  $\tau_{\ell} = (xa,y)$]{
 		\makebox[.3\textwidth]{
 			\begin{tikzpicture}[every node/.style={circle, draw, inner sep=0pt, minimum width=5pt}]
 			\node (x)[label=above:$x$] at (0,0)  {};
 			\node (a)[label=left:$a$] at (0,-1) {};
 			\node (y)[label=below:$y$] at (0,-3) {}; 
 			\draw[->] (x) -- (a);
 			\path[->,decoration={snake}] { (a) edge[decorate] (y)};
 			\path[->,decoration={snake}] { (a) edge[bend right=60,decorate] (y.west)};
 			\path[->,decoration={snake}] { (a) edge[bend left=60,decorate] (y.east)};
 			\end{tikzpicture}} \label{fig:taul}} 
 	\subfigure[$r$-tuple  $\tau_r = (x,by)$]{
 		\makebox[.3\textwidth]{
 			\begin{tikzpicture}[every node/.style={circle, draw, inner sep=0pt, minimum width=5pt}]
 			\node (x)[label=above:$x$] at (0,0)  {};
 			\node (b)[label=below left:$b$] at (0,-2) {};
 			\node (y)[label=below:$y$] at (0,-3) {}; 
 			\draw[->] (b) -- (y);
 			\path[->,decoration={snake}] { (x) edge[decorate] (b)};
 			\path[->,decoration={snake}] { (x) edge[bend right=60,decorate] (b.west)};
 			\path[->,decoration={snake}] { (x) edge[bend left=60,decorate] (b.east)};
 			\end{tikzpicture}} \label{fig:taur}}
 	\caption{Tuples (\cite{NPR14b} as for Fig.\ref{fig:illust-example2} and Fig.\ref{fig:example-table2})}
 \end{figure}			
 
 \noindent
 For $x, y \in V$, $L^{*}(x, y)$ denotes the set of vertices which can 
 be used as pre-extensions to create shortest path tuples in $G$;
 $R^{*}(x, y)$ is defined symmetrically:
 \begin{eqnarray*}
 	L^{*}(x, y) &=& \{x': \mbox {$(x', x) \in E(G)$ and $(x'x, y)$ is a $\ell$-tuple representing SPs in $G$}\} \\
 	R^{*}(x, y) &=& \{y': \mbox {$(y, y') \in E(G)$ and $(x, yy')$ is an $r$-tuple representing SPs in $G$}\}.
 \end{eqnarray*}
 
  \begin{figure}[ht]
  	\begin{minipage}{0.35\linewidth}
  		\centering
  		\begin{tikzpicture}[every node/.style={circle, draw, inner sep=0pt, minimum width=5pt}]
  		\node (x1)[label=above:$x'$] at (0,1)  {};
  		\node (x)[label=left:$x$] at (0,0.2)  {};
  		\node (a1)[label=left:$a_1$] at (-1,-0.6) {};
  		\node (a2)[label=right:$a_2$] at (0,-0.6) {};
  		\node (a3)[label=right:$a_3$] at (1,-0.6) {};
  		\node (v)[label=right:$v$] at (0,-1.4)  {};
  		\node (v1)[label=left:$v_1$] at (-1.5,-1.4)  {};
  		\node (v2)[label=right:$v_2$] at (1.5,-1.4)  {};
  		\node (b1)[label=below left:$b_1$] at (-0.5,-2.2) {};
  		\node (b2)[label=below right:$b$] at (0.5,-2.2) {};
  		\node (y)[label=below:$y$] at (0.5,-3) {};
  		\node (y1)[label=below:$y_1$] at (-0.5,-3) {};
  		\path[every node/.style={font=\sffamily\small}]
  		(a1) edge node [left] {\textbf{\scriptsize{2}}} (v1)
  		(a1) edge node [right] {\textbf{\scriptsize{4}}} (b1)
  		(v1) edge node [left] {\textbf{\scriptsize{2}}} (b1)
  		(a1) edge node [right] {\textbf{\scriptsize{10}}} (v)
  		(a2) edge node [right] {\textbf{\scriptsize{5}}} (v);
  		\draw[->] (x1) -- (x);
  		\draw[->] (x) -- (a1);
  		\draw[->] (x) -- (a2);
  		\draw[->] (x) -- (a3);
  		\draw[->] (a1) -- (v1) ;
  		\draw[->] (a1) -- (b1) ;
  		\draw[->] (a1) -- (v);
  		\draw[->] (a2) -- (v);
  		\draw[->] (a2) -- (v2);
  		\draw[->] (a3) -- (v2);
  		\draw[->] (v1) -- (b1);
  		\draw[->] (v) -- (b1);
  		\draw[->] (v) -- (b2);
  		\draw[->] (v2) -- (b2);
  		\draw[->] (b1) -- (y1);
  		\draw[->] (b2) -- (y);
  		\end{tikzpicture}
  		\caption{Graph $G'$}
  		\label{fig:illust-example2}
  	\end{minipage}
  	\begin{minipage}{0.6\linewidth}
  		
  		\centering
  		\renewcommand{\arraystretch}{1.3}
  		\begin{tabular}{|c| c|}
  			\hline
  			{Set}  &  {$G'$ (with $\weight(a_1, v) = 10$, $\weight (a_2, v) = 5$)} \\
  			\hline \hline
  			$P(x,y)$  & $\{((xa_2,by),4,1), ((xa_3,by),4,1)\}$ \\
  			$= P^*(x,y)$& \\
  			\hline
  			$P(x,b_1)$ &  $ \{((xa_1,v_1b_1),5,1),  ((xa_2,vb_1),7,1),$\\
  			&  $((xa_1,a_1b_1),5,1) \}$ \\
  			\hline
  			$P^*(x,b_1)$ &  $\{((xa_1,v_1b_1),5,1),((xa_1,a_1b_1),5,1)\}$ \\
  			\hline
  			$L^*(v,y_1)$ &   $\{a_2\}$ \\
  			\hline
  			$L(v,b_1y_1)$ &  $\{a_2\}$\\
  			\hline
  			$R^*(x,v)$ &   $\emptyset$ \\
  			\hline
  			$R(xa_2,v)$ &  $ \{b_1\}$ \\
  			\hline
  		\end{tabular}
  		\vspace{-0.03in}
  		\caption{A subset of the tuple-system for  $G'$}
  		\label{fig:example-table2}
  	\end{minipage}
  	\vspace{-0.2in}
  \end{figure}
 
 \noindent
 {\bf Data Structures.}
 The algorithm in~\cite{NPR14b} uses priority queues for $P$, $P^*$,
 and balanced search trees for $L^*$, $L$, $R^*$ and $R$, as well as for a set
 \MT that is specific only to one update. It also uses priority queues $H_c$ and $H_f$ for 
 the cleanup and fixup procedures, respectively.
 
 \noindent
 Some key differences between this representation and the one in \DI~\cite{DI04} are described in~\cite{NPR14b}.
 
 \begin{lemma}
 	\label{lem:bound-tuples-thru-v}
 	\cite{NPR14b}
 	Let $G=(V,E)$ be a directed graph with positive edge weights. 
 	The number of LSTs (or triples) that contain a vertex $v$ in $G$
 	is $O({\vstar}^2)$, and the total number of
 	LSTs (or triples) in $G$ is bounded by $O(m^* \cdot \vstar)$.
 \end{lemma}
 
 The \NPR algorithm maintains
 all STs and LSTs in the current graph,
 and for each 
 tuple,
 it maintains the $L$, $R$,  $L^*$ and $R^*$ sets. 
 To execute a new update to a vertex $v$, \NPR (similar to \DIe) first calls an
 algorithm \underline{cleanup} on $v$ which removes all STs and LSTs that contain $v$. This is followed by  a call to algorithm
 \underline{fixup} on $v$ which computes all STs and LSTs in the updated graph that are not already present in the system.
 The overall algorithm \underline{update} consists of cleanup followed by fixup. If the updates are all \decremental then \NPR
 maintains exactly all the SPs and LSPs in the graph in $O(\vstar^2 \cdot \log n)$ 
 amortized time per update. 
 Several challenges to adapting the techniques in
 the \DI \decremental method to the \system  are addressed 
 in~\cite{NPR14b}. The analysis of the amortized time bound is also more involved since with multiple
 shortest paths it is possible for the dynamic APASP  algorithm to examine a tuple and merely change its
 count; in such a case, the \DI proof method of charging the cost of the examination to the new path added to  or removed from the
 system does not apply. 
 
 \subsection{The \DI Fully Dynamic APSP Algorithm}\label{sec:di-fd}
 The \DI method first gives an \decremental APSP algorithm, and shows that this 
 is also a correct, though inefficient, fully dynamic APSP algorithm.
 The inefficiency arises because under \incremental updates the method may
 maintain some old SPs and their combinations  that are not currently SPs or LSPs; such paths are 
 called historical shortest paths
 (HPs) and locally historical paths (LHPs).
 To obtain
 an efficient fully dynamic algorithm, the \DI method introduces `dummy updates' into the update sequence. A
 dummy update performs cleanup and fixup on a vertex that was updated in the past. 
 Using a strategically chosen sequence of dummy updates, it is established in~\cite{DI04} that the resulting APSP algorithm
 runs in amortized time $O(n^2 \cdot \log^3 n)$ per real update. 
 The \DI method continues to use the notation $P^*$, $L^*$, etc., even though these 
 are supersets of the defined sets in a fully dynamic setting. We will do the same
 in our fully dynamic algorithms.
\section{Basic \fullydynamic APASP Algorithm} \label{sec:basic}

\subsection{Algorithm \fullyupdate for APASP}\label{sec:alg-basic}

We first extend the notions of historical and locally historical paths~\cite{DI04} to tuples and triples.
In the following definition, we will consider a tuple $\tau$ over an interval of time $[t',t]$.
If $\tau$ becomes a shortest tuple, after its creation at time $t'$, it will remain in the dataset as a historical tuple (HT, see Definition \ref{def:ht}) until it is completely removed from the \systemend.

Note that tuples are used instead of triples in the definition below.
This distinction is relevant in our algorithm because during cleanup or fixup, a (historical) triple could change its count without losing its property of being an historical tuple.
A similar behavior could also affect a locally historical tuple.
With our definition, we can immediately refer to it without specifying the count of its associated triple.

%\begin{definition} [HT, THT, LHT, and TLHT]\label{def:ht}
%Let $\tau$ be a tuple in the \system at time $t$. Let $ t' \le t$ denote the most recent step at which
% a vertex on a path in $\tau$ was updated. Then $\tau$  is  a {\it  historical tuple (HT) at time $t$} if
%$\tau$ was an ST at least once in the interval $[t', t]$; $\tau$ is a
%{\it true HT (THT) at time $t$} if it is not an ST 
%in the current graph.
%A tuple $\tau$ is a {\it  locally historical tuple (LHT) at time $t$}
%if either it only
%contains a single vertex or every proper sub-path in it is an HT
%at time $t$; a tuple $\tau$ is a
%{\it true LHT (TLHT) at time $t$} if it is not an LST in the current graph.
%\end{definition}

\begin{definition} [HT, THT, LHT, and TLHT]\label{def:ht}
	Let $\tau$ be a tuple in the \system at time $t$. Let $ t' \le t$ denote the time at which $\tau$ was originally added for the first time in the \systemend.
	Then $\tau$  is  a {\it  historical tuple (HT) at time $t$} if
	$\tau$ was an ST at least once in the interval $[t', t]$; $\tau$ is a
	{\it true HT (THT) at time $t$} if it is not an ST 
	in the current graph.
	A tuple $\tau$ is a {\it  locally historical tuple (LHT) at time $t$}
	if either it only
	contains a single vertex or every proper sub-path in it is an HT
	at time $t$; a tuple $\tau$ is a
	{\it true LHT (TLHT) at time $t$} if it is not an LST in the current graph.
\end{definition}

The above definition is used extensively in our proof of correctness (see Appendix \ref{append:corr}, lemmas \ref{lem:fcleanup-corr2} and \ref{proof:fflem}). In particular, in the cleanup loop inveriant (Lemma \ref{lem:fcleanup-corr2}), we show how each TLHT (and THT) is representing only paths that avoids the updated node $v$. Moreover, in the fixup loop invariant (Lemma \ref{proof:fflem}), we require that each HT (LHT) maintains the correct count of HPs (LHPs) in the graph after the fixup phase.  

\noindent
%newstart
In order to correctly extend $\ell$-tuple and $r$-tuple under fully-dynamic updates, we extend the following data structures from \NPRe.

\begin{eqnarray*}
	L(x, by) &=& \{(x',wt'): \mbox {$(x', x) \in E(G)$ and $(x'x, by)$ is an LHT of weight $wt'$ in $G$}\} \\
	R(xa, y) &=& \{(y',wt'): \mbox {$(y, y') \in E(G)$ and $(xa, yy')$ is an LHT of weight $wt'$ in $G$}\} \\
	L^{*}(x, y) &=& 
	\{ (x',wt'): (x', x) \in E(G) \textrm{ and } (x'x, y) \textrm{ is an $\ell$-tuple of weight $wt'$} \\ 
	& & \textrm{ representing HPs in $G$} \} \\
	R^{*}(x, y) &=& \{(y',wt'): (y, y') \in E(G) \textrm{ and } (x, yy') \textrm{ is an $r$-tuple of weight $wt'$ } \\
	& & \textrm{ representing HPs in $G$} \} \\
\end{eqnarray*}

\noindent
Note that $L^*$, $R^*$, $L$ and $R$ are now kept in a stack order, thus giving priority to the nodes that are inserted more recently. 
Another important data structure introduced for the cleanup phase is the \emph{combined history} $\CH(\gamma)$ of a triple $\gamma$ (see the dedicated paragraph near the end of this section).
Additional data structures will be suggested in the description of the algorithm (see Section \ref{sec:fd}), but they are only implementation oriented. 

Similar to the \DI algorithm, mentioned in Section \ref{sec:di-fd}, we need a technique to clean our data structures from old unused tuples.
However, if we try to apply the \DI dummy sequence to APASP, we are faced with the issue that a new ST for $x,y$ (with the same weight)
could be created at each update in a long sequence of successive updates. Then, a \incremental 
update could transform all of these STs into HTs. If this happens,  then 
several HTs for
$x,y$, all with the same weight, could have the same dummy-index
(in \DI only one HP can be present for this entire collection due to unique SPs).
Thus, the \DI approach of obtaining an $O(\log n)$ bound for the number of HPs for each
vertex pair does not work for HTs in our \systemend.

Our new algorithm uses a different dummy sequence, and a completely different
analysis that obtains an $O(\log n)$ bound for the number of different `PDGs' (a PDG
is a type of
derived graph defined in Section~\ref{sec:alg-main}) that can contain the HTs. Our
new dummy sequence
is inspired by the `level graph' method introduced in \Tho~\cite{Thorup04} (see Section \ref{sec:improved}, where we generalize this approach to multiple shortest paths in our \FFD algorithm) to
improve the 
amortized bound for fully dynamic APSP to $O(n^2 \cdot  \log^2 n)$, saving a log factor over \DIe. 
This
method is complex because it maintains $O(\log n)$ levels of data structures for
suitable `level graphs'. 
Our algorithm \fullydynamic does not maintain these
level graphs. Instead, \fullydynamic  performs exactly like the fully dynamic
algorithm in \DIe, except that it uses this alternate dummy update sequence, and it
calls \fullyupdate for APASP instead of the \DI update algorithm for APSP.
Our change in the update sequence requires a completely new proof of the amortized
bound which we present in Section~\ref{sec:alg-main}. 
We consider this to be a contribution of independent interest: If we replace 
\fullyupdate by the \DI update algorithm in \fullydynamicend, we get a new fully dynamic APSP
algorithm which is as simple as \DIe, with a new analysis.

\fullyupdate is given in the two-line Algorithm~\ref{algo:fupdate}. It
calls \fullycleanup and \fullyfixup in sequence, on the updated vertex $v$ (see Section \ref{sec:fd} for the complete pseudocode).
This is similar to the update
procedure in the  \decremental and
fully dynamic algorithms for unique paths in \DI~\cite{DI04} and in the \decremental algorithm 
in~\cite{NPR14b}. 

%Ms
The cleanup procedure removes from the \system all HPs and LHPs containing $v$ by decrementing the count of triples which represent them. 
Thus, each TLHT (and THT) containing the updated vertex $v$ is updated in the \system with its correct count even if some of its represented paths avoid $v$.
To achieve this property for THTs, we use the new data structure $\CH$ during the cleanup phase.
\fullycleanup works by repeatedly extracting triples from a heap $H_c$, generated as extensions of the updated node $v$ using the data structures of the \systemend.

The fixup phase adds to the \system a superset of LSPs generated in the graph by the update: if a new LSP discovered during the fixup phase is of the form $x \rightarrow a \rightsquigarrow b \rightarrow y$ and weight $wt$, \fullyfixup will increase the count of the tuple $\tau = ((xa,by),wt)$ by one (creating $\tau$ itself if not in the \systemend). Also in this case, triples are extracted from a heap $H_f$, added to the system, and finally extended to candidate triples to be processed in future steps.  

\begin{algorithm}[ht]
	\SetAlgoLined
	%\KwData{}
	%\KwResult{}
	\fullycleanupend$(v)$\;
	\fullyfixupend$(v,\weight')$\;
	set update-num($v$);
	\caption{\fullyupdateend$(v, \weight')$}
	\label{algo:fupdate}
\end{algorithm}

\noindent
In \fullyupdateend, the sets $P^*(x,y)$ and $P(x,y)$ will contain HTs (including all STs) and LHTs
(including all LSTs), respectively,  from $x$ to $y$.
 It was observed in \cite{DI04} that the \decremental algorithm they presented for the 
 unique SP case is a correct
algorithm when \incremental updates are interleaved with \decremental
ones. 
However here, in contrast to \DIe, the \decremental  APASP algorithm in~\cite{NPR14b} 
(the \NPR algorithm)
 needs to be refined before it
becomes correct for a fully dynamic
sequence, and this is due to the presence of multiple shortest paths.
In Section~\ref{sec:alg-basic}
we describe algorithm \fullyupdate which is a correct fully dynamic APASP algorithm that
 maintains  a superset of
all STs and LSTs in the
current graph. 
This algorithm is very similar to \NPR but contains several changes to ensure correctness under fully dynamic updates.
%In fact, additional features are required for the \NPR data structures to ensure correctness when we deal with fully-dynamic updates.
However, \fullyupdate is not very efficient since it may add, remove, and examine a large number 
of tuples. This is similar to
 the \DI \decremental algorithm when used as a fully dynamic algorithm without any additional features.

We first provide an example where \NPR is not correct as an APASP algorithm.
Consider, for instance, a THT
 $\tau = (xa,by)$ with weight $wt$ which is currently an LST. Using the \NPR 
 algorithm~\cite{NPR14b}, since $\tau$ is a THT it will be present
 in $P^*(x,y)$ (but not as an ST). Suppose a new
set of paths represented by a 
 new triple $\tau'= ((xa,by),wt,count')$, with
 the same weight $wt$, is added to the count of this tuple $\tau$. We cannot simply add $count'$ to $\tau$ in $P^*(x,y)$ because its extensions were performed using the old count, and
 if $\tau$ is restored as an ST in $P^*(x,y)$, these extensions will not have the correct count. (With unique SPs
 this situation can never occur.)
 
 Our solution for the above case
 is to have $\tau$ in $P(x,y)$ with the larger correct count (thus including $count'$), and 
 to leave the corresponding $\tau$ in
 $P^*(x,y)$ with its original count. Should $\tau$ later be restored as an ST 
 then the difference in counts between $\tau$ in $P^* (x,y)$ and the corresponding
 $\tau$ in $P(x,y)$ will trigger left and right extensions of $\tau$ with the correct count
 even though $\tau$ is currently  in $P^*(x,y)$ (see point [\ref{fact:fbeta}] in Section \ref{sec:fuan}).
 
 There remains another serious scenario that could occur. During \fullycleanup starting from the current updated vertex
  $v$, we may reach a 
   triple $\gamma = (\tau, wt, count)$ in $P$ through say, a left extension, while the triple in $P^*$ for $\tau$
   with weight $wt$ is $\gamma' = (\tau, wt, count')$, with $count' < count$ (the extreme case being that
   $count'=0$, in which case there is no $\gamma'$ in $P^*$).
   This is an indication that the paths in $\gamma - \gamma'$ were formed after $\gamma'$ became a THT.
	%Ms
   Moreover, the number of paths going through $v$ in $\gamma' \in P^*$ could be different from the number of paths going through $v$ in $\gamma \in P$, posing a challenging dilemma on the correct number of paths to be removed from $\gamma' \in P^*$.   
   We address this situation in \fullycleanup by using the $\CH(\gamma)$ data structure for the triple $\gamma$ in cleanup.
   
   \subsubsection{The \texorpdfstring{$\CH$}{CH} Data Structure} \label{sec:ch}
   The \emph{combined history} $\CH(\gamma)$ of a triple $\gamma$ is an array of updates, and for each update $t$ contains the number of paths that joined $\gamma$ during the $t$-th update. 
   This new data structure is designed to surgically remove paths from LHT and THT. In fact, it could be the case where some paths exist only as LHT while they are not present in the corresponding THT located in $P^*$. The example in Appendix \ref{append:example} shows how completely removing THT from $P^*$ (instead of decrementing them with the help of $\CH$) increases the complexity of the algorithm to a $\Theta(n^3)$ bound, and produces incorrect counts for the tuples.

   To efficiently implement $\CH$ we also maintain the update number associated to each node. The structure $\CH$ is initialized with update-num($v$) that is the last update in which $v$ was updated \emph{before} this cleanup phase, associated with value 1 and representing the trivial tuple $(vv,vv)$.
   The $\CH$ data structure is associated, exclusively during the cleanup phase, with the tiples built by the algorithm that will be deleted from the graph.
   More specifically, a THT $\gamma' = ((xa,by),wt, count')$ is decremented during cleanup only by $count$ paths, where $count$ is the number of paths, going through the updated node $v$, generated during the updates at times $t \leq \textrm{update-num}(\gamma')$.
   To perform this step efficiently, we can check the last time $t$ when $\gamma'$ was updated, and remove from it only the paths in $\CH(\gamma)$ that were created before or during update $t$. This technique guarantees that only the paths going through $v$ but truly represented by $\gamma'$ are removed from $P^*$. 
   More details on the implementation are given in the next section containing the description of algorithm \fullydynamicend.
   
   %  Hence, 
%  when we encounter a THT or a TLHT during cleanup, we completely remove the triple from the \system (see point [\ref{fact:crem}] in Section \ref{sec:fuan}).
 
 \subsection{Pseudocode of \fullydynamic Algorithm} \label{sec:fd}
Here we give the full pseudocode for \fullycleanup (Algorithm \ref{algo:fcleanup}) and \fullyfixup (Algorithm \ref{algo:ffixup}). They are similar to the corresponding pseudocode
cleanup and fixup in~\cite{NPR14b}, and we have marked the steps changed from~\cite{NPR14b}
with a $\bullet$ at the end of the line. 
A description of the new features of the algorithms is given below, while a detailed description is available in \cite{NPR14b}.
The other steps in  Algorithm~\ref{algo:ffixup} are described in~\cite{NPR14b}, as are the
two parameters paths$(\gamma, v)$, which
gives  the number of paths  containing the node $v$
that are represented by the triple $\gamma$, and update-num$(\gamma)$, which is a timestamp that indicates the last update in which the triple $\gamma$ is involved.
Then, the correctness of the algorithms is established.

%\fullyupdate is given in the two-line Algorithm~\ref{algo:fupdate}. It
%calls Algorithm~\ref{algo:fcleanup} (\fullycleanupend) and 
%Algorithm~\ref{algo:ffixup} (\fullyfixupend), in sequence, on the updated vertex $v$. We 
%present pseudocode for these two procedures here. They are similar to the corresponding pseudocode
%cleanup and fixup in~\cite{NPR14b}, and we have marked the steps changed from~\cite{NPR14b}
%with a $\bullet$ at the end of the line.
%Also, a triple $\gamma$ is now inserted in $P$ with a key $[wt,\cb]$, instead of just $wt$. Here
%$\cb$ is a control bit that is set during the fixup phase and indicates if $\gamma$ is present in $P^*$ with the correct count ($\cb=1$), or if $\gamma$ has the correct count only in $P$ ($\cb=0$). In the latter case $\gamma$ could be also present in $P^*$ but with a wrong count.

The only changes
from~\cite{NPR14b} in Algorithm~\ref{algo:fcleanup}
are in steps \ref{process-cleanup:left-extend-unmarked}, and \ref{process-cleanup:new3} where we decrement any THT we encounter during cleanup, while extending from the updated node $v$, using the new data structure $\CH$. 
More specifically, in step \ref{process-cleanup:left-extend-unmarked}, we use the weight associated with the extension to avoid the creation of cycles. An implementation to achieve this in \fullycleanup is the following: for a left extension $x'$, let $ldc$ be the set of paths of the form $(x'x, \times b)$ and weight $w' - \weight(b,y)$ removed from $P^*$ during this cleanup phase (note that these paths can pe maintained in a temporary data structure of size $O(mn)$). Similarly, let $rdc$ be a the set of paths of the form $(x \times, by)$ and weight $w' - \weight(x',x)$ removed from $P^*$ during this cleanup. If the updated node $v \neq y$, we set $fcount = \max(ldc,rdc)$, alias we only count all the locally historical paths of $(x'x, by)$ that are formed by a valid $l$-tuple and $r$-tuple both in $P^*$ at the beginning of the cleanup phase and now removed. Otherwise, if $v=y$ we do not overwrite $fcount$. Note that, since no cycle path is ever created during the fixup phase and placed in $P^*$, with the above implementation, \fullycleanup does not delete any real path by mistake.
In step \ref{process-cleanup:left-extend-unmarked}, we deal with the delicate task of correctly decrementing an HT from the tuple-system. As discussed in Section \ref{sec:alg-basic}, when we deal with a THT $\gamma''$  in $P^*$ we could only delete a subset of paths identified by the cleanup algorithm for its respecting LHT $\gamma$ in $P$ (note that if $\gamma''$ is a ST then we trivially subtract the same number $fcount$ of paths substracted from $\gamma$, since they have the some count). To address this issue, we use the new $\CH(\gamma')$ data structure (in the algorithm $\gamma'$ is the triple to be deleted from the \systemend): let $t=\textrm{update-num}(\gamma'')$ be the most recent update in which $\gamma''$ was updated in $P^*$; in $\CH(\gamma')$ we sum all the paths up to time $t$, we call this value $fcount''$. Finally, we decrement $fcount''$ paths from $\gamma'' \in P^*$. Now, we only need to address how to efficiently build $\CH$ for a given triple $\gamma'$. Let us assume that $\gamma'$ is a left extension to node $x'$ for a set of $k$ triples of the form $x\times,by$. Each one of the non-extended triples has its own $\CH_j$ data structure built from a previous cleanup iteration, with $j \leq k$. Then, for a given $\CH_j$, for each pair $(updnum, count)$  in $\CH_j$ we add the pair $(\max(\textrm{update-num}(x'), updnum), count)$ to $\CH(\gamma')$. Note that, if such pair is already present in $\CH(\gamma')$ with $count'$, we increment it by $count$. The complexity analysis for this implementation is given in Corollary \ref{lem:cor1}.

Algorithm~\ref{algo:ffixup} introduces new features to achieve correctness and efficiency.
We observe that we may 
revert an HT from, say, $x$ to $y$, back to an ST during an update, and this happens only if  the shortest path
distance from  $x$ to $y$ increases. This condition translates into the new check in Step \ref{ffixup:main1} of
Algorithm~\ref{algo:ffixup}. Here  we proceed as in 
\cite{NPR14b}
keeping in mind that an LHT extracted from $P$ as an ST (Step \ref{fixup:phase3-addfromP-begin}, Algorithm~\ref{algo:ffixup})
may or may not be in $P^*$. If the LHT is  not in
$P^*$ (Step~\ref{ffixup:old1}, Algorithm~\ref{algo:ffixup}) we add the triple to the \system as in 
\cite{NPR14b}.
If the LHT is already in $P^*$ with a different count (Step~\ref{ffixup:new1}, Algorithm~\ref{algo:ffixup}), we replace the count of the triple in $P^*$ with $count'$ from the triple in $P$ and we add the triple to $S$ (Step~\ref{ffixup:new2}, Algorithm~\ref{algo:ffixup}). 
%Note that the extension sets $L^*$ and $R^*$ do not need to be updated in this case.
In step \ref{fixup:phase3-addfromP-begin}, Algorithm~\ref{algo:ffixup} we check the bit $\cb$ associated with the triple $\gamma$. Since $P$ is a priority queue, we will process only the triples in $P$ with min-key $[wt,0]$, so we
avoid examining the triples that are already in $P^*$ with a correct count.
We set $\cb$ to 1 for any triple added to or updated in $P^*$ with the correct count
(Steps \ref{ffixup:b1} and \ref{ffixup:b2}, Algorithm~\ref{algo:ffixup}). Also, for an LHT updated in $P$ and not $P^*$, 
we set $\cb=0$ (Step \ref{ffixup:b3}, Algorithm~\ref{algo:ffixup}).
Finally, we use the $L^*$ and $R^*$ stacks to generate only LHT that are not cycles. We pop the extensions from the stack in reversed order of update time. This is because the newest extensions always refers to STs, while older extensions may refer to THTs. Whenever, we encounter an extension that does not satisfy the distance check at step \ref{ffixup:new}, we know that the extensions left in the stack are associated to tuples that are even older than the current one; thus we can skip them.
A possible implementation to avoid cycles is the following: for each triple $\gamma' = (x'x,by)$ that \fullyfixup generates with an extension $(x',wt')$ to the left, we check if $P^*(x',b)$ contains a triple of the form $(x'x, \times b)$ and weight $w < wt'$. If this is the case, we do not extend to $x'$ because it will create a cycle. A symmetric check is applied to the right extensions.  

\begin{algorithm}[ht]
	\SetAlgoLined
	\SetAlgoNoEnd
	%\KwData{}
	%\KwResult{}
	$H_c \leftarrow \emptyset$; \MT $\leftarrow \emptyset$\; \label{cleanup:initHcMT}
	$\gamma \leftarrow ((vv,vv), 0, 1)$; add $\gamma$ to $H_c$\; \label{cleanup:init}
	\While {$H_c \neq \emptyset$ }{ \label{cleanup:while}
		extract in $S$ all the triples with min-key $[wt,x,y]$ from $H_c$\; \label{cleanup:extract-set}
		\For {every $b$ such that $(x\times,by) \in S$}{ \label{process-cleanup:left-extend-start}
			let $fcount'$ be the number of deleted paths of the form $((xa_i,by),wt)$ \; \label{process-cleanup:left-extend-fcount}
			\For {every $(x',wt') \in L(x,by)$ such that $((x'x,by),wt')$ is an \LHT not in \MTend}{ \label{process-cleanup:left-extend-unmarked}
				$\gamma' \leftarrow ((x'x,by),wt', fcount')$\; 
				add $\gamma'$ to $H_c$\; \label{process-cleanup:create-left}				
%				\eIf {$wt' - \weight(b,y) = d(x',b)$ $\bullet$}{ \label{process-cleanup:new1}
					remove $\gamma'$ in $P(x',y)$ // decrements its count in $P$ by $fcount'$\;  \label{process-cleanup:left-extend-removeP}
					\eIf {a triple for $((x'x, by),wt')$ exists in  $P(x',y)$}{
						insert  $((x'x,by),wt')$ in \MT\; \label{process-cleanup:left-extend-markR}
					}{
					delete $(x',wt')$ from $L(x,by)$ and delete $(y,wt')$ from $R(x'x,b)$\; \label{process-cleanup:left-extend-removeL}
				}
%			}
		{
%			remove the tuple $(x'x,by)$ with weight $wt'$ from $P(x',y$) $\bullet$ \{we completely remove a TLHT\}\; \label{process-cleanup:new2}
%			\If {a triple for $(x'x, by)$ does not exist in  $P(x',y)$}{ \label{process-cleanup:new2c1}
%				delete $x'$ from $L(x,by)$ and delete $y$ from $R(x'x,b)$\;   \label{process-cleanup:new2c2}      
%			}
		}
		\If {a triple $\gamma''$ for $((x'x, by),wt')$ exists in  $P^*(x',y)$ with $count$ paths}{  \label{process-cleanup:left-extend-checkSP}
			let $fcount''$ be the number of deleted paths of the form $((x'x, by),wt')$  created during an update smaller (older) than update-num$(\gamma'')$ $\bullet$\ \newline // this step can be efficiently computed using our $CH$ data strcuture; \label{process-cleanup:new3}

			remove $\gamma'$ in $P^*(x',y)$ by decrementing $fcount''$ from $count$\; \label{process-cleanup:new4}
%		}		
		\textbf{if} $P^*(x,y)$ doesn't contain triples of weight $wt$ \textbf{then} delete $(x',wt')$ from $L^{*}(x, y)$\; \label{process-cleanup:left-extend-removeLstar}
		\textbf{if} $P^*(x',b)$ doesn't contain triples of weight $wt' -\weight(b,y)$ \textbf{then} delete $(y,wt')$ from $R^{*}(x',b)$\; \label{process-cleanup:left-extend-removeRstar}
	}
}
} \label{process-cleanup:left-extend-end}   
perform symmetric steps \ref{process-cleanup:left-extend-start} -- \ref{process-cleanup:left-extend-end} for right extensions\;  \label{cleanup:rextend}
}
\caption{\fullycleanupend$(v)$}
\label{algo:fcleanup}
\end{algorithm}

\begin{algorithm}[ht]
	\SetAlgoLined
	\SetAlgoNoEnd
	%\KwData{}
	%\KwResult{}
	$H_f \leftarrow \emptyset$; \MT $\leftarrow \emptyset$ \;
	\For {each edge incident on $v$ } {
		create a triple $\gamma$; set paths$(\gamma, v) = 1$; set update-num$(\gamma)$; add $\gamma$ to $H_f$ and to $P$\; \label{fixup:init1}
	}
	\For {each $x, y \in V$}{
		add a min-key triple from $P(x, y)$ to $H_f$\; \label{fixup:init2}
	}
	\While {$H_f \neq \emptyset$} { \label{fixup:phase3-begin}
		extract in $S'$ all triples with min-key $[wt,x,y]$ from $H_f$; $S \leftarrow \emptyset$\; \label{fixup:phase3-extract1}
		\If {$S'$ is the first extracted set from $H_f$ for $x,y$}{  \label{fixup:phase3-first-ext}
			\eIf { $P^*(x,y)$ increased min-weight after cleanup $\bullet$}{ \label{ffixup:main1}
				\For {each $\gamma' \in P(x,y)$ with min-key $[wt,0]$ $\bullet$}{ \label{fixup:phase3-addfromP-begin}
					let $\gamma' = ((xa', b'y), wt, count')$\;
					\{Next step check if $\gamma'$ is completely missing from $P^*$\}\;
					\uIf {$\gamma'$ is not in $P^*(x,y)$}{ \label{ffixup:old1}                	
						add $\gamma'$ in $P^{*}(x,y)$ and $S$; add $(x,wt)$ to $L^{*}(a',y)$ and $(y,wt)$ to $R^{*}(x,b')$\; \label{ffixup:phase3-add2LRStar2}
						\{Next step check if $\gamma'$ is in $P^*$ with a different count\}\;
					} \uElseIf{$\gamma'$ is in $P(x,y)$ and $P^*(x,y)$ with different counts $\bullet$}{ { \label{ffixup:new1} 
						replace the count of $\gamma'$ in $P^{*}(x,y)$ with $count'$ and add $\gamma'$ to $S$ $\bullet$\; \label{ffixup:new2} 
					}}
					%\Else{}
					set $\cb$ for $\gamma' \in P(x,y)$ to $1$ $\bullet$\; \label{ffixup:b1}
				} \label{fixup:phase3-addfromP-end}
			}{
			\For {each $\gamma' \in S'$ containing a path through $v$} { \label{fixup:phase3-addfromX-begin}
				let $\gamma' = ((xa', b'y), wt, count')$\;
				add $\gamma'$ with paths$(\gamma',v)$ in $P^{*}(x,y)$ and $S$; add $(x,wt)$ to $L^{*}(a',y)$ and $(y,wt)$ to $R^{*}(x,b')$\; \label{fixup:phase3-add2LRStar1}
				set $\cb$ for $\gamma' \in P(x,y)$ to $1$ $\bullet$\; \label{ffixup:b2}
			} \label{fixup:phase3-addfromX-end}
		} \label{ffixup:phase3-main-check-end}
		\For {every $b$ such that $(x \times,by) \in S$}{ \label{fixup:startleft}
			let $fcount'  = \sum_{i} ct_i$ such that $((xa_i, by), wt, ct_i ) \in S$\;
			\For {every $(x',wt')$ in $L^{*}(x,b)$ such that $((x'x,by),wt')$ is an \LHTe}{ \label{ffixup:new}
				\If {$((x'x, by),wt') \notin$ \MT}{
					$wt' \leftarrow wt+\weight(x',x)$; $\gamma' \leftarrow ((x'x,by),wt', fcount')$\;
					set update-num$(\gamma')$; paths$(\gamma',v)\leftarrow \sum_{\gamma=(x\times,by)} \textrm{paths}(\gamma,v)$; add $\gamma'$ to $H_f$\; \label{fixup:add1}
					\eIf {a triple for $((x'x, by),wt')$ exists in $P(x',y)$}{
						add $\gamma'$ with paths$(\gamma', v)$ in $P(x',y)$;  
					}{
					add $\gamma'$ to $P(x',y)$; add $(x',wt')$ to $L(x,by)$ and $(y,wt')$ to $R(x'x,b)$\; 
				}
				set $\cb$ for $\gamma' \in P(x',y)$ to $0$ $\bullet$\; \label{ffixup:b3}
				add $((x'x, by),wt')$ to \MTend\;
			}
		}
	}		\label{fixup:endleft}
	perform symmetric steps \ref{fixup:startleft} -- \ref{fixup:endleft} for right extensions\;
}
} \label{fixup:phase3-end}
\caption{\fullyfixupend$(v, \weight')$}
\label{algo:ffixup}
\end{algorithm}

 Recall that \fullyupdate is simply an execution of \fullycleanup
  followed by \fullyfixupend . The correctness of this algorithm is argued by
  noting that our method ensures that when an ST in $P^*$ is processed during \fullyfixup
  without further extensions, it has the correct weight and count and 
  all of its extensions have been performed with that count; every ST and LST is generated
  starting with singleton edges, min-weight tuples from the $P$ sets, and correct STs  
  from the $P^*$ sets,
  hence the counts of the tuples identified as STs and LSTs are maintained correctly.
  
 \subsection{Analysis and Properties of \fullyupdate} \label{sec:fuan}
 As discussed in the previous sections, both \fullycleanup and \fullyfixup are similar to the corresponding pseudocode
 cleanup and fixup in~\cite{NPR14b}, but they require some important changes to ensure correctness and efficiency.
 In this section, we highlight several new components that will be relevant for proving the properties of \fullyupdate presented in this section.
 
 \paragraph{New Components relative to \NPRe} 
 \begin{enumerate} [label=NC.{\arabic*}]
 	\item \label{fact:crem}
 	%Ms
 	Algorithm \fullycleanup decrements the \emph{count} of THTs using the new data structure $\CH$. More specifically, a THT $\gamma' = ((xa,by),wt, count')$ is decremented only by $count$ paths, where $count$ is the number of paths, going through the updated node $v$, in $\CH(\gamma)$ and generated during the updates at times $t \leq \textrm{update-num}(\gamma')$. The structure $\CH$ is initialized with update-num($v$) (i.e. the last update on $v$ before the current update) associated with the value 1 representing the trivial tuple $(vv,vv)$.
 	\item \label{fact:fbeta}
 	In algorithm \fullyfixupend, a triple $\gamma$ is now inserted in $P$ with a key $[wt,\cb]$, instead of just $wt$. Here
 	$\cb$ is a control bit that is set during the fixup phase and indicates if $\gamma$ is present in $P^*$ with the correct count ($\cb=1$), or if $\gamma$ has the correct count only in $P$ ($\cb=0$). In the latter case $\gamma$ could be also present in $P^*$ but with a wrong count.\\
 	Before processing a triple $\gamma$, we check the bit $\cb$ associated with it. In fact, since $P$ is a priority queue, we will process only the triples in $P$ with min-key $[wt,0]$, thus avoiding the triples that are already in $P^*$ with a correct count.
 \end{enumerate}
 
 Note that, since a non-shortest HT (THT) always has a larger distance than its corresponding ST (by definition), the following observation follows immediately. 
 
 \begin{observation}
  	Reverting a THT from, say $x$ to $y$, back to an ST during an update happens only if the shortest path distance from  $x$ to $y$ increases.
 \end{observation}
 
  We now establish some basic properties of algorithm \fullyupdate based on the high-level
  description we have given above. We start with a general bound on the running time.
  
  \begin{lemma}  \label{lem:basic-amort}
  Consider a sequence of $r$ calls to \fullyupdate
   on a graph with $n$ vertices.  Let $C$ be the maximum number of tuples in the \system 
  that can contain a path through a given vertex, and
  let $D$ be the maximum number of tuples that can be in the \system at any time.
  Then \fullyupdate executes the $r$ updates in 
    $O((r \cdot (n^2 +C ) + D)\cdot \log n)$ time.
  \end{lemma}
  
  \begin{proof}
  	We bound the cost of \fullyupdate by classifying every triple $\gamma$ that is actively processed by \fullyfixupend, as one of the disjoint types below.
  	Note that, if a triple is deleted or decremented its cost will be charged on the \fullycleanup bound. Any other behaviour of a triple (other than the types below) is not relevant for the correctness of our algorithm.
  	
  	\begin{itemize}
  		\item {\bf Type-0 (contains-v):} $\gamma$  represents  at least one path containing vertex $v$.
  		\item {\bf Type-1 (new-LHT):} $\gamma$ was not an LHT before the update but is an LHT after the update,
  		and no path in $\gamma$ contains $v$.
  		\item {\bf Type-2 (new-ST-old-LHT):} 
  		$\gamma$ is an ST after the update,
  		and $\gamma$ was an LHT but not an HT before the update, and no path in $\gamma$ contains $v$.
  		\item {\bf Type-3 (renew-ST):} $\gamma$  was a THT  before the update and it is restored as a ST after the update, and no path in $\gamma$ contains $v$.
  		\item {\bf Type-4 (new-LHT-old-LHT):} $\gamma$  was an LHT  before the update and continues
  		to be an LHT after the update,
  		and no path in $\gamma$ contains $v$.
  	\end{itemize}
  	
  	\noindent
  	The number of Type-0 triples, processed by \fullyfixup is at most $C$.
  	
  	The number of Type-1 triples, processed by \fullyfixup is addressed by amortizing over the entire update sequence as described in the paragraph below.
  	
  	For a \mbox{Type-2} triple processed by \fullyfixupend, we observe that after 
  	such a triple becomes a ST, it is not removed from $P^*$ unless a real or dummy update is triggered on a vertex that lies in it.
  	But in such an update this would be counted as a Type-0 triple. 
  	Further, each such Type-2 triple is examined only a constant number of times \fullyfixupend, because after they are inserted into $P^*$ the bit $\cb$ associated changes to 1 and they will not be processed again by fixup, unless the number of paths they represent is changed (point [\ref{fact:fbeta}]). Hence we charge each access to a Type-2 triple to the step in which it was created as a Type-1 triple.
  	
  	For Type-3 triples, we distinguish two cases: if $\gamma$ didn't change its count in $P$ after it became a THT then its flag is $\cb=1$ and it is present in $P^*$ with the correct count (point [\ref{fact:fbeta}]). Thus the fixup algorithm will not process it. If $\gamma$ changed its count in $P$ while it was a THT then its flag is $\cb=0$ and we can charge the processing of $\gamma$ (if extracted from $H_f$) to the sub-triple $\gamma'$ generated from the updated node that increased the count of $\gamma$; in other words there was an LST $\gamma'$, created in a previous update, whose extensions added to the count of $\gamma$.
  	Observe that, triples in $P^*$ that are not 
  	placed initially in $H_f$ and have $\cb=1$ in $P$ (no additional path was added to that triple) are not examined in any step of fixup (point [\ref{fact:fbeta}]), so no additional Type-3 triples are examined.
  	
  	For Type-4, we note that for any $x, y$ we add exactly one candidate
  	min-key triple from $P(x,y)$ to $H_f$, hence initially there are at most $n^2$ such triples in $H_f$, any of which could be Type-4.
  	Moreover, we never process an old
  	LHT which is not a new ST so no additional Type-4 triples are examined during fixup.
  	Thus the number of triples examined by a call to fixup is
  	$C$ plus $X$, where $X$ is the number of {\em new} triples fixup adds to the tuple system.
  	(This includes an $O(1)$ credit placed on each new LHT for a possible later conversion to an HT.)
  	
  	Let $r$ be the number of updates in the update sequence.  Since triples are removed only in cleanup, at most $O(r \cdot C)$ triples are removed (or decremented) by the cleanups. 
  	There can be at most $D$ triples remaining at the end of the sequence, hence the total number of new triples added by all fixups in the update sequence is
  	$O(r \cdot C + D)$. 
  	Since the time taken to access a triple is $O(\log n)$ due to the data structure operations, and we examine at least $n^2$ triples at each round, the total time spent by fixup over 
  	$r$ updates is $O((r \cdot (n^2 + C) + D)\cdot \log n)$.
  \end{proof}

  In Section~\ref{sec:alg-main}, we will use  the algorithm \fullyupdate within algorithm
  \fullydynamicend , 
   that performs a special
  sequence of `dummy' updates, to obtain a fully dynamic APASP algorithm with
  an $O(\vstar^2 \log^3 n)$ amortized cost per update; we obtain this amortized bound by establishing suitable
  upper bounds on the parameters $C$ and $D$ in Lemma~\ref{lem:basic-amort}
  when algorithm \fullydynamic is used. We conclude this section with
  some lemmas that will be used in the analysis of the amortized time bound of
   Algorithm~\ref{algo:fully-main}. 
  
  \begin{lemma}\label{lem:lhts}
  At each step $t$, the \system for Algorithm \fullyupdate maintains a subset of
HTs and LHTs  that includes all STs and LSTs for step $t$.
Further,  for every LHT  triple $((xa, by), wt, count)$  in step $t$,
there are HTs $(a*,by)$ and 
 $(xa,*b)$  with weights $wt- w(x,a)$ and $wt - w(b,y)$ respectively, in
 that step.
\end{lemma}

\begin{proof}
	For the first part see correctness in Appendix \ref{append:corr} (Lemmas \ref{lem:fcleanup-corr2} and \ref{proof:fflem}). For the second part, if a triple $\gamma = ((xa, by), wt, count)$ is present in step $t$, then $\gamma$ was generated during \fullyfixup of some step $t' \leq t$. By the construction of our algorithm, at the end of step $t'$, the \system contains at least one HT of the form $(a*,by)$ and one of the form $(xa,*b)$  with weights $wt- w(x,a)$ and $wt - w(b,y)$ respectively. W.l.o.g. suppose that the set $S$ of all the HTs of the form $(a*,by)$ and weight $wt- w(x,a)$ are removed during \fullycleanup at some step $t'' \leq t$, then since these HTs are the right constituents of $\gamma$, when $S$ is left extended to $x$ in \fullycleanup then exactly $count$ paths will be removed from $\gamma$ making the triple disappear from the \systemend. 
	Thus at step $t$, at least one HT of the form $(a*,by)$ with weight $wt- w(x,a)$ must be in $P^*$. Similarly, there must be an HT of the form $(xa,*b)$ with weight $wt- w(b,y)$ in $P^*$ at step $t$.
\end{proof}

%Ms
\begin{lemma}\label{lem:tht}
If \fullyupdate is called on vertex $v$ at step $t$, then at the end of step~$t$ any TLHT $(xa, by)$ in the \system{} that contains a path through $v$, has the vertex $v$ as one of the endpoints $x$ or $y$.
\end{lemma}

\begin{proof}
	By Lemma~\ref{lem:lhts},
	any  LHT in the \system is formed
	by combining two HTs that are in the \systemend.
	Now consider the TLHT $(xa,by)$ that contains $v$. Since by definition of
	TLHT, $(xa,by)$ is not an LST in the current graph, assume w.l.o.g. that $(xa,b)$ is the subtuple
	that is not an ST when an end edge is deleted from $(xa,by)$. But this is not possible since any
	tuple HT $(xa, *b)$ that contains $v$ must be an ST.  The lemma follows.
\end{proof}

\begin{lemma}\label{lem:z}
Let $G$ be a graph after a sequence of calls to \fullyupdateend , and
suppose every HT in the \system
 is an ST in one of $z$ different graphs $H_1, \cdots , H_z$, and
every LHT is formed from these HTs. If $n$ and $m$ bound the number of vertices and edges,
respectively, in any of these graphs, and if
$\nu^*$ bounds the maximum number of edges that lie on shortest paths through any given vertex in any of the these graphs, then:

\begin{enumerate}
\item The number of LHTs  in $G$'s \system is at most $O(z \cdot m \cdot \nu^*)$.

\item The number of LHTs that contain the newly updated vertex $v$ in $G$ is $O( z \cdot {\nu^*}^2)$.

\item Let $u$ be any vertex in $G$, and suppose that the HTs that contain $u$ lie in $z' \leq z$ of the
graphs $H_1, \cdots , H_z$. Then, the number of LHTs that contain the vertex $u$  is
 $O( (z + z'^2)  \cdot {\nu^*}^2)$.
\end{enumerate}
\end{lemma}

\begin{proof}
	For part 1, we bound the number of LHTs $(xa,by)$ (across all weights) that can exist in
	$G$. The edge $(x,a)$ can be chosen in $m$ ways, and once we fix $(x,a)$, the 
	$r$-tuple $(a,by)$ must be an ST in one of the $H_j$. Since $(b,y)$ must lie on  
	a shortest path through
	$a$ in the graph $H_i$ that contains the $r$-tuple $(a,by)$ of that weight, the number of different choices
	for $(b,y)$ that will then uniquely determine the tuple $(xa, by)$, together with its weight, is
	$z \cdot \nu^*$. Hence the number of LHTs in $G$ `s \system is $O(z \cdot m \cdot \nu^*)$.
	
	For part 2, the number of LHTs that contain $v$ as an internal vertex is simply the number of
	LSTs in the current graph by Lemma~\ref{lem:tht}, and using
	Lemma~\ref{lem:bound-tuples-thru-v}, this is  $O( \vstar^2)$. We now
	bound the number of LHTs $(va, by)$. There are $n-1$ choices for the edge $(v,a)$ and
	$z \cdot \nu^*$ choices for the $r$-tuple $(a, by)$, hence the total number of such tuples 
	is $O(z \cdot n \cdot \nu^*)$. The same bound holds for LHTs of the form $(xa, bv)$. Since
	$\nu^* = \Omega (n)$, the result in part 2 follows.
	
	For part 3, we observe that each LHT that contains $u$ as an internal vertex must be composed
	of two HTs, each of which is an ST in one of the $z'$ graphs that contain $v$. Thus, there are
	$O(z'^2 \cdot \vstar^2)$ such tuples. For an LHT, say $\tau=(ua,by)$,  that contains $u$ as an end vertex, 
	the analysis remains the same as above in part 2: there are $n-1$ choices for the edge $(u,a)$ and
	$z \cdot \nu^*$ choices for the $r$-tuple $(a, by)$, hence the total number of such tuples 
	is $O(z \cdot n \cdot \nu^*)$. This gives the desired result.
\end{proof}
\subsection{The Overall Algorithm \fullydynamic}\label{sec:alg-main}

Algorithm \fullyupdate in Section~\ref{sec:alg-basic} is a correct fully dynamic algorithm for
APASP,  but it is not a very efficient algorithm, since $C$ and $D$ in
Lemma~\ref{lem:basic-amort} could be very large.
%VLR: Rearranged
In particular, by generalizing the counterxample in \DI (for maintaining all the historical paths), we can show that storing all the THTs increases the space used by our algorithm to $\widetilde{O}(n \cdot m \cdot \vstar)$, with a worst-case amortized complexity of $\Omega(n\cdot \vstar^2)$. To address this issue we introduce some additional `dummy' updates into the sequence of update operations as described below. Dummy updates to improve performance were first used in \DI for fully dynamic APSP.
However, that dummy update sequence presents several limitations for APASP. We now briefly describe the \DI dummy upates,and then we introduce our new method.

In the \DI method for `dummy updates',
a vertex updated at time $t$ is also given a `dummy' update at steps $t + 2^i$, for each $i \geq 0$ (this update is performed along with the real update at step $t  + 2^i$).
The effect of
a dummy update on a vertex $v$ is to remove any HP or LHP that contains $v$, thereby streamlining the
collection of paths maintained.
A useful property when SPs are unique (as in \DIe) is that each HP in $P^*(x,y)$, for a given pair $x,y$, will have a different
weight.  An $O(\log n)$ bound on the number of HPs in a $P^*(x,y)$ is established in \DI 
as follows.
Let the current time step be $t$,
and consider an HP $\tau$ last updated at $t'<t$. Let us denote the smallest $i$ such that $t'+ 2^i > t$ 
as the dummy-index for  $\tau$. By observing that 
different HPs for $x,y$ must have different dummy-indices, it follows that their number is $O(\log t)$, which is $O(\log n)$
since the data structure is reconstructed after $O(n)$ updates.

% New part
Note that, by generalizing the counterxample (for maintaining all the historical paths) in \DIe, storing all the THTs increases the space used by our algorithm to $\widetilde{O}(n \cdot m \cdot \vstar)$, with a worst-case amortized complexity of $\Omega(n\cdot \vstar^2)$.
To avoid this scenario, our algorithm uses a refined update sequence, tailored for multiple shortest paths, as described in Section \ref{sec:pdg}.

We now present our overall fully dynamic algorithm
for APASP.
As in~\cite{DI04,Thorup04} we build up the \system for the initial $n$-node graph $G=(V,E)$ with $n$ inserts starting
with the empty graph (and hence $n$ \incremental updates),  and we then
perform the first $n$ updates in the given update sequence $\Sigma$. After these $2n$ updates, we
reset all data structures and start afresh. 

Consider a graph  $G= (V,E)$ with weight function $\weight$ in which a \incremental or \decremental update is
applied to a vertex $u$. Let $\weight'$ be the weight function after the update, hence the only changes to
the edge weights occur on edges incident to $u$. 
Algorithm \ref{algo:fully-main} gives the overall fully dynamic algorithm for the $t$-th update to a vertex $v$ with
the new weight function $\weight'$. 
This algorithm applies \fullyupdate to vertex $v$ with the new weight function,
thus it
  will be correct if we executed only the first step, but not necessarily efficient. To obtain an efficient algorithm we execute
`dummy updates' on a sequence $\mathcal{N}$ of the most recently updated vertices as specified in Steps 2-5. 
The length of this sequence of vertices is 
determined by the position $k$ of the least significant bit set to 1 in the bit representation  
$B=b_{r-1} \cdots b_0$ of
$t$.
We denote $k$ by $\setbit(t)$.

\begin{algorithm}
  \SetAlgoLined
  %\KwData{this text}
  %\KwResult{how to write algorithm with \LaTeX2e }
  \fullyupdateend$(v,\weight')$\;
  $k \leftarrow \setbit (t)$ (i.e., if the bit representation of $t$ is $b_{r-1} \cdots b_0$, then $b_k$ is  the least significant bit with value 1)\;
  $\NODES \leftarrow$ set of vertices updated at steps $t-1, \cdots, t-(2^k -1)$\;  \label{algostep:nodes}
  \For{each $u \in \NODES$ in decreasing order of update time}{
  	\fullyupdateend $(u,\weight')$   $~~$ (dummy updates)\;
  }	
\caption{\fullydynamicend($G, v, \weight',t$)}
\label{algo:fully-main}
\end{algorithm}

\noindent
{\bf Properties of  $\mathcal{N}$.}
Consider the current update step $t$, with its
   bit representation  $B = b_{r-1} \cdots b_0$  and with
$\setbit (t) = k$. We say that index $i$ is a {\it time-stamp} for $t$ if
 $b_i=1$ (so  $r-1 \geq i \geq k$), and for each 
 such time-stamp $i$, we let $time_t (i)$ be the earlier update
step $t'$ whose bit representation has zeros in positions $b_i -1, \cdots , 0$, and
which matches $B$ in positions  $b_{r-1} \cdots b_{i}$.  In other words,
$time _t(i) = t'$, where $t'$ has bit representation $b_{r-1} \cdots b_{i+1} b_i 0 \cdots 0$. 
 We do not define $time_t(i)$ if $b_i=0$. We define  $Prior\mh times (t)$ be the set of $time_t(i)$
where $i$ is a time-stamp for $t$, and we also include in $Prior\mh times (t)$ the initial time $t_0 = 0$.
Note that $|\priortimes (t)| \leq r + 1 = O(\log n)$, since $r \leq \log (2n)$.

 Let $G_t$ be the graph after  the $t$-th update is applied,
$t \geq 1$, with the initial graph being $G_0$ (at time $t_0 = 0$). Thus, 
the input graph to Algorithm~\ref{algo:fully-main} is $G_{t-1}$, and the updated
graph is $G_t$.

\begin{lemma}\label{lem:priortimes}
For every vertex $v$ in $G_t$, the step $t_v$ of the most recent update to $v$ is in $\priortimes(t)$.
\end{lemma}

\begin{proof}
Let the  bit representation of $t$  be $B = b_{r-1} \cdots b_0$, let $\setbit (t) =i$,
and let $j_1 > j_2 > \cdots > j_s=i$ be the time-stamps for $t$.
Let $t_u= time_t(b_{j_u})$, thus the bit representation of $t_u$ is the same as
$B$, with all bits in positions less than $j_u$ set to zero; let $t_0=1$. 
Then, we observe that  during the execution of Algorithm~\ref{algo:fully-main}  for the
$t_u$-th update, the vertex for update $t_u$ will be updated in Step 1, and the
vertices updated in steps $t_{u-1} +1, \cdots , t_u -1$ will be updated in 
Steps 4-5 of Algorithm~\ref{algo:fully-main}. Hence  all vertices updated in $[t_{u-1}+1, t_u]$ are more recently updated in step $t_u$. Thus the most recent update step for every vertex in
$G_t$ is one of the $O(\log n)$ steps in  $\priortimes(t)$.
\end{proof}

\subsection{Analysis of Algorithm \fullydynamic}\label{sec:anal}

The analysis in this section incorporates many elements from \Tho algorithm~\cite{Thorup04}.
However, these 
are present only in the analysis, and
 the only component of that rather complicated algorithm that we use is the
form of the dummy update sequence in Step~\ref{algostep:nodes} of 
Algorithm~\ref{algo:fully-main}.
 Our
Algorithm~\ref{algo:fully-main} is about as simple as the \DI algorithm when specialized to unique shortest paths since
in that case it suffices to use the update algorithm in \DI instead of the more elaborate \fullyupdate we use here.  

\subsubsection{The Prior Deletion Graph (PDG)}\label{sec:pdg}

Let $t '< t$ be two update steps, and let $W$ be the set of vertices that are updated in the
interval of steps $[t'+1, t]$. We define the
{\it prior deletion graph (PDG)} $\Gamma_{t',t}$ as the induced subgraph of $G_{t'}$ on the vertex set 
$V(G_{t'}) -W$. 
If $t$ is the current update step, then we simply use $\Gamma_{t'}$ instead of $\Gamma_{t',t}$.

We say that a path $p$ 
{\it is present in both $G_{t'}$ and $G_t$} if no call to \fullyupdate is made on
any vertex in $p$ during the update steps in the interval $[t'+1,t]$.
 
\begin{lemma}\label{lem:pdg}
Let tuple $\tau$ represent a collection of paths  in $G_{t'}$.Then,
\begin{enumerate}
\item If $\tau$ is an ST in $G_{t'}$ then $\tau$ continues to be an ST in every PDG $\Gamma_{t',t}$ with
$t \geq t'$ in which $\tau$ is present.

\item For any $\hat{t} \geq t'$, if $\tau$ is an ST in $G_{\hat{t}}$ then  $\tau$ is an ST  in every PDG $\Gamma_{t',t''}$, $ t''\geq \hat{t}$,  in which $\tau$ is present.
\end{enumerate}
\end{lemma}

\begin{proof}
As observed in~\cite{NPR14b} (and in~\cite{DI04} for unique shortest paths), an ST in a graph remains an ST
after a weight increase on any edge that is not on it. This establishes the first part since an increase-only updates does not affect the weight of existing STs that avoid the updated node. 
For the second
part, we observe that $\Gamma_{t',\hat{t}}$ can be viewed as being obtained from $G_{\hat{t}}$ by deleting the vertices updated in 
$[t'+1,\hat{t}]$. 
Since $\tau$ is an ST in $G_{\hat{t}}$, it continues to be an ST in the graph $\Gamma_{t',\hat{t}}$, which can be obtained
from $G_{\hat{t}}$ through a sequence of \decremental updates that do not change the weight of any edge on $\tau$. Finally since $\tau$ is an ST in $\Gamma_{t',\hat{t}}$, it must be an ST in any $\Gamma_{t',t''}$ in which it appears, for $t''>\hat{t}$.
\end{proof}

\noindent
{\bf PDGs for  Update $t$:}
We will associate with the current update step $t$, the set of PDGs
 $\Gamma_{t'}$, for $t' \in Prior\mh times(t)$.
These PDGs are similar to
the {\it level graphs} maintained in \Tho algorithm~\cite{Thorup04} (our \FFD algorithm in Section \ref{sec:improved} generalizes this approach to multiple shortest paths), but we choose to give them a different name since
we use them here only to analyze the performance of our algorithm.
%Note that, starting in Section \ref{sec:improved}, we will use these level graphs concretely to improve our running time by a logarithmic factor with our \FFD algorithm.

\begin{lemma}\label{lem:pdg2}
Consider a sequence of fully dynamic updates performed using Algorithm~\ref{algo:fully-main}.
Let the current update step be $t$.
%and consider the set of graphs $\Gamma_{t'}$, for $t' \in Prior\mh times(t)$.
Then, each HT in the \system{} for
 $G_t$ is an ST in at least one of the $\Gamma_{t'}$, where $t'  \in Prior\mh times(t)$.
 Further  $z=O( \log n)$ in Lemma~\ref{lem:z} for $G_t$.
\end{lemma}

%\begin{proof}
%	Consider an HT $\tau= (xa,by)$ in $G_t$. Let the most recently updated vertex in
%	$\tau$ be $v$, and let its update step be  $t_v \leq t$. By definition of HT, $\tau$ is an ST
%	in some $t'$ in $[t_v,t]$, and hence by part 2 of Lemma~\ref{lem:pdg}, using $\hat{t} = t' = t_v$ and $t''=t$, 
%	$\tau$ is an ST in $\Gamma_{t_v}$. By Lemma~\ref{lem:priortimes}, $t_v \in \priortimes (t)$.
%	Finally, since  $|\priortimes (t)| \leq \log n$ for any $t$, $z= O( \log n)$ in 
%	Lemma~\ref{lem:z}.
%\end{proof}

\begin{proof}
	Consider an HT $\tau= (xa,by)$ in $G_t$. Let the most recently updated vertex in
	$\tau$ be $v$, and let its update step be  $t_v \leq t$. 
	By definition of HT, $\tau$ is an ST
	in some $t' \leq t$. If $t' < t_v$ then trivially $\tau$ is an ST in $\Gamma_{t_v}$.
	Otherwise if $t'$ in $[t_v,t]$ then, by part 2 of Lemma~\ref{lem:pdg}, using $\hat{t} = t' = t_v$ and $t''=t$, 
	$\tau$ is an ST in $\Gamma_{t_v}$. By Lemma~\ref{lem:priortimes}, $t_v \in \priortimes (t)$.
	Finally, since  $|\priortimes (t)| \leq \log n$ for any $t$, $z= O( \log n)$ in 
	Lemma~\ref{lem:z}.
\end{proof}

\noindent
Recall that $\CH$ is the combined history of a triple $\gamma$ (as defined in Section \ref{sec:alg-basic}). Here, we bound its space complexity.
%newstart
\begin{corollary}  \label{lem:cor1}
The size of $\CH$ is $O(\log n)$ for each THT processed during cleanup. Thus the additional processing time for a triple in cleanup is $O(\log n)$ and, by Lemma \ref{lem:z} part 2, the overall space of the $\CH$ structures is $\widetilde{O}(\vstar^2)$.
\end{corollary}

\begin{proof}
	From lemma \ref{lem:pdg2}, a single THT $\gamma$ can exist in at most $z= O( \log n)$ PDG (with different counts). Moreover, all the paths joining $\gamma$ in a specific PDG $\Gamma_t$ can only be generated by the node updated at time $t$. Thus, the combined history $\CH(\gamma)$ contains at most $O(\log n)$ entries.   
\end{proof}
%newend

We will use the above lemma to obtain our amortized time bound in
Lemma~\ref{lem:amort-main} in the next section. It is not clear that a similar result can be obtained with
the \DI dummy update sequence.

\subsubsection{Amortized Cost of Algorithm \fullydynamic}\label{sec:amortized}

We will now bound the amortized cost of an update in  a sequence $\Sigma$  of $n$ real updates on an initial
$n$-node graph $G=(V,E)$. As mentioned earlier, in our method this will translate to a sequence $\Sigma'$ of
$2n$ real (as opposed to dummy) updates: there is  
an initial sequence of $n$ updates, starting with the empty
graph, which inserts each of the $n$ vertices in $G$ along with  incident edges that have not yet been
inserted. Following this is the sequence of $n$ real updates in $\Sigma$.  Each of these $2n$ updates in
$\Sigma'$ will make a call to Algorithm~\ref{algo:fully-main}. 
As before, let $\nu^*$ be a bound on the number of edges 
that lie on shortest paths through any given vertex 
in any of $G_t$.  The following lemma establishes our main Theorem~\ref{th:main}.

\begin{lemma}\label{lem:amort-main}
Algorithm~\ref{algo:fully-main} executes a sequence $\Sigma$ of $n$ real updates on an $n$-node graph in
$O(\vstar^2 \cdot \log^3 n)$ amortized time per update.
\end{lemma}

\begin{proof}
Let $n'=2n$, and as described above, let $\Sigma'$ be the sequence of $n'$ calls to 
Algorithm~\ref{algo:fully-main} 
used  to execute the $n$ updates in $\Sigma$.
We first observe that 
Algorithm~\ref{algo:fully-main} performs $O(n' \log n')$ dummy updates when executing these $n'$ calls.
This is because there are $n'/2^k$ real updates for update steps $t$ with $set\mh bit(t)=k$, and each such update is
accompanied by $2^k-1$ dummy updates. So, across all real updates there are $O(n')$
dummy updates for each position of $set\mh bit$, adding up to $O(n' \log n')$ in total, across all $set\mh bit$ positions.

We now use Lemma~\ref{lem:basic-amort} to bound the time needed to execute the $d= n'\log n'$ dummy updates and $n'$ real updates.
We need to bound the parameters $C$ and $D$ in Lemma~\ref{lem:basic-amort}. We first consider $D$.
For this we will use part 1 of Lemma~\ref{lem:z}.
By Lemma~\ref{lem:pdg2}, $z= O(\log n)$ for any $t$. Hence by  Lemma~\ref{lem:z}
the maximum number of tuples that can remain at the end of the update sequence is 
$D=O(m \cdot \vstar \cdot \log n)$.

Now we bound the parameter $C$, for which we will obtain separate bounds, $C_1$ for the real updates,
and $C_2$ for the dummy updates.
For each real update we have $z=z'= O(\log n)$ in part 3 of Lemma~\ref{lem:z}, hence 
the number of tuples  that contain a path
through the updated vertex is $O(\vstar^2 \cdot \log^2 n)$, thus $C_1 = O(\vstar^2 \cdot \log^2 n)$. 

Now consider a dummy update on a vertex $u$ in Step 5 of Algorithm~\ref{algo:fully-main}, and let $u$ be in some
level $i$ (note the $i$ must be less than the current level $k$). At the time
this call is made, all vertices that were updated after $u$ was last updated in the graph (i.e., after $time (i)$) are now in the newest
level $k$. Thus, any LHT that contains $u$ lies  in either $\Gamma_{time(i)}$ or in $G_t = \Gamma_t$. Hence $z'=2$  in part 3
of Lemma~\ref{lem:z} for any vertex $u$ undergoing a dummy update. Thus $C_2 = O(\vstar^2 \cdot \log n)$.

Hence the total time taken by Algorithm~\ref{algo:fully-main} for its $d= n'\log n'$ dummy updates  and the $n'$ real updates is,
by Lemma~\ref{lem:basic-amort},  $O(~(n' \cdot (n^2 + C_1)  + (n' \log n') \cdot (n^2 + C_2) + D) \cdot \log n)
= O(n' \cdot \vstar^2 \cdot \log^2 n + (n' \cdot \log n')  \cdot \vstar^2 \cdot \log n + D \log n) =  O(n \cdot \vstar^2 \cdot \log^3 n)$ 
(the $n^2$ and $D$ terms are  dropped
since $\vstar = \Omega (n)$, 
hence $m=O(n\cdot \vstar)$).
 
It follows that the amortized cost of each of the $n$ updates in $\Sigma$ 
 is $O(\frac{1}{n} \cdot n' \cdot \vstar^2 \cdot \log^3 n) = O(\vstar^2 \cdot  \log^3 n)$.
\end{proof}

\section{A Faster Fully Dynamic APASP Algorithm: \FFD } \label{sec:improved}
In this section we present our improved algorithm \FFD for computing APASP, which uses new techniques data structures (see Section \ref{sec:newfeat}). Our approach is based on the results obtained in \cite{Thorup04}, for which we give a brief introduction in the following section. We then discuss the new data structures involved in our result (Sections \ref{sec:lvl-tuple-system} and \ref{sec:newfeat}), new features not present in \cite{Thorup04}, that arise from having multiple SPs (Sections \ref{sec:ffeatures}), and an overall description of the \FFD algorithm and its new components (Sections \ref{sec:ffdalgo}).  

\subsection{The \Tho Fully Dynamic APSP Algorithm}
In \cite{Thorup04}, Thorup improves by a logarithmic factor over \DI (for unique shortest paths) by using a \emph{level system} of \decremental graphs. 
The shortest paths and locally shortest paths are generated level by level leading to a different complexity analysis from \DIe. When a node is removed from the current graph, it is also removed from every older level graph that contains it. 
The implementation of the \Tho APSP algorithm is not fully specified in \cite{Thorup04}. For our \FFD algorithm, we present generalizations of the data structures sketched in \Tho together with new data structures required to achieve efficiency (see Section \ref{subsec:fully-impl} for a summary of these data structures).     

\subsection{Data Structures for Algorithm  \FFDe} \label{subsec:fully-impl}
Our  algorithm  \FFD requires several data structures. Some of these are already present in \NPR and \Thoe, while others are newly defined or generalized from earlier ones. We will use components from our basic algorithm such as the abstract representation of the level system using PDGs (see Section \ref{sec:lvl-tuple-system}) and the flag bit $\beta$ for a triple, the 
 \MT scheme introduced in \NPR (see \cite{NPR14b} for more details), and the maintenance of level  graphs from \Thoe.

In the following sections, we describe all data structures used by our algorithm. 
In Table II we summarize the structures we use, including those inherited from \cite{NPR14b,Thorup04}.
The new components we introduce in this paper to achieve efficiency for fully dynamic APASP, are described in section \ref{sec:newfeat} and listed in Table II, Part D.

\subsubsection{A Level System for Centered Tuples} \label{sec:lvl-tuple-system}
Algorithm \FFD  uses the PDGs defined in \ref{sec:pdg} as real data structures similar to \Tho for APSP. This is done in order to generate a smaller superset of LSTs than \fullydynamicend, and it is the key to achieving the improved efficiency. Here we describe the level system and the data structures we use in \FFDe, with special attention to
the new elements we introduce.

As in~\cite{DI04,Thorup04} we build up the \system for the initial $n$-node graph $G=(V,E)$ with $n$ insert updates (starting
with the empty graph), and we then
perform $n$ updates according to the update sequence $\Sigma$. After $2n$ updates, we
reset all data structures and start afresh.

Our level system is a generalization of \Tho to fully dynamic APASP.
For an update at step $t$, let $k$ be the position of the least significant bit with value 1 in the binary representation of $t$. Then
the $t$-th update activates level $k$, and deactivates all levels $j<k$ by folding these levels into level $k$.
These levels are considered implicitly in our basic result, and using the same notation, we will say that $time(k)=t$, and $level(t)=k$; moreover $G_t$ indicates the graph after the $t$-th update.
Note that the largest level created before we start afresh is $r= \log 2n$.

\paragraph{\bf Centering vertices and tuples/triples} As in \Thoe, each \emph{vertex $v$ is centered} in level $k=level(t)$, where $t$ is the most recent step in which $v$ was updated.
A path $p$ in a tuple is centered in level $k'=level(t')$, where $t'$ is the most recent step in which $p$ entered the tuple system (within some tuple) or was modified by a vertex update. 
Hence, in contrast to \Thoe, a triple can represent paths centered in different levels. Thus,
for a triple $\gamma = ((xa, by), wt, count)$ we maintain an array $\CA_{\gamma}$ where 
$$\CA_{\gamma}[i] = \textrm{number of paths represented by $\gamma$ that are all centered in level $i$}$$

\noindent
It follows that $\sum_i{\CA_{\gamma}[i]}=count$. 
The \emph{level center of the triple $\gamma$} is the smallest (i.e., most recent level) $i$ such that $\CA_{\gamma}[i] \neq 0$.

\paragraph{\bf Level graphs (PDGs)}
In our basic result, the PDGs (introduced in \ref{sec:pdg}) are used only in the analysis, and are not maintained by the algorithm. Here, in \FFDe, we will maintain a set of local data
structures for each PDG that is relevant to the current graph; also, in a small change of notation, we will denote a level graph for
time $t'\leq t$ as $\Gamma_{k'}$, where $k' = level(t')$ rather than the our previous notation of $\Gamma_{t'}$. These graphs are similar to the level graphs in \Thoe. 
As in \Thoe, 
  only certain information for $\levelgraph_k$ is explicitly maintained in its local data structures: the STs centered in level $k$ plus all the extensions that can generate STs in $\levelgraph_k$. The data structures used by our algorithm to maintain triples are Global and Local, which we now describe.

\begin{table}[ht] 			
	\begin{center}
		\begin{tabular}{|c|c|c|c|}
			\hline 
			\textbf{Notation} & \textbf{Data Structure} &  \multicolumn{2}{|c|} {\textbf{Appears}} \\ 
			\hline
			\hline 
			\multicolumn{4}{|c|}{\textbf{Part A :: Global Data Structures} {\scriptsize (for each pair of nodes $(x, y)$)}} \\ 
			\hline 
			\hline 
			$P(x,y)$ & all (centered) \LHT s from $x$ to $y$ with weight as key &  \multicolumn{2}{|c|}{ \multirow{2}{*}{\parbox[c]{3cm}{(\cite{DI04} for paths, \cite{NPR14b} for LSTs)}}} \\ 
			\cline{1-2} 
			$P^{*}(x,y)$ & all (centered) \HT s from $x$ to $y$ with weight as key & \multicolumn{2}{|c|}{ } \\ 
			\hline 
			$L(x,by)$ & $\{x' : (x'x, by)$ denotes a (centered) \LHTe $\}$ & \multicolumn{2}{|c|}{ (\cite{NPR14b} for LSTs) } \\ 
			\hline 
			$R(xa, y)$ & $\{y' : (xa, yy')$ denotes a (centered) \LHTe $\}$ & \multicolumn{2}{|c|}{ (\cite{NPR14b} for LSTs) }  \\ 
			\hline 
			\MT & global dictionary for Marking scheme &  \multicolumn{2}{|c|}{ \cite{NPR14b}}\\ 
			\hline 
			\hline
			\multicolumn{4}{|c|}{\textbf{Part B :: Local Data Structures} {\scriptsize (for each active level $i$, for each pair of nodes $(x,y)$)}} \\  
			\hline 
			\hline
			$P^{*}_i(x,y)$ & STs from $x$ to $y$ centered in level $i$ & \multicolumn{2}{|c|}{  \multirow{5}{*}{\parbox[c]{2.5cm}{(sketched in \cite{Thorup04} for paths)}}}\\ 
			\cline{1-2} 
			$L_i^{*}(x,y)$ & $\{x' : (x'x, y)$ denotes an $\ell$-tuple for SPs centered in level $i \}$ & \multicolumn{2}{|c|}{ } \\ 
			\cline{1-2} 
			$R_i^{*}(x,y)$ &  $\{y' : (x, yy')$ denotes an $r$-tuple for SPs centered in level $i \}$ & \multicolumn{2}{|c|}{ } \\ 
			\cline{1-2} 
			$LC_i^{*}(x,y)$ & the subset $\{ x' \in L_i^{*}(x,y) : x' \textrm{ is centered in level } i \}$  & \multicolumn{2}{|c|}{ }\\ 
			\cline{1-2} 
			$RC_i^{*}(x,y)$ &  the subset $\{ y' \in R_i^{*}(x,y) : y' \textrm{ is centered in level } i \}$ & \multicolumn{2}{|c|}{ }\\ 
			\hline
			$dict_i$ & dictionary of pointers from local STs to global $P$ and $P^*$ &  \multicolumn{2}{|c|}{new} \\
			\hline 
			\hline
			\multicolumn{4}{|c|}{\textbf{Part C :: Inherited Data Structures}} \\  
			\hline 
			\hline
			$\beta(\gamma)$ & flag bit for the (centered) triple $\gamma$ & \multicolumn{2}{|c|}{ basic algorithm}\\ 
			\hline
			$level(t)$ & level activated during $t$-th update &  \multicolumn{2}{|c|}{ basic algorithm}\\ 
			\hline 
			$time(k)$ & most recent update in which level $k$ is activated & \multicolumn{2}{|c|}{ basic algorithm}\\ 
			\hline 
			$\mathcal{N}$ & nodes (centered in levels) deactivated in the current step & \multicolumn{2}{|c|}{ basic algorithm }\\ 
			\hline 
			$\Gamma_k$ & level graph (PDG) created during $time(k)$-th update &  \multicolumn{2}{|c|}{ (\cite{Thorup04} for paths)}\\ 
			\hline
			\hline
			\multicolumn{4}{|c|}{\textbf{Part D :: New Data Structures}} \\  
			\hline 
			\hline
			$C_\gamma$ & distribution of all paths in triple $\gamma$ among active levels & \multicolumn{2}{|c|}{new} \\ 
			\hline
			$\DMs(x,y)$ & linked-list containing the history of distances from $x$ to $y$ & \multicolumn{2}{|c|}{new}\\
			\hline
			$\LN(x,y,wt)$ & the set $\{b : \exists \, (xa,by)$ of weight $wt$ in $P(x,y)\}$ & \multicolumn{2}{|c|}{new}\\
			\hline
			$\RN(x,y,wt)$ & the set $\{a : \exists \, (xa,by)$ of weight $wt$ in $P(x,y)\}$ & \multicolumn{2}{|c|}{new}\\
			\hline  
		\end{tabular}
		\label{table:dsA}
		\caption{Notation summary.}
	\end{center}
\end{table}

\paragraph{\bf Global Structures}
The global data structures are $P^*$, $P$, $L$ and $R$ (see Table II, Part~A).
\begin{itemize}
	\item The structures $P^*(x,y)$ and $P(x,y)$ will contain HTs (including all STs) and LHTs (including all LSTs), respectively,  from $x$ to $y$. They are priority queues with the weights of the triple and a flag bit $\beta$ as key. For a triple $\gamma$ in $P$, the flag bit $\beta(\gamma)=0$ if the triple $\gamma$ is in $P$ but not in $P^*$, and $\beta(\gamma)=1$ if the triple $\gamma$ is in $P$ and $P^*$.
	\item The structure $L(x,by)$ ($R(xa,y)$) is the set of all left (right) extension vertices that generate a centered LHT in the \systemend.
\end{itemize}

\paragraph{\bf Local Structures}
The local data structures we introduce in this paper are $L^*_i, R^*_i, LC^*_i$ and $RC^*_i$ (see Table II, Part~B). These are generalization of the data structures sketched in \Tho for unique SPs in the graph. For every pair of nodes $(x,y)$:

\begin{itemize}
	\item The structure $P^*_i(x,y)$ contains the set of STs from $x$ to $y$ centered in $\levelgraph_i$. It is implemented as a set. 
	\item The structure $L^*_i(x,y)$ ($R^*_i(x,y)$) contains all left (right) extensions that generate a shortest $\ell$-tuple ($r$-tuple) centered in level $i$. It is implemented as a balanced search tree.
	\item The structure $LC^*_i(x,y)$ ($RC^*_i(x,y)$) contains left (right) extensions centered in level $i$ that generate a shortest $\ell$-tuple ($r$-tuple) centered in level $i$. It is implemented as a balanced search tree.
	\item  A dictionary $\dict_i$, contains STs in $P^*_i$ using the key $[x, y, a, b]$ and two pointers stored along with each ST. The two pointers refer to the location in $P(x,y)$ and $P^*(x,y)$
	of the triple of the form $(x,a,b,y)$ contained in $P^*_i(x,y)$.
\end{itemize}

\noindent
In order to recompute BC scores (see Section \ref{sec:fdbc}) we will consider $R^*(x,y)=\bigcup_i{R_i^*(x,y)}$ and similarly $L^*(x,y)=\bigcup_i{L_i^*(x,y)}$. 

\subsubsection{New Structures} \label{sec:newfeat}
We introduce two completely new data structures, not used in previous results \cite{DI04,Thorup04,NPR14b}, which are essential to achieve efficiency for our \FFD algorithm. Both are needed to address the Partial Extension Problem (PEP) which does not occur in \Tho for fully dynamic APSP (see section \ref{sec:ffeatures}).

\paragraph{\bf Distance History Matrix}
 The \emph{distance history matrix} is a matrix $\DMs$ where each entry is a pointer to a linked list: for each $x, y \in V$, the linked list $\DMs(x,y)$ contains the sequence of different pairs $(wt,k)$, where each one represents an SP weight $wt$ from $x$ to $y$, along with the most recent level $k$ in which the weight $wt$ was the shortest distance from $x$ to $y$ in the graph $\levelgraph_k$. 
 The pair with weight $wt$ in $\DMs(x,y)$ is double-linked to every triple from $x$ to $y$ with weight $wt$ in the system. Precisely, when a new triple $\gamma$ from $x$ to $y$ of weight $wt$ is inserted in the algorithm, a link is formed between $\gamma$ and the pair $(wt,k)$ in $\DMs(x,y)$. With this structure, the \FFD algorithm can quickly check if there are still triples of a specific weight in the \systemend, especially for example when we need to remove a given weight from $\DMs$.
 Note that the size of each linked list is $O(\log n)$.

\paragraph{\bf Historical Extension (HE) Sets RN and LN}
Another important type of structure we introduce are the sets $RN$ and $LN$. These structures are crucial to select efficiently the set of restored historical tuples that need to be extended (see Appendix \ref{ffixup-desc}). $RN(x,y,wt)$ ($LN(x,y,wt)$ works symmetrically) contains all nodes $b$ such that there exists at least one tuple of the form $(x \times,by)$ and weight $wt$ in $P(x,y)$. Similarly to $\DMs$, every time a new triple $\gamma$ of this form is inserted in the \systemend, a double link is created between $\gamma$ and the occurrence of $b$ in $RN(x,y,wt)$ in order to quickly access the triple when needed.

\noindent
The total space used by $\DMs$, $\RN$ and $\LN$ is $O(n^2 \log n)$. This is dominated by the overall space used by the algorithm to maintain all the triples in the \system across all levels (see Lemma \ref{lemma:count2}, Section \ref{sec:proof}).
\subsection{New Features in Algorithm \FFDe} \label{sec:ffeatures}
In this section, we discuss two challenges that arise when we attempt to generalize the level graph method used in \Tho (for APSP with unique SPs) to a fully dynamic APASP algorithm. Both are addressed by the algorithms in Appendix \ref{sec:algo}.

\noindent
\paragraph{\bf The bit $\beta$ feature}
The control-bit $\beta$ is introduced (and only briefly described) in our basic result to avoid the processing of untouched historical triples.
Here, we elaborate on this technique in more details and we also describe how it helps in the more complex setting of the level \systemend.

Consider figure \ref{fig:b}.
The ST $\gamma=((xa,by),wt,count)$ is created in level $k$ (Fig. \ref{fig:bk}). At $time(k)$, we have $\gamma \in P^*$ and also $\gamma \in P$ with $\beta(\gamma)=1$. In a more recent level $j < k$, a shorter triple $\gamma'=((xv,vy),wt', count')$, with weight $wt' < wt$, that goes trough an updated vertex $v$ is generated (Fig. \ref{fig:bj}). Thus at $time(j)$, we have $\gamma' \in P^*$ and also $\gamma' \in P$ with $\beta(\gamma')=1$; but $\gamma$ still appears in both $P^*$ and $P$ as a historical triple. Finally, a new LST $\gamma''=((xa',b'y),wt,count'')$, with the same weight as $\gamma$, is generated in level $i<j$ (Fig. \ref{fig:bi}). Note that $\gamma''$ is only in $P$ with $\beta(\gamma'')=0$ and not in $P^*$, as is the case of every LST that is not an ST. When an increase-only update removes $v$ and the triple $\gamma'$, the algorithm needs to restore all the triples with shortest weight $wt$. But while $\gamma$ is historical and does not require any additional extension, $\gamma''$ is only present in $P$ and needs to be processed. Our \FFD algorithm achieve this by checking the bit $\beta$ associated to each of these triples. The algorithm will extract and process all the triples with $\beta=0$ from $P(x,y)$.
These guarantees that a triple only present in $P$, or present in $P$ and $P^*$ with different counts is never missed by the algorithm.
\begin{figure}
	\centering
	\subfigure[level $k$]{
		\makebox[.3\textwidth]{
			\begin{tikzpicture}[every node/.style={circle, draw, inner sep=0pt, minimum width=5pt}]
			\node (x)[label=above:$x$] at (0,0)  {};
			\node (a)[label=left:$a$] at (0,-0.8) {};
			\node (b)[label=below left:$b$] at (0,-2.2) {};
			\node (y)[label=below:$y$] at (0,-3) {}; 
			\draw[->] (x) -- (a);
			\path[->,decoration={snake}] { (a) edge[decorate] (b)};
			\path[->,decoration={snake}] { (a) edge[bend right=60,decorate] (b.west)};
			\path[->,decoration={snake}] { (a) edge[bend left=60,decorate] (b.east)};
			\draw[->] (b) -- (y);
			\end{tikzpicture}} \label{fig:bk}} 
	\subfigure[level $j<k$]{
		\makebox[.3\textwidth]{
			\begin{tikzpicture}[every node/.style={circle, draw, inner sep=0pt, minimum width=5pt}]
			\node (x)[label=above:$x$] at (0,0)  {};
			\node (v)[label=right:$v$] at (1,-1.5) {};
			\node (y)[label=below:$y$] at (0,-3) {}; 
			\draw[->] (x) -- (v);
			\draw[->] (v) -- (y);
			\end{tikzpicture}} \label{fig:bj}} 
	\subfigure[level $i<j$]{
		\makebox[.3\textwidth]{
			\begin{tikzpicture}[every node/.style={circle, draw, inner sep=0pt, minimum width=5pt}]
			\node (x)[label=above:$x$] at (0,0)  {};
			\node (a1)[label=above left:$a'$] at (-1,-0.8) {};
			\node (b1)[label=below left:$b'$] at (-1,-2.2) {};
			\node (y)[label=below:$y$] at (0,-3) {}; 
			\draw[->] (x) -- (a1);
			\path[->,decoration={snake}] { (a1) edge[bend right=60,decorate] (b1.west)};
			\path[->,decoration={snake}] { (a1) edge[bend left=60,decorate] (b1.east)};
			\draw[->] (b1) -- (y);
			\end{tikzpicture}} \label{fig:bi}}
	\caption{The bit $\beta$ feature}
	\label{fig:b}
\end{figure} 

\noindent
\paragraph{\bf The partial extension problem (PEP)}
Consider the update sequence described below and illustrated in figure \ref{fig:ll}. 
Here the STs $\gamma=((xa,by),wt,count)$ and $\hat{\gamma}=((xa,cy),wt,count')$ 
are created in level $k$ (Fig. \ref{fig:llk}). Later, a left-extension to $x'$ generates the STs $\gamma'=((x'x,by),wt',count)$ and $\hat{\gamma}'=((x'x,cy),wt',count')$ in level $j<k$ (Fig. \ref{fig:llj}). Note that $\gamma$, $\hat{\gamma}$, $\gamma'$ and $\hat{\gamma}'$ are all present in $P^*$ and $P$ at $time(j)$. In a more recent level $i<j$, a decrease-only update on $v$ generates a shorter triple $\gamma_s=((xv,vy),wt_s, count_s)$ from $x$ to $y$, with $wt_s < wt$ going through $v$. 
In the same level, the triple $\gamma_s$ is also extended to $x'$ generating a triple $\gamma'' = ((x'x,vy),wt'',count_s)$ shorter than $\gamma'$ and $\hat{\gamma}'$ (Fig. \ref{fig:lli}). 
Thus at $time(i)$, $\gamma$, $\hat{\gamma}$, $\gamma'$ and $\hat{\gamma}'$ remain in $P^*$ as historical triples. 
Then, in level $h<i$, an update on $x''$ inserts the edges $(x'',x)$ and $(x'',c)$. This update generates an ST $\gamma'''=((x''c,cy))$ (shorter than $(x''x,vy)$) and also inserts $x'' \in LC_h^*(x,b)$) since $(x'',x)$ is on a shortest path from $x''$ to $b$; but it should not generate the triple of the form $(x''x,by)$ because $b$ is not on a shortest path from $x$ to $y$ at $time(h)$ (Fig. \ref{fig:llh}). 

\begin{figure}
	\centering
	\subfigure[level $k$]{
		\makebox[.20\textwidth]{
			\begin{tikzpicture}[every node/.style={circle, draw, inner sep=0pt, minimum width=5pt}]
			\node (x)[label=above:$x$] at (0,0)  {};
			\node (a)[label=left:$a$] at (0,-0.8) {};
			\node (b)[label=below left:$b$] at (0,-2.2) {};
			\node (c)[label=left:$c$] at (-1,-2.2) {};
			\node (y)[label=below:$y$] at (0,-3) {}; 
			\draw[->] (x) -- (a);
			\path[->,decoration={snake}] { (a) edge[decorate] (b)};
			\path[->,decoration={snake}] { (a) edge[decorate] (c)};
			\path[->,decoration={snake}] { (a) edge[bend left=60,decorate] (b.east)};
			\draw[->] (b) -- (y);
			\draw[->] (c) -- (y);
			\end{tikzpicture}} \label{fig:llk}} 
	\subfigure[level $j < k$]{
		\makebox[.20\textwidth]{
			\begin{tikzpicture}[every node/.style={circle, draw, inner sep=0pt, minimum width=5pt}]
			\node (x1)[label=above:$x'$] at (-0.5,0.8)  {};
			\node (x)[label=above right:$x$] at (0,0)  {};
			\node (a)[label=left:$a$] at (0,-0.8) {};
			\node (b)[label=below left:$b$] at (0,-2.2) {};
			\node (c)[label=left:$c$] at (-1,-2.2) {};
			\node (y)[label=below:$y$] at (0,-3) {}; 
			\draw[->] (x1) -- (x);
			\draw[->] (x) -- (a);
			\path[->,decoration={snake}] { (a) edge[decorate] (b)};
			\path[->,decoration={snake}] { (a) edge[decorate] (c)};
			\path[->,decoration={snake}] { (a) edge[bend left=60,decorate] (b.east)};
			\draw[->] (b) -- (y);
			\draw[->] (c) -- (y);
			\end{tikzpicture}} \label{fig:llj}} 
	\subfigure[level $i < j$]{
		\makebox[.20\textwidth]{
			\begin{tikzpicture}[every node/.style={circle, draw, inner sep=0pt, minimum width=5pt}]
			\node (x1)[label=above:$x'$] at (-0.5,0.8)  {};
			\node (x)[label=above right:$x$] at (0,0)  {};
			\node (v)[label=right:$v$] at (1,-1.5) {};
			\node (y)[label=below:$y$] at (0,-3) {}; 
			\draw[->] (x1) -- (x);
			\draw[->] (x) -- (v);
			\draw[->] (v) -- (y);
			\end{tikzpicture}} \label{fig:lli}}
	\subfigure[level $h < i$]{
		\makebox[.20\textwidth]{
			\begin{tikzpicture}[every node/.style={circle, draw, inner sep=0pt, minimum width=5pt}]
			\node (x2)[label=above:$x''$] at (0.5,0.8)  {};
			\node (x)[label=above left:$x$] at (0,0)  {};
			\node (a)[label=left:$a$] at (0,-0.8) {};
			\node (v)[label=right:$v$] at (1,-1.5) {};
			\node (b)[label=below left:$b$] at (0,-2.2) {};
			\node (c)[label=left:$c$] at (-1,-2.2) {};
			\node (y)[label=below:$y$] at (0,-3) {};
			\draw[->] (x2) -- (x);
			\path[->] { (x2) edge[bend right=45] (c)}; 
			\draw[->] (x) -- (a);
			\path[->,decoration={snake}] { (a) edge[decorate] (b)};
			\path[->,decoration={snake}] { (a) edge[bend left=60,decorate] (b.east)};
			\draw[->] (x) -- (v);
			\draw[->] (c) -- (y);
			\end{tikzpicture}} \label{fig:llh}} 
	\caption{PEP instance (only centered STs are kept in each level) -- all edge weights are unitary}
	\label{fig:ll}
\end{figure} 

When an increase-only update removes $v$ and the shortest triple ${\gamma}_s$ from $x$ to $y$, the algorithm needs to restore all  historical triples with shortest weights from $x$ to $y$. When $\gamma$ and $\hat{\gamma}$ are restored, we need to perform suitable left extensions as follows. An extension to $x''$ is needed only for $\gamma$: in fact $\hat{\gamma}$ should not be extended to $x''$ because the $\ell$-tuple of the newly generated tuple is not an ST in the graph.  On the other hand, no extension to $x'$ is needed since both $\gamma'$ and 
$\hat{\gamma}'$ will be restored (from HT to ST).  Our algorithm needs to distinguish all of these cases correctly and efficiently.

In order to maintain both correctness and efficiency in this scenario for APASP, we use two new data structures (described in Section \ref{sec:newfeat}): (1) the
historical distance matrix $\DMs$ that allow us to efficiently determine the most recent level graph in which an HT was an ST, 
and (2) the HE sets $\LN$ and $\RN$ that allow us to efficiently identify exactly those new extensions that need to be performed.
The methodology of these data structures is fully discussed in the description of \ffullyfixup (Appendix \ref{ffixup-desc}).
Note that, the PEP doesn't arise in \Tho because of the unique SP assumption: in fact when only a single SP of a given length is present in the graph for each pair of nodes, the algorithm can check for all the $O(n^2)$ paths maintained in each level and decide which one should be extended. Given the presence of multiple SPs in our setting, we cannot afford to look at each tuple in the \system.

\subsection{The \FFD Algorithm} \label{sec:ffdalgo}
We now give an overall summary of the \FFD algorithm. A detailed description of its sub-procedures can be found in Appendix \ref{sec:algo}.

Algorithm \ffullydynamic is similar to our basic fully dynamic algorithm and its overall description is given in Algorithm \ref{algo:ffully-main}.
The main difference is the introduction of the notion of levels as described in Section \ref{subsec:fully-impl}, and their activation/deactivation as in \Thoe.
At the beginning of the $t$-th update (with $k=level(t)$), we first activate the new level $k$ and we perform \ffullyupdate (Alg. \ref{algo:ffupdate}) on the updated node $v$. As in our basic algorithm and shown in Table II - Part C, the set $\NODES$ consists of all vertices centered at these lower deactivated levels. All vertices in $\NODES$  are re-centered at level $k$ during the $t$-th update (Alg. \ref{algo:ffully-main}, Step \ref{fully-main:new2}), 
and `dummy' update operations are performed on each of these vertices. Note that 
$\NODES$ contains the $2^{k}-1$ most recently updated vertices in reverse order of update time (from the most recent to the oldest).
Procedure \ffullyupdate is invoked with the parameter $k$ representing the newly activated level.
Finally, all levels $j < k $  are deactivated (Alg. \ref{algo:ffully-main}, Step \ref{fully-main:new1}).

\begin{algorithm}
	\SetAlgoLined
	\SetAlgoNoEnd
	%\KwData{}
	%\KwResult{}
	activate the new level $k$\;
	\ffullyupdateend($v,\weight',k$)\;
	generate the set $\NODES$\; 
	\For {each $u \in \NODES$ in decreasing order of update time}{
		\ffullyupdateend$(u,\weight',k)$ \{dummy updates\}\; \label{fully-main:new2}
	}
	deactivate all levels lower than $k$\; \label{fully-main:new1}
	\caption{\ffullydynamicend($G, v, \weight', k$)}
	\label{algo:ffully-main}
\end{algorithm}

\noindent
\paragraph{\bf \ffullyupdateend}
As in \DIe, \NPR and our basic algorithm, the update of a node occurs in a sequence of two steps: a \emph{cleanup phase} and a \emph{fixup phase}. 
Both \ffullycleanup and \ffullyfixup are more involved algorithms than their counterparts in our basic result, and the resulting algorithm will save a $O(\log n)$ factor over the amortized cost. Note that \ffullyfixup is called with the additional argument $k$ which indicates the current active level.

\begin{algorithm}
	\SetAlgoLined
	\SetAlgoNoEnd
	%\KwData{}
	%\KwResult{}
	\ffullycleanupend$(v)$\;
	\ffullyfixupend$(v,\weight',k)$\;
	\caption{\ffullyupdateend$(v,\weight',k)$}
	\label{algo:ffupdate}
\end{algorithm}

We now highlight some new components of our cleanup and fixup algorithms that will be helpful to prove the lemmas in the next section. 

\paragraph{New Components in \ffullycleanup relative to \fullycleanupend} 
\begin{enumerate} [label=FF--C.{\arabic*}]
	\item \label{fact:carray} For each triple $\gamma$ processed by \ffullycleanupend, the array $C_\gamma$ can be updated with the new centers in time $O(z')$, where $z'$ is the number of active levels that contains the updated vertex $v$.
	\item \label{fact:crn} Each triple requires constant time to be linked and unlinked from structures $\DMs$, $\RN$ and $\LN$. 
\end{enumerate}
\paragraph{New Components in \ffullyfixup relative to \fullyfixupend} 
\begin{enumerate} [label=FF--F.{\arabic*}]
	\item \label{fact:farray} For each triple $\gamma$ processed by \ffullyfixupend, the array $C_\gamma$ can be updated with the new centers in time $O(z')$, where $z'$ is the number of active levels that contain the triple~$\gamma$.
	\item \label{fact:frn} Each triple requires constant time to be linked and unlinked from structures $\DMs$, $\RN$ and $\LN$.
	\item \label{fact:fproc} \ffullyfixup only processes a triple $\gamma$ with $\beta = 1$ if it has a centered extensions in some active level younger than then the level in which $\gamma$ was shortest for the last time.  
\end{enumerate}

\noindent
We will establish in Section \ref{sec:proof} that \ffullydynamic
correctly updates the data structures with the amortized bound given in Theorem \ref{th:main}.
 
\section{Complexity of \FFD Algorithm} \label{sec:proof}
In this section we will prove the complexity bounds of our \FFD algorithm. The correctness is addressed in Appendix \ref{sec:proofcorr}.
The complexity analysis is similar to that for our basic algorithm. We highlight the following new elements:

\begin{enumerate}
	\item Every triple created by \ffullyfixup is an LST in the level graph (PDG) in which is centered (see Observation \ref{centerlht} in the Appendix), 
	and by the \decremental properties of level graphs, it will continue to be an LST in that level graph until it is removed. In contrast, our basic algorithm can create LHTs by combining HTs not centered in any PDG. This results in an additional $\Theta(\log n)$ factor in the amortized bound there. 
	
	\item  We can bound the number of LHTs that contain a given vertex $u$ as $O(z' \cdot \vstar^2)$, where $z'$ is the number of active level graphs that contain vertex $u$ and tuples passing through $u$ (by Corollary \ref{lemma:count1}). Given our level \systemend, $z'$ is clearly $O(\log n)$. In our basic result, this bound is $(z + z'^2)$ where $z$ is the number of active PDGs, and $z'$ is the number of PDGs that contain $v$.
		
	\item We can show that the number of accesses to $\RN$ and $\LN$, outside of the newly created tuples, is worst-case $O(n \cdot \nu^*)$ per call to \ffullyupdateend. The overhead given by the level data structures is $O(\log n)$ for each access (see Lemma \ref{lem:amor1}). These structures are not used in our basic algorithm.
\end{enumerate}

\noindent
All the algorithms referenced in the following lemmas are described in Appendix \ref{sec:algo}.

\begin{lemma} \label{lemma:count}
	Let $G$ be a graph after a sequence of calls to \ffullyupdateend. Let $z$ be the number of active level graphs (PDGs), and let $z' \leq z$ be the number of level graphs that contain a given vertex $v$. Suppose that every HT in the \system
	is an ST in some level graphs, and every LHT is an LST in some level graph. If $n$ and $m$ bound the number of vertices and edges,
	respectively, in any of these graphs, and if
	$\nu^*$ bounds the maximum number of distinct edges that lie on shortest paths through any given vertex in any of the these graphs, then:
	
	\begin{enumerate}
		\item The number of LHTs  in $G$'s \system is at most $O(z \cdot m \cdot \nu^*)$.
		
		\item The number of LHTs that contain a vertex $v$ in $G$ is $O( z' \cdot {\nu^*}^2)$.
	\end{enumerate}
\end{lemma}

\begin{proof}
	For part 1, we bound the number of LHTs $(xa,by)$ (across all weights) that can exist in
	$G$. The edge $(x,a)$ can be chosen in $m$ ways, and once we fix $(x,a)$, the 
	$r$-tuple $(a,by)$ must be an ST in one of the $\levelgraph_j$. Since $(b,y)$ must lie on  
	a shortest path through
	$a$ centered in a graph $\levelgraph_i$, that contains the $r$-tuple $(a,by)$ of shortest weight in $\levelgraph_i$, the number of different choices
	for $(b,y)$ that will then uniquely determine the tuple $(xa, by)$, together with its weight, is
	$z \cdot \nu^*$. Hence the number of LHTs in $G$'s \system is $O(z \cdot m \cdot \nu^*)$.
	
	For part 2, the number of LHTs that contain $v$ as an internal vertex is simply the number of
	LSTs across the $z'$ graphs that contains $v$, and this is  $O(z' \cdot \vstar^2)$. We now
	bound the number of LHTs $(va, by)$. There are $n-1$ choices for the edge $(v,a)$ and
	$z' \cdot \nu^*$ choices for the $r$-tuple $(a, by)$, hence the total number of such tuples 
	is $O(z' \cdot n \cdot \nu^*)$. The same bound holds for LHTs of the form $(xa, bv)$. Since
	$\nu^* = \Omega (n)$, the result in part 2 follows.
\end{proof}
 
  \begin{corollary} \label{lemma:count1}
  	At a given time step, let $B$ be the maximum number of tuples in the \system containing a path through
  	a given vertex in a given level graph.
	Then, $B=O(\vstar^2)$.
  \end{corollary}
  
 \begin{lemma}
 	\label{lem:amor}
 	(a) - The cost for an \ffullycleanup call on a node $v$ when $z'$ active levels contain triples through $v$  is $O(z' \cdot \vstar^2 \cdot \log n)$.\\
 	(b) - The cost for a real \ffullycleanup call is $O(\vstar^2 \cdot \log^2 n)$\\
 	(c) - The cost for a dummy \ffullycleanup call is $O(\vstar^2 \cdot \log n)$.
 \end{lemma} 
 \begin{proof}
 	(a) - Since the number of \LHTe s containing the updated vertex $v$, processed by \ffullycleanupend, is bounded by $B$ at each level (by Corollary \ref{lemma:count1}), the total cost is $O(z' \cdot B \log n)$ where $z'$ is the number of active levels that contain triples through $v$.
	The worst-case cost for update an array $C_\gamma$ within an \ffullycleanup phase is $O(z')$ (point [\ref{fact:carray}]).
 	Note that a triple can be processed by a constant number of priority queues among $z'$ different active levels. Moreover, for the structures $\DMs$, $\RN$ and $\LN$ each triple spends a constant time to be unlinked and eventually to update the structures (point [\ref{fact:crn}]). Since, priority queue operations have a $O(\log n)$ cost and the number of triples examined is bounded by $O(z' \cdot \vstar^2)$, the complexity of \ffullycleanup that operates on $z'$ active levels requires at most $O(z' \cdot \vstar^2 \cdot \log n)$.\\
 	(b) - Since the active levels are bounded by $z \leq \log 2n$, the cost for a real \ffullycleanup call is $O(\vstar^2 \cdot \log^2 n)$ (by part (a)).\\
 	(c) - For a dummy cleanup on a vertex $w$, \ffullycleanup only needs to clean the local data structures in level $center(w)$, where $w$ is centered, and in the current level graph. In fact, let $t$ be the current update step; in the dummy cleanup phase, we start with the node $u$  that was updated at time $t-1$ (the most recent update before the current one). The node $u$ received an update in the previous phase, thus it disappeared from all the levels older than $level(t-1)$ and, with it, all the LSTs containing $u$ in these levels. Hence, all the triples containing $u$ in the \system must be LSTs in $level(t-1)$. We have at most $B$ of them and \ffullycleanup spends $O(B \cdot \log n)$ (considering the access to the data structures) to remove them.   
 	Then, the dummy update reinserts $u$ only in the current graph. The next phase moves on the node $u'$ updated at time $t-2$. Again, all the tuples containing $u'$ must be LSTs in $level(t-2)$ and eventually the current graph if they were inserted because of the previous dummy update on $u$.
 	
 	Suppose in fact that there is a tuple $\gamma$ that contains $u'$ in another level (except the current graph). The tuple $\gamma$ cannot be in a level older than $level(t-2)$ because when $u'$ was updated at time $(t-2)$, the cleanup algorithm removed all the tuples containing $u'$ from any level older than $t-2$. Moreover, a tuple containing $u'$ present in a level younger than $level(t-2)$ could appear if and only if it was generated by any update more recent of $t-2$ (in this case only the dummy update on $u$ performed in the current graph). Thus a contradiction.
 	
 	This argument can be recursively applied to  every other node in the sequence: in fact for the node $u''$ updated at time $(t-i)$ all the nodes updated in the interval $[t-i+1,t-1]$ will be already processed by \ffullycleanupend, leaving all the tuples containing $u''$ only in $level(t-i)$ and $t$. It follows that, for a dummy update, $z' = 2$.
 	Thus the cost for a dummy \ffullycleanup call is $O(\vstar^2 \cdot \log n)$ (by part (a)).
 \end{proof}
 
  \begin{lemma}
  	\label{lem:amor1d}
  	The cost for a dummy \ffullyfixup call on a node $v$ is $O(\vstar^2 \cdot \log n)$. 
  \end{lemma} 
  \begin{proof}
  	Consider a dummy \ffullyfixup applied to a vertex $v$ in $\NODES$. We only need
  	to bound the cost for accessing the entries in the $P^*(x,y)$ and the cost of re-adding LSTs containing $v$, previously removed by the dummy \ffullycleanup but still in the current graph after the dummy update.
  	In fact the vertex $v$ is removed by an earlier dummy \ffullycleanupend, and while this removes all the HPs containing the vertex $v$, it does not change any LST centered in any $\levelgraph_j$ that does not contain $v$. Hence these other
  	LSTs will be present in the \system with unchanged weight and count, when dummy \ffullyfixup is applied to $v$.
  	Since for any pair $x,y$, the SP distance will not change after the dummy update, the dummy \ffullyfixup will only insert in the set $S$ triples containing the node $v$ for additional extension.
  	Hence, only the LSTs containing $v$ in the current $level(t)$ graph will be processed and added to the \systemend, and there are at most $B$ of them (by Corollary \ref{lemma:count1}). Thus a dummy \ffullyfixup for any $v$ needs to access $P^*$ for each pair of nodes, and reinsert at most $B$ tuples (containing $v$) in the current graph. Hence the overall complexity for a dummy \ffullyfixup is  $O((n^2 + B) \cdot \log n) = O(\vstar^2 \cdot \log n)$.
  \end{proof}

  We now address the complexity of a real \ffullyfixup call. We first define the concept of a \emph{triple pair} that will be used in lemma \ref{lem:amor1}  to establish the bound for a real \ffullyfixup call.
  Finally, we complete our analysis by presenting a proof of Theorem~\ref{th:main}. 
 
  \begin{definition}
  	If $C_{\gamma}[i] \geq 1$ then $(\gamma, i)$ is a \emph{triple pair} in the \systemend.
  	If $(\gamma,i)$ is not a triple pair in the \system at the start of step $t$ but is a triple pair after the update at time step $t$, then $(\gamma, i)$ is  a \emph{newly created triple pair}  at time step $t$.
  	\end{definition}
  
   \begin{lemma} \label{lemma:count2}
   	At a given time step,
   	let $D$ be
   	the number of triple pairs in the level \systemend.
   	Then,
   	\begin{enumerate}
   		\item The value of $D$ is at most $O(m \cdot \nu^* \cdot \log n)$.
   		\item  The space used is $O(m \cdot \nu^* \cdot \log n)$.
   	\end{enumerate}
   \end{lemma}
   
   \begin{proof}
   	
   	1. Every $C_{\gamma}[i] \geq 1$ represents a distinct LST in $\Gamma_i$, hence the result follows since the number of levels is $O(\log n)$ and the
   	number of LSTs in a graph is $O(\vstar \cdot m^*)$.
   	
   	\noindent
   	2. Since every triple is of size $O(1)$, the memory used by our \FFD algorithm is dominated by $D$, and result follows from 1.
   \end{proof}	 
 
 \begin{lemma}  
 	\label{lem:amor1}
 	The cost for a real \ffullyfixup call is $O  (\vstar^2 \cdot \log^2 n + X  \cdot \log n)$ , where $X$ is the number of newly created triple pairs after the update step.
 \end{lemma}
 \begin{proof}
  Recall that, for the structures $\DMs$, $\RN$ and $\LN$ each triple spends a constant time to be unlinked and eventually to update the structures (fact[\ref{fact:frn}]); moreover, updating an array $C_\gamma$ with the new centers requires only additional $O(z'\leq 2\log n)$ time (fact[\ref{fact:farray}]).
  Thus, a triple is accessed only a constant number of time during \ffullyfixup with a total cost of $O(\log n)$, and it suffices to establish that the
  number of existing triples accessed during the call is $O(\vstar^2 \cdot \log n)$.
  
   There are only $O(n^2)$ accesses to triples to initialize \ffullyfixup since 
 	$O(n^2)$ entries in the global $P^*(x,y)$ structures are accessed to populate $H_f$ (a shortest triple for each pair $(x,y)$). This takes $O(n^2 \cdot \log n)$ time after considering
	the $O(\log n)$ cost per data structure operation. We now address the accesses made in the main loop. We will distinguish two cases and they will be charged to $X$ as follows.
	
	\vhalf
	\noindent
	\underline{1: $\beta(\gamma)=0$} -- Any triple $\gamma$ that is accessed with $\beta (\gamma) = 0$ is an
	 LST at some level $i$ where it is not identified as an ST in $\Gamma_i$. 
	 In this case, if the distance for the endpoints of $\gamma$ did not change, $\gamma$ is added as an ST in level $i$, and will never be removed as an ST for level $i$ until it is removed from the tuple system
	 (due to the fact that $\Gamma_i$ is a purely \decremental graph). Since $\gamma$ with $\beta (\gamma) =0$ is a newly added triple to level $i$, then the pair $(\gamma,i)$ is a newly created triple pair at step $t$. Hence, we can charge $(\gamma,i)$ to $X$ in this call of \ffullyfixupend.
	
	\vhalf
	\noindent
	\underline{2: $\beta(\gamma)=1$} -- We now consider triples accessed that have $\beta = 1$. This is the most nontrivial part of our analysis since even though any
	 such triple $\gamma$ must exist with the same count in every level in both $P$ and $P^*$, we may still need to form some extensions
	 since the triple may have been an HT when extension vertices were updated, and hence these extension may not have been
	 performed. Here is where the $LN$ and $RN$ sets are accessed, and we now analyze the cost of these accesses. (The
	 correctness of the associated steps is analyzed in the Appendix.)
	 
	 Let $j$ be the most recent level in which $\gamma$ was an ST in $G$ and assume we are dealing with left extensions (right extensions are symmetrical). Now that $\gamma$ is restored, the only case in which we need to process it is when there exists a left extension for the $\ell$-tuple of $\gamma$ to a node $x'$ centered in a level $i$ more recent than $j$. In fact, the LST generated by this extension will appear for the first time centered in level $i$, hence the pair $(\gamma,i)$ is a newly created triple pair at step $t$ and we can charge its creation to $X$.
	 We now show how our HE sets efficiently handle this case.
     \ffullyfixup only processes a restored triple $\gamma$ with $\beta(\gamma)=1$ when it has at least one centered extension in some active level younger than the level in which $\gamma$ was shortest for the last time (fact[\ref{fact:fproc}]). We can bound the total computation for these steps as follows: for a given $x$, 
	$RN(x,y,wt)$ contains a node $b$ for every incoming edge to $y$ in one of the SSSP dags (historical and shortest) rooted at $x$. 
	Since we can extend in at most $O(\log n)$ active levels during any update and the size of a single dag is at most $\nu^*$, these steps take time $O(\nu^*\cdot n\log n)$ throughout the entire update computation. 
 	 \end{proof}	 

We can now establish the proof of our main theorem. 
	 
\noindent
{\bf Proof of Theorem~\ref{th:main}.}
Consider a sequence $\Sigma$ of $r= \Omega (n)$ calls to algorithm \FFDe. Recall that the data structure is reconstructed after every $2n$ steps,
so we can assume $r=\Theta (n)$. These $r$ calls to \FFD make $r$ real calls to \ffullyupdateend, and also make additional dummy updates.
As in our basic algorithm,
 	across the $r$ real updates in $\Sigma$,
 	the algorithm performs $O(r \log n)$ dummy updates.
 	This is because $r/2^k$ real updates are performed at level $k$ during the entire computation, and each such update is
 	accompanied by $2^k-1$ dummy updates. So, across all real updates there are $O(r)$
 	dummy updates per level, adding up to $O(r \log n)$ in total, across the $O(\log n)$ levels.
	
	When \ffullycleanup is called on a vertex $v$ for a dummy update, $z'=2$ since $v$ can be present 
	 only in the most recent current level and the level at which it is centered. (This is because every vertex that was centered
	 at a more recent level than $v$ has already been subjected to a dummy update, and hence all of these vertices are now
	 centered in the current level.)
 	Thus, by Lemma \ref{lem:amor}, each \ffullycleanup for a dummy update has cost $O(B \cdot \log n)$. By Lemma \ref{lem:amor1d}, a call to  \ffullyfixup for a dummy update
 	has cost $O(\vstar^2 \cdot \log n)$. Thus the total cost is $O((\vstar^2 \cdot \log n)\cdot r\log n)$ across all dummy updates.
	Also, the number of tuples accessed by all of the dummy update calls to \ffullycleanupend, and hence the number of tuples
	removed by all dummy updates, is $O(r \cdot \vstar^2 \cdot \log n)$.
	
	 For the real calls to \ffullyfixupend, let $X_i$ be the   number of newly added triple pairs in the $i$th real call to \ffullyfixupend. Then by
	Lemma~\ref{lem:amor1}, the cost of this $i$th call is $O(\vstar^2 \cdot \log^2 n + X_i \cdot  \log n)$. Let $X= \sum_{i=1}^r X_i$. Hence the total
	cost for the $r$ real calls to  \ffullyfixup is $O(r \cdot \vstar^2 \cdot \log^2 n + X \cdot \log n)$. We now bound $X$ as follows: $X$ is no more than
the maximum number of triples that can remain in the system after $\Sigma$ is executed, plus the number of tuples $Y$ removed
	from the \systemend. 	Tuples are removed only in calls to \ffullycleanupend. The total number removed by $r \log n$ dummy calls is 
	$O(r \cdot  \log n \cdot \vstar^2)$ (by Lemma \ref{lem:amor}). The total number removed by the $r$ real calls is $O(r \cdot \vstar^2 \cdot \log n)$ (by Lemma \ref{lem:amor}).
	Hence $Y= O(r \cdot \vstar^2 \cdot \log n)$. Clearly the maximum number of triples in the \system is no more than $D$, which
	counts the number of triple pairs, and we have 
	$D= O(m \cdot \vstar \cdot \log n) = O(n^2 \cdot \vstar \cdot \log n)$ (by Lemma \ref{lemma:count2}). Since  $r=\Theta (n)$, we have
	$D=O(r \cdot n \cdot \vstar \cdot \log n)$, and this is dominated by $Y$ since $\vstar = \Omega (n)$.
	Hence the cost of the $r$ calls to \FFD is $O(r \cdot \vstar^2 \cdot \log^2 n)$ (after factoring in the $O(\log n)$ cost per tuple access), and 
		hence the amortized cost of  each call to \ffullyupdate is $O(\vstar^2 \cdot \log^2 n)$.

\section{Conclusion}\label{sec:extn}

%VLR: Need to re-word the first two paragraphs slightly and remove third paragraph.
We have presented efficient fully dynamic algorithms for APASP.
Our algorithms stores a superset of the STs and LSTs in the current graph in two priority queues $P^*(x,y)$ and
$P(x,y)$ for each vertex pair $x,y$. To generate all shortest paths from all other vertices to $x$, we construct the shortest
path in-dag rooted at $x$ in $O(\nu^*)$ time. Then, the shortest paths ending in $x$ can be enumerated in a traversal
of this dag, starting from $x$. This query time is output-optimal, and takes time proportional to the number of
edges on these paths.

Our basic algorithm, when specialized to fully dynamic APSP (i.e., for unique SPs), is a variant of the 
\DI method \cite{DI04},
and it uses a different `dummy update sequence' from the one in~\DIe, with different properties. 
Our dummy update
sequence is inspired by the updates performed on `level graphs' in~\cite{Thorup04}, though our algorithm
is considerably simpler (but is also slower by a logarithmic factor). As noted in Section~\ref{sec:pdg}, our
analysis is tailored to the dummy update sequence we use, and a different analysis would be needed if
the \DI update sequence is to be used.

Our improved fully dynamic algorithm for APASP runs in amortized $O({ \nu^*}^2 \cdot \log^2n)$ time, which is a log factor faster than the basic result. This algorithm is considerably more involved and adapts the \Tho method~\cite{Thorup04} to the APASP problem by maintaining the PDGs explicitly. An additional complexity in this 
fully dynamic APASP algorithm
(beyond that present in~\Thoe) is the need to maintain several different time-stamps for each tuple in the
data structures, since the component paths in a tuple may have been updated at different time steps.

It would be interesting to investigate if one could
stay with our method here of only performing dummy updates and not maintaining the PDGs explicitly, but still
obtain the improved bound in our \FFD algorithm (and in \Tho for unique shortest paths). This would give a  reasonably simple fully dynamic APASP algorithm
 and it would lead to a 
 simpler  fully dynamic
 APSP algorithm with $O(n^2 \cdot \log^2n)$
 amortized time than \Tho algorithm.
 One appealing approach is to only form LSTs in the current
 graph during the fixup phase (instead of forming all LHTs). It is not difficult to see that this would reduce the cost of the cleanup phase by a
 logarithmic factor. However, if an HT $\tau$ becomes an ST at a later step, we would then 
 have no guarantee that all of its extensions have been generated. Hence, if this approach is
 to succeed, a modified algorithm is needed.\\
 
 \noindent
\textbf{Acknowledgement.} We thank Russell F. McQueeney for working on implementing the \fullydynamic algorithm described in this paper. His implementation work provided useful insights that led to this improved and more readable write-up and pseudocode.

% Appendix
\newpage
\appendix
\section*{APPENDIX}
\section{An example for the necessity of \texorpdfstring{$\CH$}{CH}} \label{append:example}
Here, we present an example (inspired by \cite{DI04}) which shows why the $\CH$ data structure is needed for the complexity and correctness of the algorithm. The graph $G=(V,E)$ in Figure \ref{fig:prof1} is a directed graph (from left to right), where the nodes in layers $L_1$ and $L_2$ form an oriented complete bipartite graph, and similarly the nodes in layers $L_3$ and $L_4$.
When referring to a generic layer $L_i$ we consider any possible node in $L_i$.
All the edges have weight 1 unless specified. The height of the structure is $\Theta(n)$ (where $n=|V|$) as each level requires only a constant number of nodes. 
Initially, the node $c$ (and its incident edges) and the edge $(u,v)$ are not part of the graph. A `bad' update sequence, when $\CH$ is not implemented, is the following:
\begin{enumerate}
    \item add node $c$ to $G$ and its incident edges with the corresponding weights indicated in the figure;
    \item add edge $(u,v)=1$ to $G$;
    \item remove $c$ from $G$. Note that the paths of the form $(L_2a_i,b_iL_3)$ (for $1 \leq i \leq h$) are also removed from $P^*$;
    \item remove $(u,v)$ from $G$. Note that, the graph $G$ is back in its initial configuration.
\end{enumerate}

In Step 1, additional LSPs of the form $(L_1L_2,L_3L_4)$ are added to the graph via the new node $c$. In Step 2, all the shortest paths going through the $(L_2,L_3)$ bridges are replaced with shortest paths going through the new edge $(u,v)$. In Step3, $c$ is removed: this deletes all the paths in the $(L_1L_2,L_3L_4)$ historical tuples. In Step 4, $(u,v)$ is removed and $\Theta(n^3)$ tuples need to be recreated making the process computationally too expensive for a sequence of 4 updates.
%VLR reword
Thus, if we did not retain THTs in the graph $G$ this example shows that we will have amortized cost $\Omega(n^3)$, since all LSTs are being recreated by \fullyfixup during update (step 4 above). On the other hand, if we retain  THTs we will need to maintain their correct counts between updates. This is efficiently achieved using our $\CH$ data structure presented in Section \ref{sec:ch}.

   \begin{figure}[!ht]
    \begin{center}
    \makebox[\textwidth]{\includegraphics[width=0.75\paperwidth]{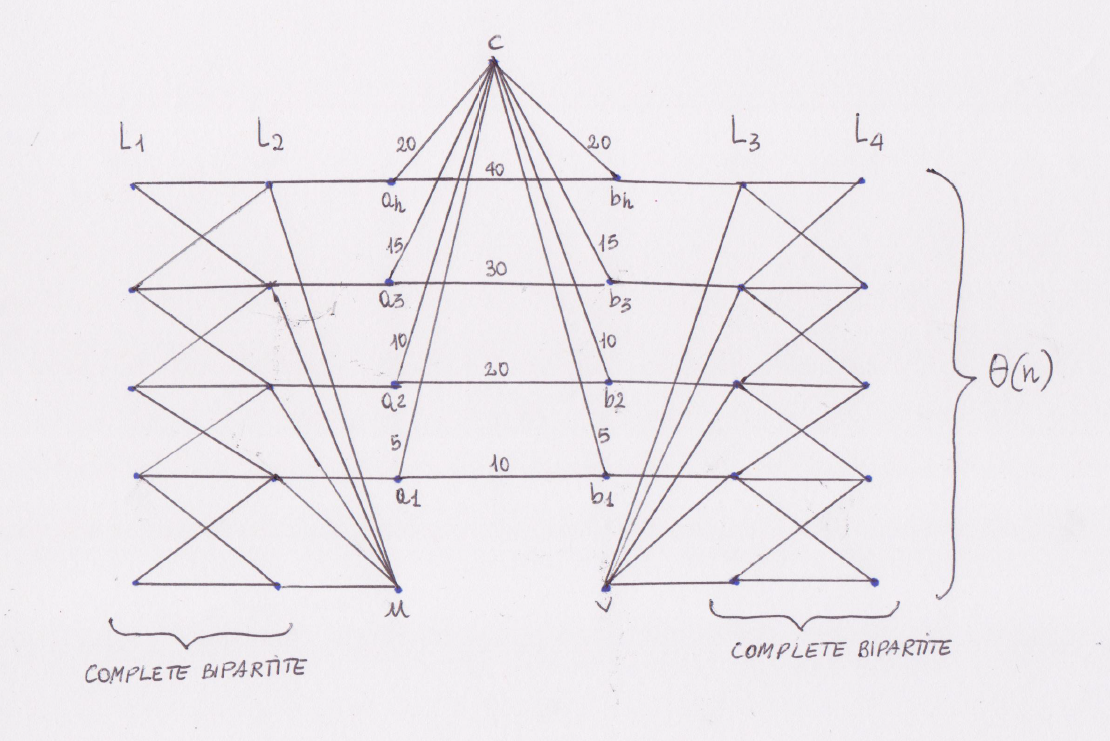}}
    \caption{\label{fig:prof1}An example for the necessity of $\CH$. The graph is oriented from left to right. Edges have unitary weight unless specified.}
    \end{center}
    \end{figure}
    
\section{Correctness of \fullydynamicend} \label{append:corr}
In this section, we establish the correctness of Algorithm \ref{algo:fcleanup} and \ref{algo:ffixup}.

We assume that all the local structures are correct before the update, and we will show the correctness of them after the update.
\fullycleanup works with a heap $H_c$ of triples. The algorithm maintains the loop invariant that any triple inserted into $H_c$ has already been deleted from the \systemend: only its extensions remain to be processed and deleted.
We prove the following lemma:

\begin{lemma} \label{lem:fcleanup-corr2}
	After Algorithm~\ref{algo:fcleanup} (\fullycleanupend)
	is executed, for any $(x, y) \in V$, the STs
	in $P^{*}(x,y)$ (LSTs in $P(x,y)$) represent all the SPs (LSPs) from $x$ to $y$ in $G$ that do not pass through $v$. Moreover, every THT (TLHT) present in the \system represents a collection of HPs (LHPs) in $G$ that contains only paths that do not pass through $v$.
	Finally, the sets $L, L^{*}, R, R^{*}$ are correctly maintained.
\end{lemma}
\begin{proof}
	The lemma is established with the following loop invariant.
	
	{\noindent \bf Loop Invariant:} At the start of each iteration of the while loop in Step~\ref{cleanup:while} of Algorithm~\ref{algo:fcleanup} the following properties hold about the \system and $H_c$. Assume that the first triple to be extracted from $H_c$ and processed has min-key = $[wt, x, y]$.
	\begin{enumerate}
		\item \label{proof:item1}
		Any LHP going through the update node $v$, which is contained in a triple $\gamma$ already processed, is removed from the \systemend. 
		For any $a, b \in V$, if $G$ contains $c_{ab}$ LHPs of weight ${wt}$ of the form $(xa, by)$
		passing through $v$,
		then $H_c$ contains a triple $\gamma = ((xa, by), { wt}, c_{ab})$ with key $[wt,x,y]$ already processed: the $c_{ab}$ LHPs through $v$ are cleaned from the system. 
		\item \label{proof:claim2}
		For each triple $\gamma=((xa,by),wt, count)$ already processed by the algorithm, $\gamma$ represents the exact set of LHPs of the form $x \rightarrow a \rightsquigarrow b \rightarrow y$ of weight $wt$ which avoid $v$. Moreover if $count=0$ the triple $\gamma$ is correctly removed from all the data structures in the \system associated to it. Finally, $\gamma$ has been placed in $H_c$ for future extensions.
		
		Let  $[\hat{wt},\hat{x},\hat{y}]$ be the last key extracted from $H_c$ and processed before $[wt,x,y]$. For any key $[wt_1, x_1, y_1] \leq [\hat{wt},\hat{x},\hat{y}]$, let 
		$G$ contain ${ c}  > 0 $ number of LHPs of weight  ${ wt_1}$ of the
		form $(x_1 \times, b_1y_1)$. Further, let ${ c_v}$ (resp. ${ c_{\bar v}}$) denote the number of such LHPs
		that pass through $v$ (resp. do not pass through $v$).
		Here ${ c_v + c_{\bar v} = c}$. For every extension $(x',wt') \in L(x_1,b_1y_1)$, let let $wt' = wt_1 + \weight (x',x_1)$ be the weight of the extended triple $(x'x_1, b_1y_1)$. Then, (the following assertions are similar for $(y',wt') \in R(x_1a_1, y_1)$)
		\begin{enumerate}
		\item \label{proof:item2} if $c>c_{v}$ there is a triple in $P(x', y_1)$ 
		of the form $(x'x_1, b_1y_1)$ and weight $wt'$ representing $c-c_v$ LHPs. If $c=c_v$ there is no such triple in $P(x', y_1)$.
		\item \label{proof:item21}
		If a triple of the form $(x'x_1, b_1y_1)$ and weight $wt'$ is present as an HT in $P^*(x',y_1)$, then it represents the exact same number of LHPs $c-c_v$ of the corresponding triple in $P(x',y_1)$. This is exactly the number of HPs of the form $(x'x_1, b_1y_1)$ and weight $wt'$ in $G - \{v\}$.
		\item \label{proof:item3} $(x',wt') \in L(x_1, b_1y_1)$,  $(y_1,wt') \in R(x'x_1, b_1)$ and    $(x'x_1, b_1y_1) \in $ \MT
		iff ${ {c_{\bar v}} > 0}$. If the triple $(x'x_1, b_1y_1)$ is an HT, a similar statement holds for $(x',wt') \in L^*(x_1,y_1)$ iff there is a triple of weight $wt_1$ in $P^*(x_1,y_1)$, and $(y_1,wt') \in R^*(x',b_1)$  iff there is a triple of weight $wt' - \weight(b_1, y_1)$ in $P^*(x',b_1)$.
		\item \label{proof:item4} A triple corresponding to $(x' x_1, b_1y_1)$ with weight $wt'$ and counts $c_v$ is in $H_c$.
		A similar assertion holds for $y' \in R(x_1a_1, y_1)$.
%		\item \label{proof:item6} The \system does not contain any TLHT in $P(x',y_1)$ and THT in $P^*(x',y_1)$ of the form $(x'x_1, b_1y_1)$ and weight $wt'$.
		\end{enumerate}
		\item \label{proof:item5}
		Any triple $\gamma$  in $H_c$, which is longer than the last triple processed, is also in \MT if and only if it contains at least one path not passing through $v$.  
		For any key $[wt_2, x_2, y_2 ] \geq [wt, x, y]$, let
		$G$ contain $c > 0$ LHPs of weight  ${ wt_2}$ of the
		form $(x_2a_2, b_2y_2)$. 
		Further, let ${ c_v}$ (resp. ${ c_{\bar v}}$) denote the number of such LHPs
		that pass through $v$ (resp. do not pass through $v$).
		Here ${ c_v + c_{\bar v} = c}$. 
		Then the tuple $(x_2 a_2, b_2 y_2) \in$ \MTend, iff $c_{\bar v} > 0$ and a triple for
		$(x_2a_2, b_2 y_2)$ is present in $H_c$
	\end{enumerate}
	
\noindent {\bf Initialization:} We start by showing that the invariants hold before the first loop iteration.
The min-key triple in $H_c$ has key $[0, v, v]$. Invariant assertion~$\ref{proof:item1}$
holds since  we inserted into
$H_c$ the trivial triple of weight $0$ corresponding to the vertex $v$
and that is the only triple of such key. Moreover, since we do not represent trivial
paths containing the single vertex, no counts need to be decremented.
Since we assume positive edge weights, there are no LHPs
in $G$ of weight less than zero. Thus all the points of invariant assertion~$\ref{proof:claim2}$ hold trivially.
Invariant assertion~$\ref{proof:item5}$ holds since $H_c$ does not contain any triple of weight $> 0$ and we initialized \MT to empty.

\noindent {\bf Maintenance:} Assume that the invariants are true before an iteration $k$ of the loop.
We prove that the invariant assertions remains true before the next iteration $k+1$.
Let the min-key triple at the beginning of the $k$-th iteration be $[wt_k, x_k, y_k]$.
By invariant assertion~$\ref{proof:item1}$, we know that for any $a_i, b_j$, if there exists a triple $\gamma$ of the form $(x_k a_i, b_j y_k)$
of weight $wt_k$ representing $count$ paths going through $v$, then it is present in $H_c$.
Now consider the set of triples with key $[wt_k, x_k, y_k]$
which we extract in the set $S$ (Step~\ref{cleanup:extract-set}, Algorithm~\ref{algo:fcleanup}).
We consider left-extensions of triples in $S$; symmetric arguments apply for right-extensions.
Consider for a particular~$b$, the set $S_{b} \subseteq S$ of triples of the form $(x_k \times, by_k)$  
and let $fcount'$ denote the sum of the counts of the paths represented by triples in $S_b$.
Let $(x',wt') \in L(x_k, by_k)$ be a left extension; our goal is to generate the triple $(x'x_k, b y_k)$ with
count $fcount'$ and weight $wt' = wt_k+\weight(x',x_k)$.
However, we generate such triple
only if it has not been generated by a right-extension of another set of paths.
We observe that the paths
of the form $(x'x_k, by_k)$ can be generated by right extending to $y_k$ the set of
triples of the form $(x'x_k, \times b)$. Without loss of generality assume that the triples of
the form $(x'x_k, \times b)$ have a key which is greater than the key $[wt_k, x_k, y_k]$. Thus,
at the beginning of the $k$-th iteration, by invariant assertion~$\ref{proof:item5}$, we know that $(x'x_k, by_k) \notin$ \MTend.
Steps~\ref{process-cleanup:create-left}--\ref{process-cleanup:left-extend-removeP},
Algorithm~\ref{algo:fcleanup} create a triple of the form $(x'x_k, by_k)$ of weight $wt'$.
The generated triple can be an LST or a TLHT in $P$. In both cases the condition at step \ref{process-cleanup:left-extend-removeP}, Algorithm~\ref{algo:fcleanup} holds and we remove $\gamma'$ by decrementing $fcount'$ many paths from the appropriate triple in $P(x',y_k)$. 
%Moreover, if the condition does not hold, we just encountered a TLHT: in this case we completely remove this TLHT from $P$ in step \ref{process-cleanup:new2}, Algorithm~\ref{algo:fcleanup}.
Moreover, if the generated triple is also contained in $P^*(x',y_k)$, we check if it is an ST or a THT using step \ref{process-cleanup:new3}, Algorithm~\ref{algo:fcleanup}. In the case of an ST we normally decrement $fcount'$ paths from the appropriate triple in $P^*(x',y_k)$, otherwise we decrement the THT from $P^*(x',y_k)$ in step \ref{process-cleanup:new4},  Algorithm~\ref{algo:fcleanup} removing only the paths going through the updated node $v$ that are contained in the THT, by using the new data structure $\CH$.
This establishes invariant assertions~$\ref{proof:item2}$ and $\ref{proof:item21}$.
In addition, if there are no
LSPs in $G$ of the form $(x'x_k, by_k)$ which do not pass through $v$,
we delete $(x',wt')$ from $L(x_k, by_k)$ and delete $(y_k,wt')$ from $R(x'x_k, b)$ (Step~\ref{process-cleanup:left-extend-removeL}, Algorithm~\ref{algo:fcleanup}).
On the other hand, if there exist LSPs in $G$ of the form $(x'x_k, by_k)$,
then $x'$ (resp. $y_k$) continues to exist in $L(x_k, by_k)$ (resp. in $R(x'x_k, b)$).
Further, we add the tuple $(x'x_k, by_k)$ to \MT and observe that the corresponding triple is already present in $H_c$ (Step~\ref{process-cleanup:left-extend-markR}, Algorithm~\ref{algo:fcleanup}).
Similarly, if the generated triple $(x'x_k, by_k)$ is an HT, then we check if $P^*(x_k,y_k)$ does not contain any triple of weight $wt_k$, and similarly $P^*(x',b)$ does not contain any triple of weight $wt' - 
\weight(b,y_k)$, in order to delete $(x',wt')$ from $L^*(x_k,y_k)$ and $(y_k,wt')$ from $R^*(x',b)$.
By the loop invariant, invariant assertions~$\ref{proof:item3}$ and $\ref{proof:item4}$ were true for every key  $ < [wt_k, x_k, y_k]$
and by the above steps we ensure that these invariant assertions hold for every key = $[wt_k, x_k, y_k]$.
Thus, invariant assertions~$\ref{proof:item3}$ and $\ref{proof:item4}$ are true at the beginning of the $(k+1)$-th iteration.
Note that any triple that is generated by a left extension (or symmetrically
right extension) is inserted into $H_c$ as well as into \MTend. This establishes invariant assertion~$\ref{proof:item5}$ at the beginning of the $(k+1)$-th iteration.

Finally, to see that invariant assertion~$\ref{proof:item1}$ holds at the beginning of the $(k+1)$-th iteration, let the
min-key at the $(k+1)$-th iteration be $[wt_{k+1}, x_{k+1}, y_{k+1}]$. Observe that triples
with weight $wt_{k+1}$ starting with $x_{k+1}$ and ending in $y_{k+1}$ can be created
either by left extending or right extending the triples of smaller weight. And since for each of
iteration $\le k$, invariant assertion~$\ref{proof:claim2}$ holds for any extension, we conclude that invariant assertion~$\ref{proof:item1}$ holds at the beginning of the $(k+1)$-th iteration.
This finishes our maintenance step.

\noindent {\bf Termination:} The condition to exit the loop is $H_c = \emptyset$. Because invariant assertion $\ref{proof:item1}$ maintains in $H_c$ all the triples already processed, then $H_c = \emptyset$ implies that there are no other triples to extend in the graph $G$ that contain the updated node $v$. Moreover, because of invariant assertion $\ref{proof:item1}$, every triple containing the node $v$ inserted into $H_c = \emptyset$, has been correctly decremented from the \systemend. 
%Also, since invariant assertion $\ref{proof:item6}$ holds, every THT and TLHT containing the node $v$ has been completely deleted from the \systemend. 
Finally, for invariant assertion $\ref{proof:item3}$, the stacks $L, L^{*}, R, R^{*}$ are correctly maintained. This completes the proof.
\end{proof}

For \fullyfixup, we first show that Algorithm \ref{algo:ffixup} computes the correct distances for all the SPs in the updated graph $G'$ (Lemma \ref{proof:fflem01}). Moreover, we process all the \emph{new} SPs in $G'$ (Lemma \ref{proof:fflem0}). Finally, we show that data structures and counts are correctly maintained after the algorithm (Lemma \ref{proof:fflem}).
Here we use the notion of a \emph{\freshend} LHT for a triple that represents at least one path that is in $P$ but not in $P^*$. We will consider \fresh triples in Lemma \ref{proof:fflem0} and Observation \ref{obs:fresh}.

\begin{invariant} \label{proof:ffinv01}
	During the execution of Algorithm \ref{algo:ffixup}, for any pair $(x,y)$, consider the first extraction from $H_f$ of a set of triples from $x$ to $y$, and let their weight be $wt$. Then $wt$ is the shortest path distance from $x$ to $y$ in the updated graph $G'$.
\end{invariant}

\begin{lemma} \label{proof:fflem01}
	Algorithm \ref{algo:ffixup} maintains Invariant \ref{proof:ffinv01}.
\end{lemma}

\begin{proof}
	Suppose for a contradiction that the invariant is violated at some extraction. Consider the earliest event in which the first set of triples $S'$ of weight $\hat{wt}$, extracted for some pair $(x,y)$, does not contain STs in $G'$. Let $\gamma = ((xa,by),wt,count)$ be a triple in $G'$ that represents at least one shortest path from $x$ to $y$ in $G'$, with $wt < \hat{wt}$. 
	The triple  cannot be in $P(x,y)$ at the beginning of fixup otherwise it (or another triple with same weight $wt$) would have been inserted in $H_f$ during step~\ref{fixup:init2} of Algorithm~\ref{algo:ffixup}. Moreover, $\gamma$ cannot be in $H_f$ otherwise it would have been extracted before any triple of weight $\hat{wt}$ in $S'$; hence $\gamma$ must be a \emph{\new} LST generated by the algorithm.
	Since all the edges incident to $v$ are added to $H_f$ during step~\ref{fixup:init1} of Algorithm~\ref{algo:ffixup}, then $\gamma$ must represent SPs of at least two edges. We define $left(\gamma)$ as the set of LSTs of the form $((xa,c_ib), wt-\weight(b,y), count_{c_i})$ that represent all the LSPs in the left tuple $((xa,b), wt-\weight(b,y))$; similarly we define $right(\gamma)$ as the set of LSTs of the form $((ad_j,by), wt-\weight(x,a), count_{d_j})$ that represent all the LSPs in the right tuple $((a,by), wt-\weight(x,a))$. 
	
	Observe that since $\gamma$ is an ST, all the LSTs in $left(\gamma)$ and $right(\gamma)$ are also STs. 
	A triple in $left(\gamma)$ and a triple in $right(\gamma)$ cannot be present in $P^*$ together at the beginning of fixup. 
	In fact, if at least one triple from both sets is present in $P^*$ at the beginning of fixup, then the last one inserted during the fixup triggered by the previous update, would have generated an LST of the form $((xa,by), wt)$ automatically inserted and thus present in $P$ at the beginning of fixup (a contradiction). Thus either there is no triple in $left(\gamma)$ in $P^*$, or there is no triple in $right(\gamma)$ in $P^*$.
	
	Assume w.l.o.g. that no triple in $right(\gamma)$ is in $P^*$. 
	Since edge weights are positive, $wt-\weight(x,a) < wt < \hat{wt}$, and because all the extractions before $\gamma$ were correct, then the triples in $right(\gamma)$ were correctly extracted from $H_f$ and placed in $P^*$ before the wrong extractions in $S'$. 
	If at least one triple in $left(\gamma)$ is in $P^*$ then the fixup would generate the tuple $((xa,by),wt)$ and place it in $P$ and $H_f$ (Steps~\ref{fixup:startleft}--\ref{fixup:endleft},  Algorithm~\ref{algo:ffixup}). Otherwise, since $wt-\weight(b,y) < wt < \hat{wt}$, the triples in $left(\gamma)$ were discovered by the algorithm before the wrong extractions in $S'$. Moreover the algorithm would generate the tuple $((xa,by),wt)$ (as right extensions) and place it in $P$ and $H_f$ (because at least one triple in $right(\gamma)$ is already in $P^*$). Thus, in both cases, a tuple $((xa,by),wt)$ should have been extracted from $H_f$ before any triple in $S'$. A contradiction.
\end{proof}

\begin{invariant} \label{proof:ffinv}
	The set $S$ of triples constructed in Steps~\ref{ffixup:main1}--\ref{ffixup:phase3-main-check-end} of Algorithm~\ref{algo:ffixup}
	represents all the {\em new} shortest paths from $x$ to $y$.
\end{invariant}

\begin{lemma} \label{proof:fflem0}
	Algorithm \ref{algo:ffixup} maintains Invariant \ref{proof:ffinv}.
\end{lemma}

\begin{proof}
	Any new SP from $x$ to $y$ is of the following three types:
	\begin{enumerate}
		\item a single edge containing the vertex $v$ (such a path is
		added to $P(x, y)$ and $H_f$ in Step~\ref{fixup:init1})
		\item a path generated via left/right extension of some shortest path previously extracted from $H_f$ during the execution of Algorithm~\ref{algo:ffixup}  (this generated path is added to $P(x, y)$ and $H_f$ in Step~\ref{fixup:add1} and an analogous step in right-extend).
		\item a path that was an LSP but not an SP before the update and is an SP after the update.
	\end{enumerate}
	In type (1) and (2) above
	any new SP from $x$ to $y$  which is added to $H_f$ is also added to $P^*(x, y)$. However,
	amongst the several triples representing paths of the type (3) listed above, only one candidate triple
	will be present in $H_f$. Thus we conclude that, for a given $x, y$, when we extract from $H_f$ a type (3) triple of weight $wt$, $P(x,y)$ could contain
	a superset of triples with the same weight $wt$ that are not present in $H_f$. We now consider the two cases the algorithm deals with.
	\begin{itemize}
		\item $P^{*}(x, y)$ increased its min-weight, when the first set of triples for $x, y$ is extracted from $H_f$. This is the only case where we could restore historical triples, or process \fresh triples from scratch because they are not yet in $P^*$ or they are present in $P^*$ with a lower count than the corresponding triple in $P$ (note that this condition is triggered only by \decremental updates). To do that we process all the min-weight triples in $P(x,y)$, but before we really add a triple in $S$ for further extensions, we check if it is present in $P^*$ with a lower count (Step~\ref{ffixup:new1}, Algorithm~\ref{algo:ffixup}), or it is not present in $P^*$ (Step~\ref{ffixup:old1}, Algorithm~\ref{algo:ffixup}). By the above argument, we consider all the new STs from $x$ to $y$ present in $P(x, y)$.
		Therefore it suffices to argue that all of them contains \emph{new} shortest paths to be processed. Suppose for contradiction that
		some triple $\gamma$ does not contains {\em new} shortest paths. Thus, $\gamma$ was a ST before
		the update and already in $P^*$ 
		with at least one path not going through $v$. However, since cleanup only removes paths that contain $v$,
		the triple $\gamma$ remained in $P^{*}(x,y)$ after the \fullycleanup phase. This contradicts the fact that $P^{*}(x,y)$
		increased its min-weight.
		\item $P^{*}(x,y)$ didn't change its min-weight when the first set of triples for $x, y$ is extracted from $H_f$.
		Let the weight of triples in $P^{*}(x,y)$ be $wt$.
		This implies that the shortest path distance from $x$ to $y$ before and after the update is $wt$.
		Both in the case of \incremental and \decremental updates, all the new paths that we need to consider from $x$ to $y$ are going through the updated node $v$. By construction of the Algorithm \ref{algo:ffixup}, every triple containing the updated node $v$ is always placed into $H_f$. Thus it suffices to consider only triples in $H_f$.
	\end{itemize}
\end{proof}

\begin{observation} \label{obs:fresh}
	During the execution of Algorithm \ref{algo:ffixup},
	consider a THT $\tau$ that becomes shortest. If $\tau$'s corresponding triple in $P$ is not \freshend, then it is simply restored (not processed); otherwise $\tau$'s count is replaced with the updated count from $P$ and it is extended anew.
\end{observation}

\begin{proof}
	When we restore an existing HT $\tau$ from $P^*$, we always check if its corresponding triple in $P$ contains more paths (Step~\ref{ffixup:new1}, Algorithm~\ref{algo:ffixup}) or the counts match. In the first case $\tau$ in $P^*$ is carrying an obsolete number of SPs and is therefore replaced with the correct count in $P$ and extended anew (Step~\ref{ffixup:new2}, Algorithm~\ref{algo:ffixup}). Otherwise it is still representing the correct number of SPs to be restored and it is not processed.
\end{proof}

\begin{lemma} \label{proof:fflem}
	After the execution of Algorithm~\ref{algo:ffixup} (\fullyfixupend), for any $(x, y) \in V$, the STs
	in $P^{*}(x,y)$ (LSTs in $P(x,y)$) represent all the SPs (LSPs) from $x$ to $y$ in the updated graph.
	Also, the sets $L, L^{*}, R, R^{*}$ are correctly maintained.
\end{lemma}
\begin{proof}

	{\noindent \bf Loop Invariant:} 
	At the start of each iteration of the while loop in Step~\ref{fixup:phase3-begin} of Algorithm~\ref{algo:ffixup},
	assume that the first triple in $H_f$ to be extracted and processed has min-key = $[wt, x, y]$. Then the following properties hold about the \system and $H_f$.
	\begin{enumerate}
		\item \label{proof:fitem1}
		If $G$ contains $c_{ab}$ SPs of form $(xa, by)$ and weight ${wt}$,
		then $H_f$ contains a triple of form $(xa, by)$ and weight $wt$ to be extracted and processed.
		For any $a, b \in V$, if $G'$ contains $c_{ab}$ SPs of form $(xa, by)$ and weight ${wt}$,
		then $H_f$ contains a triple of form $(xa, by)$ and weight $wt$ to be extracted and processed. Further, a triple $\gamma = ((xa, by), { wt}, c_{ab})$ is present in $P(x,y)$.
		\item Let $[\hat{wt},\hat{x},\hat{y}]$ be the last key extracted from $H_f$ and processed before $[wt,x,y]$.
		Every SP in $G$, with an associated key $[wt_1, x_1, y_1] \leq [\hat{wt},\hat{x},\hat{y}]$, is represented by a triple $\gamma$ in the \system having the correct count. Moreover, the \system contains all the valid extensions derived by $\gamma$, and any extension of $\gamma$ itself is present in $H_f$ as a candidate triple.   
		For any key $[wt_1, x_1, y_1] \leq [\hat{wt},\hat{x},\hat{y}]$, let
		$G'$ contain ${ c}  > 0 $ number of LHPs of weight  ${ wt_1}$ of the
		form $(x_1a_1, b_1y_1)$. Further, let ${ c_{new}}$ (resp. ${ c_{old }}$) denote the number of these LHPs
		 that are {\em new} (resp. not {\em new}).
		Here ${ c_{new} + c_{old} = c}$. If $c_{new} > 0$ then,
		\begin{enumerate}
		\item \label{proof:fitem2}
		there is an LHT in $P(x_1, y_1)$ of the form $(x_1a_1, b_1y_1)$ and weight $wt_1$ that represents $c$ LHPs. 
		\item \label{proof:fitem21}
		If a triple of the form $(x_1a_1, b_1y_1)$ and weight $wt_1$ is present as an HT in $P^*$, then it represents the exact same count of $c$ HPs of its corresponding triple in $P$. This is exactly the number of HPs of the form $(x_1a_1, b_1y_1)$ and weight $wt_1$ in $G'$.
		\item \label{proof:fitem3} $(x_1,wt_1) \in L(a_1, b_1y_1)$, $(y_1,wt_1) \in R(x_1a_1, b_1)$, and if the triple of the form $(x_1a_1, b_1y_1)$ and weight $wt_1$ is also shortest then $(x_1,wt_1) \in L^{*}(a_1, y_1)$, $(y_1,wt_1) \in R^{*}(x_1, b_1)$.
		Further, $(x_1a_1, b_1y_1) \in $ \MT
		iff ${ {c_{old}} > 0}$.
		\item \label{proof:fitem4} If $c_{new} > 0$, for every $(x',wt') \in L(x_1, b_1y_1)$, an LHT corresponding to $(x' x_1, b_1y_1)$
		with weight $wt'~=~wt_1~+~\weight (x',x_1) \geq wt$ and counts equal to the sum of new paths represented by its constituents, is in $H_f$ and $P$. 
		A similar assertion holds for $(y',wt') \in R(x_1a_1, y_1)$.
		\end{enumerate}
		\item \label{proof:fitem5}
		Any triple $\gamma$ in $H_f$, which is longer than the last triple processed, is also in \MT if and only if it contains at least one old LHP (generated by a previous update) and a new LHP added to $H_f$ during a previous step of the current update.   
		For any key $[wt_2, x_2, y_2 ] \geq [wt, x, y]$, let
		$G'$ contain $c > 0$ number of LHPs of weight  ${ wt_2}$ of the
		form $(x_2a_2, b_2y_2)$. Further, let ${ c_{new}}$ (resp. ${ c_{old }}$) denote the number of such LHPs
		that are {\em new} (resp. not {\em new}).
		Here ${ c_{new} + c_{old} = c}$. Then the tuple $(x_2 a_2, b_2 y_2) \in$ \MTend, iff 
		$c_{old} > 0$ and $c_{new}$ paths have been added to $H_f$ by some earlier iteration of the while loop.
	\end{enumerate}
	
	Initialization and Maintenance for the 3 invariant assertions are similar to the proof of Lemma~\ref{lem:fcleanup-corr2}.
	
	{\noindent \bf Termination:} The condition to exit the loop is $H_f = \emptyset$. Because invariant assertion $\ref{proof:item1}$ maintains in $H_f$ the first triple to be extracted and processed, then $H_f = \emptyset$ implies that there are no triples, formed by a valid left or right extensions, that contain \emph{new} SPs or LSPs, that need to be added or restored in the graph $G$. Moreover, because of invariant assertions $\ref{proof:item2}$ and $\ref{proof:item21}$, every triple containing the node $v$, extracted and processed before $H_f = \emptyset$, has been added or restored with its correct count in the \systemend. Finally, for invariant assertion $\ref{proof:item3}$, the stacks $L, L^{*}, R, R^{*}$ are correctly maintained. This completes the proof of the loop invariant.
	
    \vone
	By Lemma \ref{proof:fflem0}, all the new SPs in $G'$ are placed in $H_f$ and processed by the algorithm and hence are in $P^*$ after the execution of Algorithm~\ref{algo:ffixup}. Moreover, for a pair $(x,y)$, the check in Step~\ref{ffixup:main1} of Algorithm~\ref{algo:ffixup} would fail if the distance from $x$ to $y$ doesn't change after the update. Thus the old SPs from $x$ to $y$ will remain in $P^*(x,y)$. Hence, after Algorithm~\ref{algo:ffixup} is executed, every SP in $G'$ is in its corresponding $P^*$.   
	
	Since every LST of the form $(xa,by)$ in $G'$ is formed by a left extension of a set of STs of the form $(a \times,by)$ (Steps~\ref{fixup:startleft}--\ref{fixup:endleft},  Algorithm~\ref{algo:ffixup}), or a right extension of a set of the form $(xa,\times b)$ (analogous steps for right extensions), and all the STs are correctly maintained by the algorithm, then all the LSTs are correctly maintained at the end of the fixup algorithm. This completes the proof of the Lemma.
\end{proof}

\section{The \ffullydynamic Algorithm for APASP} \label{sec:algo}

\subsection{Description of \ffullycleanup} \label{ffcleanup-desc}
\ffullycleanup removes all the LHPs going through the updated vertex $v$ from all the global structures $P$, $P^{*}$, $L$ and $R$, and from all local structures in any active level graph $\levelgraph_j$ that contains these triples. This involves decrementing the count of some triples or removing them completely (when all the paths in the triple go through $v$). The algorithm also updates local dictionaries and the $\DMs$, $\RN$ and $\LN$ structures.
Algorithm \ffullycleanup is a natural extension of the \NPR cleanup. An extension of the \NPR cleanup is used also in our basic algorithm but in a different way.

\begin{algorithm}
   \SetAlgoLined
   \SetAlgoNoEnd
  %\KwData{}
  %\KwResult{}
		$H_c \leftarrow \emptyset$; \MT$\leftarrow \emptyset$\;
		$\gamma \leftarrow [(v,v), 0, 1]$; $\CA_{\gamma}[center(v)]=1$; add $[\gamma, \CA_{\gamma}]$ to $H_c$\; \label{ffcleanup:init}
		\While {$H_c \neq \emptyset$ }{ \label{ffcleanup:while}
		extract in $S$ all the triples with the same min-key $[wt,x,y]$ from $H_c$\; \label{ffcleanup:extract}
		\ffullycleanupend-$\ell$-extend($S$,$[wt,x,y]$) (see Algorithm \ref{algo:ffcleanup-process})\;
		\ffullycleanupend-$r$-extend($S$,$[wt,x,y]$)\;
		}
	\caption{\ffullycleanupend$(v)$}
	\label{algo:ffcleanup}
\end{algorithm}

\ffullycleanup starts as in the \NPR algorithm. We add the updated node $v$ to $H_c$ (Step \ref{ffcleanup:init} -- Alg. \ref{algo:ffcleanup}) and we start extracting all the triples with same min-key (Step \ref{ffcleanup:extract} -- Alg. \ref{algo:ffcleanup}).
The main differences from \NPR start after we call Algorithm \ref{algo:ffcleanup-process}.
As in \cite{NPR14b}, we start by forming a new triple $\gamma'$ to be deleted (Steps \ref{ffcleanup:triplep} -- Alg. \ref{algo:ffcleanup}).
A new feature in Algorithm \ref{algo:ffcleanup-process} is to accumulate the paths that we need to remove level by level using the array $\CA'$. This is inspired by \Tho where unique SPs are maintained in each level. However, our algorithm maintains multiples paths spread across different levels using the $\CA_{\gamma}$ arrays associated to LSTs, and the technique used to update the $\CA_{\gamma}$ arrays is significantly different and more involved than the one described in \Thoe. Step~\ref{ffcleanup:updvect} - Alg.~\ref{algo:ffcleanup-process} calls \ffullyccenters (Alg.~\ref{algo:ffcleanup-vector}) that will perform this task.

\begin{algorithm}
   \SetAlgoLined
   \SetAlgoNoEnd
  %\KwData{}
  %\KwResult{}
		\For {every $b$ such that $(x\times,by) \in S$}{
		let $S_b \subseteq S$ be the set of all triples of the form $(x\times,by)$\;
		let $fcount'$ be the sum of all the $counts$ of all triples in $S_b$\;
		\For {every $x'$ in $L(x,by)$ s.t. $(x'x,by) \notin$ \MT}{
		$wt' \leftarrow wt+\weight(x',x)$; $\gamma' \leftarrow ((x'x,by),wt', fcount')$\; \label{ffcleanup:triplep}
		$\CA_{\gamma'} \leftarrow$ \ffullyccentersend$(\gamma', S_b)$\; \label{ffcleanup:updvect}
		add $[\gamma', \CA_{\gamma'}]$ to $H_c$\; \label{ffcleanup:addHcp}
		remove $\gamma'$ in $P(x',y)$ // decrements $count$ by $fcount'$\; \label{ffcleanup:remP}
		set new center for $\gamma''=((x'x,by),wt')$ in $P(x',y)$ as $\arg\!\min_i(\CA_{\gamma''}[i] \neq 0)$\;
		\label{ffcleanup:setcen}	
		\eIf { a triple for $(x'x,by)$ exists in $P(x',y)$}{ \label{ffcleanup:LRstart}		
		insert $(x'x,by)$ in \MTend\;
		}{
		delete $x'$ from $L(x,by)$ and delete $y$ from $R(x'x,b)$\; \label{ffcleanup:remLR}
		} \label{ffcleanup:LRend}
		\If { no triple for $((x'-,by),wt')$ exists in $P(x',y)$}{ \label{ffcleanup:RN}
		remove $b$ from $\RN(x',y,wt')$\;
		}
		\If { no triple for $((x'x,-y),wt')$ exists in $P(x',y)$}{ 
		remove $x$ from $\LN(x',y,wt')$\;	
		}	\label{ffcleanup:LN}				
		\If {a triple for $((x'x,by),wt')$ exists in $P^*(x',y)$}{ 
		remove $\gamma'$ in $P^*(x',y)$ // decrements $count$ by $fcount'$\; \label{ffcleanup:remPS}
		\If {$\gamma' \notin P^*(x',y)$}{
		remove the element with weight $wt'$ from $\DMs(x',y)$ if not linked to other tuples in $P^*(x',y)$\; \label{ffcleanup:remDM}
		}
		\For {each $i$}{ \label{ffcleanup:forPri}
		decrement $\CA_{\gamma'}[i]$ paths from $\gamma' \in P^*_i(x',y)$\; \label{ffcleanup:updPri0}
		\If {$\gamma'$ is removed from $P^*_i(x',y)$}{ \label{ffcleanup:remLSs}
		\uIf	{$x'$ is centered in level $i$}{
		\If {$\forall \, j\geq i, P^*_j(x,y)=\emptyset$ or $x'=v$}{
		remove $x'$ from $L^*_i(x,y)$ and remove $x'$ from $LC^*_i(x,y)$\; \label{ffcleanup:updLC}
			}
		}
		\uElseIf {$P^*_i(x,y)=\emptyset$ or $x'=v$}{
		remove $x'$ from $L^*_i(x,y)$\;
		}
		\uIf {$y$ is centered in level $i$}{
		\If {$\forall \, j\geq i, P^*_j(x',b)=\emptyset$ or $y=v$}{
		remove $y$ from $R^*_i(x',b)$ and remove $y$ from $RC^*_i(x',b)$\; \label{ffcleanup:updRC}
		}
		}
		\uElseIf {$P^*_i(x',b)=\emptyset$ or $y=v$}{
		remove $y$ from $R^*_i(x',b)$\;
			}
		} \label{ffcleanup:remLSe}
		}
		}
		}
	}
	\caption{\ffullycleanupend-$\ell$-extend($S,[wt,x,y]$)}
	\label{algo:ffcleanup-process}
\end{algorithm}

\begin{algorithm}
  \SetAlgoLined
  \SetAlgoNoEnd
  %\KwData{}
  %\KwResult{}
	let $\gamma' = ((x'x,by),wt',fcount')$ (the triple of the form $((x'x,by),wt')$ that contains all the paths through $v$ to be removed)\;	
	let $\gamma'' = ((x'x,by),wt',fcount'')$ (the triple of the form $((x'x,by),wt')$ in $P(x',y)$. Note that $\gamma'$ represents a subset of $\gamma''$)\;
	$j \leftarrow \arg\!\max_j(\CA_{\gamma''}[j]\neq 0)$ \label{ffcleanup:gpivot} // This is the oldest level in which a path in $\gamma''$ appeared for the first time\;
	$\CA' \leftarrow \sum_{\gamma \in S_b}{\CA_{\gamma}[r-1, \ldots, 0]}$ // This is the sum (level by level) of the triples of the form $(xa_i,by)$ that go through~$v$ and are extending to $\gamma'$ during this stage\; \label{ffcleanup:dist1}
	create a new center vector $\CA_{\gamma'}$ for the triple $\gamma'$ as follows\; 	\label{ffcleanup:newArr0}	
	\hspace{0.3in} for all the levels $m>j$ we set $\CA_{\gamma'}[m] = 0$\; \label{ffcleanup:newArr1}
	\hspace{0.3in} for the level $j$ we set $\CA_{\gamma'}[j] = \sum_{k = j}^{r-1}{\CA'[k]}$\; \label{ffcleanup:newArr2}
    \hspace{0.3in} for all the levels $i<j$ we set $\CA_{\gamma'}[i] = \CA'[i]$\; \label{ffcleanup:newArr3}	
	$\CA_{\gamma''}[r-1, \ldots, 0] \leftarrow \CA_{\gamma''}[r-1, \ldots, 0] - \CA_{\gamma'}[r-1, \ldots, 0]$ // We update the $\CA$ vector for $\gamma'' \in P(x',y)$\; \label{ffcleanup:Pvecupd}
	return $\CA_{\gamma'}$ // We return the correct vector for the generated $\gamma'$ triples\;
	\caption{\ffullyccentersend$(\gamma', S_b)$}
	\label{algo:ffcleanup-vector}
\end{algorithm}

\ffullyccenters takes as input the generated triple $\gamma'$ of the form $(x'x,by)$ which contains all the paths going through the updated node $v$ to be removed, and the set of triples $S_b$ of the form $(x\times,by)$ that are extended to $x'$ to generate $\gamma'$. This procedure has two tasks: (1) generating the $\CA_{\gamma'}$ vector for the triple $\gamma'$ that will be reinserted in $H_c$ for further extensions, and (2) updating the $\CA_{\gamma''}$ vector for the tuple $\gamma''$ in $P(x',y)$ (note that $\gamma''$ is the corresponding triple in $P$ of $\gamma'$, before we subtract all the paths represented by $\gamma'$ level by level).\\
(1) - This task, which is more complex than the second task (which is a single step in the algorithm, see point (2) below), is accomplished in steps \ref{ffcleanup:dist1} to \ref{ffcleanup:newArr3}, Alg. \ref{algo:ffcleanup-vector} and uses the following technique.
In step \ref{ffcleanup:dist1} -- Alg. \ref{algo:ffcleanup-vector}, we store into the $\log n$-size array $C'$ the distribution over the active levels for the set of triples in $S_b$ that generates $\gamma'$ using the left extension to $x'$. In order to generate the correct vector $\CA_{\gamma'}$ (to associate with the triple $\gamma'$), we need to reshape the distribution in $\CA'$ according to the corresponding distribution of the triple $\gamma'' \in P$.  
The reshaping procedure works as follows: we first identify the oldest level $j$ in which the triple $\gamma''$ appeared in $P$ for the first time (Step \ref{ffcleanup:gpivot} -- Alg. \ref{algo:ffcleanup-vector}).
Recall that we want to remove $\gamma'$ paths containing $v$ from $\gamma''$, and $\gamma''$ does not exist in any level older than $j$. Vector $\CA'$ is the sum of $\CA_{\gamma}$ for all $\gamma \in S_b$ (Step \ref{ffcleanup:dist1} -- Alg. \ref{algo:ffcleanup-vector}). Those triples are of the form $(xa_i,by)$ and they could exist in levels older, equal or more recent than $j$. But the triples in $S_b$ that were present in a level older than $j$, were extended to $\gamma''$ in $P$ for the first time in level $j$. For this reason, step \ref{ffcleanup:newArr2} - Alg. \ref{algo:ffcleanup-vector} aggregates all the counts in $\CA'$ in levels older or equal $j$ in $\CA_{\gamma'}[j]$.
Moreover, for each level $i < j$, if a triple $\gamma \in S_b$ is present in the level graph $\Gamma_i$ with $count$ paths centered in level $i$, then $\levelgraph_i$ also contains its extension to $x'$ that is a subtriple of $\gamma''$ located in level $i$ with at least $count$ paths. Thus for each level $i < j$ step \ref{ffcleanup:newArr3} - Alg. \ref{algo:ffcleanup-vector}, copies the number of paths level-wise.   
This procedure allows us to precisely remove the LHPs only from the level graphs where they exist. After $\CA'$ is reshaped into $\CA_{\gamma'}$ (steps \ref{ffcleanup:newArr0} to \ref{ffcleanup:newArr3} - Alg. \ref{algo:ffcleanup-vector}), the algorithm returns this correct array for $\gamma'$ to Alg. \ref{algo:ffcleanup-process}.\\
(2) - This task is performed by the simple step \ref{ffcleanup:Pvecupd}, Alg. \ref{algo:ffcleanup-vector}, which is a subtraction level by level of LHPs.

After adding the new triple $\gamma'$ to $H_c$ (Step \ref{ffcleanup:addHcp} - Alg. \ref{algo:ffcleanup-process}), the algorithm continues as the \NPR (Steps \ref{ffcleanup:addHcp} to \ref{ffcleanup:remLR} -- Alg. \ref{algo:ffcleanup-process}) with some differences: we need to update centers, local data structures, $\DMs$, $\RN$ and $\LN$. We update the center of $\gamma'$ using $C_{\gamma'}$ (Step \ref{ffcleanup:setcen} - Alg. \ref{algo:ffcleanup-process}).
If $\gamma'$ is a shortest triple, we decrement the count of $\gamma' \in P^*(x',y)$ (Step \ref{ffcleanup:remPS} - Alg. \ref{algo:ffcleanup-process}). 
If $\gamma'$ is completely removed from $P^*(x',y)$ and $\DMs(x',y, wt')$ is not linked to any other tuple, we remove the entry with weight $wt'$ from $\DMs(x',y)$ (Step \ref{ffcleanup:remDM} - Alg. \ref{algo:ffcleanup-process}).
Moreover, we subtract the correct number of paths from each level using the (previously built) array $\CA_{\gamma'}$ (Step \ref{ffcleanup:updPri0} - Alg. \ref{algo:ffcleanup-process}). Finally for each active level $i$, if $\gamma'$ is removed from $P^*_i(x',y)$, we take care of the sets $L_i^*$ and $R_i^*$ (Steps \ref{ffcleanup:remLSs} to \ref{ffcleanup:remLSe} - Alg. \ref{algo:ffcleanup-process}). In the process, we also update $LC_i^*$ and $RC_i^*$ in case the endpoints of $\gamma'$ are centered in level $i$.
If $\gamma'$ is completely removed from $P(x',y)$, using the double links to the node $b$ in $\RN(x',y,wt')$, we check if there are other triples that use $b$ in $P(x',y)$ (Step \ref{ffcleanup:RN} - Alg. \ref{algo:ffcleanup-process}): if not we remove $b$ from $\RN(x',y,wt')$. 
A similar step handles $\LN(x',y,wt')$.

\subsection{Description of \ffullyfixup} \label{ffixup-desc}
\ffullyfixup is an extension of \fullyfixup rather than \NPRe. This is because of the presence of the control bit $\beta$ (defined in Section \ref{subsec:fully-impl}), and the need to process historical triples (that are not present in \NPRe).
Algorithm \ffullyfixup will efficiently maintain exactly the LSTs and STs for each level graph in the \systemend. This is in contrast to \fullyfixupend, which can maintain \LHTe s that are not LSTs in any level graph (PDG).
\ffullyfixup maintains a heap $H_f$ of candidate LHTs to be processed in min-weight order. The main phase (Alg. \ref{algo:fffixup}) is very similar to the fixup in our basic algorithm. The differences are again related to levels, centers and the new data structures. 

We start describing Algorithm \ref{algo:fffixup}.
We initialize $H_f$ by inserting the edges incident on the updated vertex $v$ with their updated weights (Steps \ref{algo:finit-f-start} to \ref{algo:finit-f-end} -- Alg. \ref{algo:ftrivial}), as well as a candidate min-weight triple from $P$ for each pair of nodes $(x,y)$ (Step \ref{fixup:phase2-begin} -- Alg. \ref{algo:ftrivial}). Then we process $H_f$ by repeatedly extracting collections of triples of the same min-weight for a given pair of nodes, until $H_f$ is empty (Steps \ref{ffixup:phase3-begin} to \ref{ffixup:phase3-end} -- Alg. \ref{algo:fffixup}). We will establish that the first set of triples for each pair $(x,y)$ always represents the shortest path distance from $x$ to $y$ (see Lemma \ref{fdfixcorr}), and the triple extracted are added to the \system if not already there (see Alg. \ref{algo:fget-n-paths} and Lemma \ref{proof:ffflem}). For efficiency, among all the triples present in the \system for a pair of nodes, we select only the ones that need to be extended: this task is performed by Algorithm \ref{algo:fget-n-paths} (this step is explained later in the description).
After the triples in $S$ are left and right extended by Algorithm \ref{algo:fprocess-f}, we set the bit $\beta(\gamma') = 1$ for each triple $\gamma'$ that is identified as shortest in $S$, since $\gamma'$ is correctly updated both in $P^*(x,y)$ and $P(x,y)$ (Step \ref{ffixup:setbit1} -- Alg. \ref{algo:fffixup}). Finally, we update the $\DMs(x,y)$ structure by inserting (or updating if an element with weight $wt$ is already present) the element with weight $wt$ and the current level at the end of the list (Step \ref{ffixup:addDM} -- Alg. \ref{algo:fffixup}). This concludes the description of Algorithm \ref{algo:fffixup}.

We now describe Algorithm \ref{algo:fget-n-paths} which is responsible to select only the triples that have valid extensions that will generate LHTs in the current graph. 
In Algorithm \ref{algo:fget-n-paths}, we distinguish two cases. When the set of extracted triples from $x$ to $y$ contains at least one path not containing $v$ (Step \ref{ffixup:phase3-main-check} -- Alg. \ref{algo:fget-n-paths}), then we process all the triples from $P(x,y)$ of the same weight. Otherwise, if all the paths extracted go through $v$ (Step \ref{ffixup:phase3-nomain-check} -- Alg. \ref{algo:fget-n-paths}), we only use the triples extracted from $H_f$. 

\begin{algorithm}
   \SetAlgoLined
   \SetAlgoNoEnd
  %\KwData{}
  %\KwResult{}
	$H_f \leftarrow \emptyset$; \MT$\leftarrow \emptyset$\; \label{ffixup:init-empty}
	\ffullypopulate$(v, \weight',k)$\; \label{ffixup:phase2}
	\While {$H_f \neq \emptyset$}{ \label{ffixup:phase3-begin}
	extract in $S'$ all the triples with min-key $[wt,x,y]$ from $H_f$\; \label{ffixup:phase3-extract1}
	\If {$S'$ is the first extracted set from $H_f$ for $x,y$}{  \label{ffixup:phase3-first-ext}
		$S \leftarrow$ \ffullygetnew($S',P(x,y)$)\;
		\ffullyfixupend-$\ell$-extend($S$,$[wt,x,y]$) (see Algorithm \ref{algo:fprocess-f})\; \label{ffixup:lext}
		\ffullyfixupend-$r$-extend($S$,$[wt,x,y]$)\; \label{ffixup:rext}
		for every $\gamma \in S$ set $\beta(\gamma)=1$\; \label{ffixup:setbit1}
		add an element with weight $wt$ and level $k$ to $\DMs(x,y)$ or update the level in the existing one\; \label{ffixup:addDM}
	}
	} \label{ffixup:phase3-end}
	\caption{\ffullyfixupend$(v, \weight', k)$}
	\label{algo:fffixup}
\end{algorithm}

\begin{algorithm}
   \SetAlgoLined
   \SetAlgoNoEnd
  %\KwData{}
  %\KwResult{}
	\For {each $(u, v)$}{
	$\weight(u,v) = \weight'(u,v)$\; \label{algo:finit-f-start}
	\If {$\weight(u,v) < \infty$ }{
		$\gamma = ((uv,uv),\weight(u,v),1)$; $C_\gamma[k] \leftarrow 1$\;
		update-num($\gamma$) $\leftarrow$ curr-update-num; num-v-paths($\gamma$) $\leftarrow 1$\;
		add $[\gamma, C_{\gamma}]$ to $H_f$ and $P(u,v)$\;
		add $u$ to $L(-,vv)$ and $v$ to $R(uu,-)$\;
	} \label{algo:finit-f-end}
	}
	\For {each $(v,u)$}{
	symmetric processing as Steps~\ref{algo:finit-f-start}--\ref{algo:finit-f-end} above\;
	}
	\For {each $x, y \in V$}{ \label{fixup:phase2-begin}
		add a min-key triple $[\gamma, C_{\gamma}] \in P(x,y)$ to $H_f$\;
	} \label{fixup:phase2-end}
	\caption{\ffullypopulate$(v, \weight',k)$}
	\label{algo:ftrivial}
\end{algorithm}

\begin{algorithm}
	   \SetAlgoLined
	   \SetAlgoNoEnd
	   %\KwData{}
	   %\KwResult{}	   
	$S \leftarrow \emptyset$; let $i$ be the min-weight level associated with $\DMs(x,y)$\;
	\eIf {$P^*(x,y)$ increased min-weight after cleanup}{  \label{ffixup:phase3-main-check}
	\For {each $\gamma' \in S$ with-key $[wt,0]$}{ \label{ffixup:phase3-addfromP-begin}
		let $\gamma' = ((xa', b'y), wt, count')$ and $j=\arg\!\min_j(\CA_{\gamma'}[j] \neq 0)$\; \label{ffixup:newcenter}
		\uIf {$\gamma'$ is not in $P^*(x,y)$}{                	
		add $\gamma'$ in $P^{*}(x,y)$ and $S$; add $x$ to $L^{*}(a',y)$ and $y$ to $R^{*}(x,b')$\;
		add $b'$ to $\RN(x,y,wt)$; place a double link between $\gamma'$ and $\DMs(x,y,wt)$\;
	}
	\uElseIf {$\gamma'$ is in $P(x,y)$ and $P^*(x,y)$ with different counts}{ \label{fffixup:new1} 
		replace the count of $\gamma'$ in $P^{*}(x,y)$ with $count'$ and add $\gamma'$ to $S$\; \label{fffixup:new2} 
		}
		add $\gamma'$ to $P^*_j(x,y)$ and $dict_j$\; \label{ffixup:phase3-addj}
		add $x$ to $L^*_j(a',y)$ and $y$ to $R^*_j(x,b')$\;
		add $x$ to $LC^*_j(a',y)$ ($y$ to $RC^*_j(x,b')$) if $x$ ($y$) is a level $i$ center\; \label{ffixup:phase3-addPe}
		add $\gamma'$ in $S$\;
	}
	\For {each $b' \in \RN(x,y,wt)$}{ \label{ffixup:RNstart}
		\If {$\exists \, h<i : L_h^*(x,b')\neq \emptyset$}{
		add any $\gamma'$ of the form $(x \times,b'y)$ and weight $wt$ in $P^*(x,y)$ with $\beta(\gamma')=1$ to $S$\; \label{ffixup:RNend}
		}
	}
	}{
		\For {each $\gamma' \in S'$ containing a path through $v$}{ \label{ffixup:phase3-nomain-check}
		let $\gamma' = ((xa', b'y), wt, count')$ and $k$ the current level\;
		add $\gamma'$ with paths$(\gamma',v)$ to $P^*(x,y)$, and $[\gamma', \CA_{\gamma'}]$ to $S$\; \label{ffixup:phase3-add2PStar1}
		add $\gamma'$ to $P^*_k(x,y)$ and $dict_k$, $x$ to $L^*_k(a',y)$ and $y$ to $R^*_k(x,b')$\;
		\label{ffixup:phase3-addk}
		add $x$ to $LC^*_k(a',y)$ ($y$ to $RC^*_k(x,b')$) if $x$ ($y$) is a level $k$ center\;
		\label{ffixup:phase3-add2LRCStar1}
		} \label{ffixup:phase3-addfromX-end}	
	}
	return $S$\;
	\caption{\ffullygetnew($S', P_{xy}$)}
	\label{algo:fget-n-paths}
\end{algorithm}

Both cases have a similar approach but here we focus on the former which is more involved than the latter. 
As soon as we identify a new triple $\gamma'$ we compute its center $j$ by using its associated array $\CA_{\gamma'}$ (Step \ref{ffixup:newcenter} -- Alg. \ref{algo:fget-n-paths}). This is straightforward if compared to \ffullycleanup where we first need to update the center arrays. We add this triple to $P^*(x,y)$ and to $S$, which contains the set of triples that need to be extended.
We also add $\gamma'$ to $P_j^*(x,y)$ (Steps \ref{ffixup:phase3-addj} and \ref{ffixup:phase3-addk} -- Alg. \ref{algo:fget-n-paths}). We update $dict_j$ to keep track of the locations of the triple in the global structures. 
A similar sequence of steps takes place when all the extracted paths go through $v$ (Steps \ref{ffixup:phase3-nomain-check} to \ref{ffixup:phase3-addfromX-end} -- Alg. \ref{algo:fget-n-paths}). The only difference is that the local data structures to be updated are only the $\levelgraph_k$ data stuctures (Steps \ref{ffixup:phase3-addk} and \ref{ffixup:phase3-add2LRCStar1} -- Alg. \ref{algo:fget-n-paths}). 

A crucial difference from \fullyfixup and this algorithm is the way we collect the set $S$ of triples to be extended. Here we require the new HE data structures $\RN$ and $\LN$ (see Section \ref{sec:newfeat}) because of PEP instances (see Section \ref{sec:ffeatures}).
Let $i$ be the min-weight level associated with $DL(x,y)$.
For each node $b \in \RN(x,y,wt)$ we check if $L_h^*(x,b)$ contains at least one extension, for every $h < i$ (Steps \ref{ffixup:RNstart} to \ref{ffixup:RNend} -- Alg. \ref{algo:fget-n-paths}). In fact we need to discover all tuples with $\beta=1$ that are inside a PEP instance. In this instance, the triples restored as STs may or may not be extended. We cannot afford to look at all of them, thus our solution should check only the triples with an available extension. Moreover, all the extendable triples with with $\beta=1$ have extension only in levels younger than the level where they last appear as STs. 
Thus, we check for extensions only in the levels $h < i$. 

Using the HE sets, is the key to avoid an otherwise long search of all the valid extensions for the set of examined triples with $\beta=1$. In particular, without the HE sets, the algorithm could waste time by searching for extensions that are not even in the \systemend.
Correctness of this method is proven in section \ref{sec:proof}.
After the algorithm collects the set $S$ of triples that can be extended, \ffullyfixup calls \ffullyfixupend-$\ell$-extend (Alg. \ref{algo:fprocess-f}).

\begin{algorithm}
   \SetAlgoLined
   \SetAlgoNoEnd
  %\KwData{}
  %\KwResult{}
	\For {every $b$ such that $(x\times,by) \in S$}{
		let $S_b \subseteq S$ be the set of all triples of the form $(x\times,by)$\;
		let $fcount'$ be the sum of all the $counts$ of all triples in $S_b$; let $h$ be the $center(S_b)$\;
		\eIf {$\exists \, \gamma \in S_b : \beta(\gamma) = 0$}{
		let $j$ be the level associated to the minweight $wt'>wt$ in $\DMs(x,y)$\; \label{fprocess-f:DM}
		\For {every active level $h \leq i<j$}{
		\For {every $x'$ in $L_i^{*}(x,b)$}{ \label{fprocess-f:exts}
		\If {$(x'x,by) \notin$ \MT}{ \label{fprocess-f:start}
		$wt' \leftarrow wt+\weight(x',x)$; $\gamma' \leftarrow ((x'x,by),wt', fcount')$\; \label{fprocess-f:Lexts}
		$\CA_{\gamma'} \leftarrow$ \ffullyfcentersend($S_b$); add $\gamma'$ to $H_f$\; \label{fprocess-centers}
		\eIf {a triple $\gamma''$ for $((x'x,by),wt')$ exists in $P(x',y)$}{
		update the count of $\gamma''$ in $P(x',y)$ and $\CA_{\gamma''} = \CA_{\gamma''} + \CA_{\gamma'}$\;
		add $(x'x,by)$ to \MT\;
		}{
		add $[\gamma',\CA_{\gamma'}]$ to $P(x',y)$; add $x'$ to $L(x,by)$ and $y$ to $R(x'x,b)$\; \label{fprocess-f:addL}
		}
		set $\beta(\gamma')=0$; set update-num($\gamma'$)\; \label{fprocess-f:end}
		}
		}
	} \label{fprocess-f:exte}
		\For {every level $i < h$}{ \label{fprocess-f:extihs}
		\For {every $x'$ in $LC_i^{*}(x,b)$}{
		execute steps \ref{fprocess-f:start} to \ref{fprocess-f:end}\;
		}
		} \label{fprocess-f:extihe}
	}{
		let $j$ be the level associated to the minweight $wt$ in $\DMs(x,y)$\; \label{fprocess-f:DM1} \label{fprocess-f:extis}
		\For {every level $i < j$}{ 
		\For {every $x'$ in $LC_i^{*}(x,b)$}{
		execute steps \ref{fprocess-f:start} to \ref{fprocess-f:end}\;
		}
		} \label{fprocess-f:extie}
	}
	}
	\caption{\ffullyfixupend-$\ell$-extend($S$,$[wt,x,y]$)}
	\label{algo:fprocess-f}
\end{algorithm}

\begin{algorithm}
   \SetAlgoLined
   \SetAlgoNoEnd
  %\KwData{}
  %\KwResult{}
		let $\CA'=\sum_{\gamma \in S_b}{\CA_{\gamma}}$ be the sum (level by level) of the new paths that are found shortest\;
		let $j$ be $\arg\!\max_j(\CA'[j]\neq 0)$, and $k=center(x')$\;	\label{fcenter:level}
		\eIf {$k < j$}{
		for all the levels $i<k$ we set $\CA_{\gamma'}[i] = \CA'[i]$\; \label{fcenter:ress}
		for the level $k$ we set $\CA_{\gamma'}[k] = \sum_{q = k}^{r-1}{\CA'[q]}$\;
		for all the levels $m>k$ we set $\CA_{\gamma'}[m] = 0$\; \label{fcenter:rese}
		}{
		$\CA_{\gamma'} = \CA'$\; \label{fcenter:nores}
		}	
		return $\CA_{\gamma'}$\;
	\caption{\ffullyfcentersend$(S_b)$}
	\label{algo:fffixup-centers}
\end{algorithm}

Here we describe the details of algorithm \ref{algo:fprocess-f}.
Its goal is to generate LHTs for the current graph $G$ by extending HTs.
Let $h$ be the \emph{center of $S_b$} defined as the most recent center among all the triples in $S_b$, and let $j$ be the level associated to the first weight $wt'$ larger than $wt$ in $\DMs(x,y)$.
The extension phase for triples is different from \fullyfixupend: in fact, the set of triples $S_b$ could contain only triples with $\beta(\gamma)=1$. In \fullyfixupend, the corresponding set $S_b$ contains only triples with $\beta(\gamma)=0$. We address two cases:\\ 
(\textbf{a}) -- If $S_b$ contains at least one triple $\gamma$ with $\beta(\gamma)=0$, we extend $S_b$ using the sets $L_i^{*}$ and $R_i^{*}$ with $h \leq i < j$ (Steps  \ref{fprocess-f:exts} to \ref{fprocess-f:exte} -- Alg.\ref{algo:fprocess-f}). In fact, the set $S_b$ contains at least one new path that was not extended in the previous iterations when $wt$ was the shortest distance from $x$ to $y$ (because of the $\beta(\gamma)=0$ triple). The LST generated in this way remains centered in level $h$. Moreover we extend $S_b$ also using the sets $LC_i^{*}$ and $RC_i^{*}$ with $i < h$ (Steps \ref{fprocess-f:extihs} to \ref{fprocess-f:extihe} -- Alg.\ref{algo:fprocess-f}). This ensures that every LST generated in a level $i$ lower than $h$ is centered in $i$ thanks to the extension node itself.  
This technique guarantees that each LHT generated by Algorithm \ref{algo:fprocess-f} is an LST centered in a unique level.\\
(\textbf{b}) -- In the case when there is no triple $\gamma$ in $S_b$ with $\beta(\gamma)=0$, then there is at least one extension to perform for $S_b$ and it must be in some level younger than the level where $wt$ stopped to be the shortest distance from $x$ to $y$ (this follows from the use of the HE sets in Alg. \ref{algo:fget-n-paths}). To perform these extensions we set $j$ as the level associated with the min-weight element in $\DMs(x,y)$, and we extend $S_b$ using the sets $LC_i^{*}$ and $RC_i^{*}$ with $i < j$ (Steps \ref{fprocess-f:extis} to \ref{fprocess-f:extie} -- Alg.\ref{algo:fprocess-f}). Again, every LHT generated is an LST centered in a unique level.
Finally, every generated LHT is added to $P$ and $H_f$ and we update global $L$ and $R$ structures. 
\begin{observation} \label{centerlht}
Every LHT generated by algorithm \ffullyfixup is an LST centered in a unique level graph.
\end{observation}
\begin{proof}
As described in (a) and (b) above, every LHT is generated using two triples which are shortest in the same level graph $\levelgraph_i$. Moreover, since at least one of them must be centered in level $i$, the resulting LHT is an LST centered in level $i$.
\end{proof}	

The last novelty in the algorithm is updating center arrays (Alg. \ref{algo:fffixup-centers} called at step \ref{fprocess-centers} -- Alg. \ref{algo:fprocess-f}) in a similar way of \ffullycleanupend: Algorithm \ref{algo:fffixup-centers} identifies the oldest level $j$ related to the triples contained in $S_b$ (Step \ref{fcenter:level} -- Alg. \ref{algo:fffixup-centers}). 
If $j > k$ then we reshape the distribution for $\gamma'$ similarly to \ffullycleanup (Steps \ref{fcenter:ress} to \ref{fcenter:rese} -- Alg. \ref{algo:fffixup-centers}). Otherwise $\gamma'$ is completely contained in level $k$ and no reshaping is required (Step \ref{fcenter:nores} -- Alg. \ref{algo:fffixup-centers}).

\subsection{Correctness of \FFDe} \label{sec:proofcorr}
For the correctness, we assume that all the global and local data structures are correct before the update, and we will show the correctness of them after the update.

\paragraph{\bf Correctness of Cleanup} The correctness of \ffullycleanup is established in Lemma \ref{lemma:fclean1}. We will prove that all paths containing the updated vertex $v$ are removed from the \systemend. Moreover, the center of each triple is restored, if necessary, to the level containing the most recently updated node on any path in this triple. 
Note that (as in \cite{DI04,NPR14b}) at the end of the cleanup phase, the global structures $P$ and $P^*$ may not have all the LHTs in $G \setminus \{v\}$.

\begin{lemma} \label{lemma:fclean1}
	At the end of the cleanup phase triggered by an update on a vertex $v$, every LHP that goes through $v$ is removed from the global structures. Moreover, in each level graph $\levelgraph_i$, each SP that goes through $v$ is removed from $P^*_i$. For each level $i$, the local structures $L_i^*$, $R_i^*$, $RC_i^*$ and $LC_i^*$ contain the correct extensions; the global structures $L$ and $R$ contain the correct extensions, for each $r$-tuple and $\ell$-tuple respectively, and the structures $\RN$ and $\LN$ contain only nodes associated with tuples in $P$. The $\DMs$ structure only contains historical distances represented by at least one path in the updated graph.
	Finally, every triple in $P$ and $P^*$ has the correct updated center for the graph $G \setminus \{v\}$.
\end{lemma}

\begin{proof}
The lemma is established with the following loop invariant. For more details and a full proof see \cite{PR15}.
	
	{\noindent \bf Loop Invariant:} At the start of each iteration of the while loop in Step~\ref{ffcleanup:while} of Algorithm~\ref{algo:ffcleanup},
	assume that the first triple to be extracted from $H_c$ and processed has min-key = $[wt, x, y]$. Then the following properties hold about the \system and $H_c$.
	\begin{enumerate}
		\item \label{fproof:item1} For any $a, b \in V$, if $G$ contains $c_{ab}$ LHPs of weight ${wt}$ of the form $(xa, by)$
		passing through $v$,
		then $H_c$ contains a triple $\gamma = ((xa, by), { wt}, c_{ab})$ with key $[wt,x,y]$ already processed: the $c_{ab}$ LHPs through $v$ are not present in the \systemend.
		
		\item \label{fproof:claim2} Let  $[\hat{wt},\hat{x},\hat{y}]$ be the last key extracted from $H_c$ and processed before $[wt,x,y]$. For any key $[wt_1, x_1, y_1] \leq [\hat{wt},\hat{x},\hat{y}]$, let 
		$G$ contain ${ c}  > 0 $ number of LHPs of weight  ${ wt_1}$ of the
		form $(x_1 \times, b_1y_1)$. Further, let ${ c_v}$ (resp. ${ c_{\bar v}}$) denote the number of such LHPs
		that pass through $v$ (resp. do not pass through $v$).
		Here ${ c_v + c_{\bar v} = c}$. For every extension $x' \in L(x_1, b_1y_1)$, let $wt' = wt_1 + \weight (x',x_1)$ be the weight of the extended triple $(x'x_1, b_1y_1)$. Then, (the following assertions are similar for $y' \in R(x_1a_1, y_1)$)\\
		\textbf{Global Data Structures: }
		\begin{enumerate}
			\item \label{fproof:item2} if $c>c_{v}$ there is a triple in $P(x', y_1)$ 
			of the form $(x'x_1, b_1y_1)$ and weight $wt'$ representing $c-c_v$ LHPs. Moreover, its center is updated according to the last update on any path represented by the triple. If $c=c_v$ there is no such triple in $P(x', y_1)$.
			\item \label{fproof:item21}
			If a triple of the form $(x'x_1, b_1y_1)$ and weight $wt'$ is present as an HT in $P^*(x', y_1)$, then it represents the exact same number of LHPs $c-c_v$ of the corresponding triple in $P(x',y_1)$. This is exactly the number of HPs of the form $(x'x_1, b_1y_1)$ and weight $wt'$ in $G \setminus \{v\}$.
			\item \label{fproof:item3} $x' \in L(x_1, b_1y_1)$,  $y_1 \in R(x'x_1, b_1)$, and    $(x'x_1, b_1y_1) \in $ \MT
			iff ${ {c_{\bar v}} > 0}$. 
			\item \label{fproof:item4} A triple corresponding to $(x' x_1, b_1y_1)$
			with weight $wt'$ and counts $c_v$ is in $H_c$.
			A similar assertion holds for $y' \in R(x_1a_1, y_1)$.
			\item \label{fproof:item6} The structure $\RN(x',y_1,wt')$ contains a node $b$ iff at least one path of the form $(x' \times,by_1)$ and weight $wt'$ is still represented by a triple in $P(x', y_1)$. A similar assertion holds for a node $a$ in $\LN(x',y_1,wt')$.
			\item \label{fproof:item7} If there is no HT of the form $(x'x, b_1y_1)$ and weight $wt'$ in $P^*(x', y_1)$ then the entry $\DMs(x', y_1)$ with weight $wt'$ does not exists.
		\end{enumerate}
		\textbf{Local Data Structures:} for each level $j$, let $c_j$ be the number of LSPs of the form $(x'x_1, b_1y_1)$ and weight $wt'$ centered in $\levelgraph_j$ and let $c_j(v)$ be the ones that go through $v$. Thus $c=\sum_{j}{c_j}$ and $c_v=\sum_{j}{c_j(v)}$. Then,
		\begin{enumerate}[resume]
			\item \label{fproof:litem1} the value of $\CA_\gamma[j]$, where $\gamma$ is the triple of the form $(x'x_1, b_1y_1)$ and weight $wt'$ in $P(x',y_1)$, is $c_j-c_j(v)$. 
			\item \label{fproof:litem2}
			If a triple $\gamma$ of the form $(x'x_1, b_1y_1)$ and weight $wt'$ is present as an HT in $P^*$, then $P_j^*(x',y_1)$ represents only $c_j-c_j(v)$ paths. If $c_j-c_j(v)=0$ then the link to $\gamma$ is removed from $dict_j$. Moreover, $x' \in L_j^*(x_1, y_1)$ (respectively $LC_j^*(x_1, y_1)$ if $x'$ is centered in $\levelgraph_j$) iff $x'$ is part of a shortest path of the form $(x'x_1, \times y_1)$ centered in $\levelgraph_j$. A similar statement holds for $y_1 \in R_j^*(x', b_1)$ (respectively $RC_j^*(x', b_1)$ if $y_1$ is centered in $\levelgraph_j$).
		\end{enumerate}
		\item \label{fproof:item5} For any key $[wt_2, x_2, y_2 ] \geq [wt, x, y]$, let
		$G$ contain $c > 0$ LHPs of weight  ${ wt_2}$ of the
		form $(x_2a_2, b_2y_2)$. 
		Further, let ${ c_v}$ (resp. ${ c_{\bar v}}$) denote the number of such LHPs
		that pass through $v$ (resp. do not pass through $v$).
		Here ${ c_v + c_{\bar v} = c}$. 
		Then the tuple $(x_2 a_2, b_2 y_2) \in$ \MTend, iff $c_{\bar v} > 0$ and a triple for
		$(x_2a_2, b_2 y_2)$ is present in $H_c$
	\end{enumerate}
\end{proof}

\paragraph{\bf Correctness of Fixup} For the fixup phase, we need to show that the triples generated by our algorithm are sufficient to maintain all the ST and LST in the current graph $G$.
As in our basic algorithm, we first show in the following lemma that \ffullyfixup computes all the correct distances for each pair of nodes in the updated graph. 
Finally, we show that data structures and counts are correctly maintained at the end of the algorithm (Lemma \ref{proof:ffflem}).
\begin{lemma} \label{fdfixcorr} 
	For every pair of nodes $(x,y)$, let $\gamma = ((xa, by), wt, count)$ be one of the min-weight triples from $x$ to $y$ extracted from $H_f$ during \ffullyfixupend. Then $wt$ is the shortest path distance from $x$ to $y$ in $G$ after the update.
\end{lemma}

\begin{proof}
	Suppose that the lemma is violated. Thus, there will be an extraction from $H_f$ during \ffullyfixup such that the set of extracted triples $S'$, of weight $\hat{wt}$ is not shortest in $G$ after the update. Consider the earliest of these events when $S'$ is extracted from $H_f$. Since $S'$ is not a set of STs from $x$ to $y$, there is at least one shorter tuple from $x$ and $y$ in the updated graph. 
	Let $\gamma' = ((xa',b'y),wt,count)$ be this triple that represents at least one shortest path from $x$ to $y$, with $wt < \hat{wt}$. 
	Since $S'$ is extracted from $H_f$ before any other triple from $x$ to $y$, $\gamma'$ cannot be in $H_f$ at any time during \ffullyfixupend. Hence, it is also not present in $P(x,y)$ as an LST at the beginning of the algorithm, otherwise it (or another triple with the same weight) would be placed in $H_f$ by step \ref{ffixup:phase2} - Alg. \ref{algo:fffixup}. Moreover, if $\gamma'$ is a single edge (trivial triple), then it was already an LST in $G$ present in $P(x,y)$ before the update, and it is added to $H_f$ by step \ref{fixup:phase2-begin} - Alg. \ref{algo:ftrivial}; moreover since all the edges incident to $v$ are added to $H_f$ during steps~\ref{algo:finit-f-start} to \ref{algo:finit-f-end} of Alg.~\ref{algo:ftrivial}, then $\gamma'$ must represent SPs of at least two edges. We define $left(\gamma')$ as the set of LSTs of the form $((xa',c_ib'), wt-\weight(b',y), count_{c_i})$ that represent all the LSPs in the left tuple $((xa',b'), wt-\weight(b',y))$; similarly we define $right(\gamma')$ as the set of LSTs of the form $((a'd_j,b'y), wt-\weight(x,a'), count_{d_j})$ that represent all the LSPs in the right tuple $((a',b'y), wt-\weight(x,a'))$. 
	
	Observe that since $\gamma'$ is an ST, all the LSTs in $left(\gamma')$ and $right(\gamma')$ are also STs. A triple in $left(\gamma')$ and a triple in $right(\gamma')$ cannot be present in $P^*$ together at the beginning of \ffullyfixupend. 
	In fact, if at least one triple from both sets is present in $P^*$ at the beginning of \ffullyfixupend, then the last one inserted during the fixup phase triggered during the previous update, would have generated an LST of the form $((xa',b'y), wt)$ automatically inserted, and thus present, in $P$ at the beginning of the current fixup phase (a contradiction). Thus either there is no triple represented by $left(\gamma')$ in $P^*$, or there is no triple represented by $right(\gamma')$ in $P^*$.
	
	Assume w.l.o.g. that the set of triples in $right(\gamma')$ is placed into $P^*$ after $left(\gamma')$ by \ffullyfixupend. 
	Since edge weights are positive, $wt-\weight(x,a') < wt < \hat{wt}$, and because all the extractions before $\gamma$ were correct, then the triples in $right(\gamma')$ were correctly extracted from $H_f$ and placed in $P^*$ before the wrong extraction of $S'$.
	Let $i$ be the level in which $left(\gamma')$ is centered, and let $j$ be the level in which $right(\gamma')$ is centered.
	By the assumptions, all the triples in $left(\gamma')$ are in $P^*$ and we need to distinguish 3 cases:
	\begin{enumerate}
		\item if $j=i$, then \ffullyfixup generates the tuple $((xa',b'y),wt)$ in the same level and place it in $P$ and $H_f$.
		\item if $i>j$, the algorithms \ffullyfixup extends the set $right(\gamma')$ to all nodes in $L_i^*(a',b')$ for every $i \geq j$ 
		(see Steps \ref{fprocess-f:exts} to \ref{fprocess-f:end} - Alg. \ref{algo:fprocess-f}). Thus, since  $left(\gamma')$ is centered in some level $i > j$, the node $x$ is a valid extension in $L_i^*(a',b')$, making the generated $\gamma'$ an LST in $\levelgraph_j$ that will be placed in $P(x,y)$ and also into $H_f$ (during Step \ref{fprocess-centers} - Alg. \ref{algo:fprocess-f}).
		\item if $j>i$, then $x$ was inserted in a level younger than $i$. In fact, all the paths from $a'$ to $b'$ must be the same in $right(\gamma')$ and $left(\gamma')$ otherwise the center of $right(\gamma')$ should be $i$. Hence, the only case when $j>i$ is when the last update on $left(\gamma')$ is on the node $x$ in a level $i$ younger than $j$. Thus $x \in LC_i^*(a',b')$. But \ffullyfixup extends $right(\gamma')$ to all nodes in $LC_i^*(a',b')$ for every $i < j$, placing the generated LST $\gamma'$ in $P(x,y)$ and also into $H_f$  
		(see Steps \ref{fprocess-f:extihs} to \ref{fprocess-f:extihe} - Alg. \ref{algo:fprocess-f}).
	\end{enumerate} 
	
	Thus the algorithm would generate the tuple $((xa',b'y),wt)$ (as a left extension) and place it in $P$ and $H_f$ (because all the triples in $left(\gamma')$ are already in $P^*$). Therefore, in all cases, a tuple $((xa',b'y),wt)$ should have been extracted from $H_f$ before any triple in $S'$. A contradiction.
\end{proof}

\begin{lemma} \label{proof:ffflem}
	After the execution of \ffullyfixupend, for any $(x, y) \in V$, the sets $P^{*}(x,y)$ ($P(x,y)$) contains all the SPs (LSPs) from $x$ to $y$ in the updated graph.
	Also, the global structures $L, R$ and the local structures $P_i^*, L_i^{*},R_i^{*},LC_i^{*},RC_i^{*}$ and $dict_i$ for each level $i$ are correctly maintained. The structures $\RN$ and $\LN$ are updated according to the newly identified tuples. The $\DMs$ structure contains the updated distance for each pair of nodes in the current graph. Finally, the center of each new triple is updated.
\end{lemma}
\begin{proof}
	We prove the lemma statement by showing the following loop invariant.
	For more details and a full proof see \cite{PR15}.
	Let $G'$ be the graph after the update.
	
	{\noindent \bf Loop Invariant:} 
	At the start of each iteration of the while loop in Step~\ref{ffixup:phase3-begin} of \ffullyfixupend,
	assume that the first triple in $H_f$ to be extracted and processed has min-key = $[wt, x, y]$. Then the following properties hold about the \system and $H_f$.
	\begin{enumerate}
		\item \label{fproof:fitem1} For any $a, b \in V$, if $G'$ contains $c_{ab}$ SPs of form $(xa, by)$ and weight ${wt}$,
		then $H_f$ contains a triple of form $(xa, by)$ and weight $wt$ to be extracted and processed. Further, a triple $\gamma = ((xa, by), { wt}, c_{ab})$ is present in $P(x,y)$. 
		\item Let $[\hat{wt},\hat{x},\hat{y}]$ be the last key extracted from $H_f$ and processed before $[wt,x,y]$. For any key $[wt_1, x_1, y_1] \leq [\hat{wt},\hat{x},\hat{y}]$, let
		$G'$ contain ${ c}  > 0 $ number of LHPs of weight  ${ wt_1}$ of the
		form $(x_1a_1, b_1y_1)$. Further, let ${ c_{new}}$ (resp. ${ c_{old }}$) denote the number of these LHPs
		that are {\em new} (resp. not {\em new}).
		Here ${ c_{new} + c_{old} = c}$. If $c_{new} > 0$ then,\\
		\textbf{Global Data Structures:}
		\begin{enumerate}
			\item \label{fproof:fitem2}
			there is an LHT $\gamma$ in $P(x_1, y_1)$ of the form $(x_1a_1, b_1y_1)$ and weight $wt_1$ that represents $c$ LHPs, with an updated center defined by the last update on any of the paths represented by the LHT. 
			\item \label{fproof:fitem21}
			If a triple of the form $(x_1a_1, b_1y_1)$ and weight $wt_1$ is present as an HT in $P^*$, then it represents the exact same count of $c$ HPs of its corresponding triple in $P$. This is exactly the number of HPs of the form $(x_1a_1, b_1y_1)$ and weight $wt_1$ in $G'$. Its control bit $\beta$ is set to 1.
			\item \label{fproof:fitem3} $x_1 \in L(a_1, b_1y_1)$, $y_1 \in R(x_1a_1, b_1)$.
			Further, $(x_1a_1, b_1y_1) \in $ \MT
			iff ${ {c_{old}} > 0}$.
			\item \label{fproof:fitem4} If $\beta(\gamma)=0$ or $\beta(\gamma)=1$ and there is an extension $x' \in L_j^*(x_1, y_1)$ that generates a centered LST in a level $j$, an LHT corresponding to $(x' x_1, b_1y_1)$
			with weight $wt'~=~wt_1~+~\weight (x',x_1) \geq wt$ and counts equal to the sum of new paths represented by its constituents, is in $H_f$ and $P$. 
			A similar assertion holds for an extension $y' \in R_j^*(x_1, y_1)$.
			\item \label{fproof:fitem6} The structure $\RN(x_1,y_1,wt_1)$ contains a node $b$ iff at least one path of the form $(x_1 \times,by_1)$ and weight $wt_1$ is represented by a triple in $P(x_1, y_1)$. A similar assertion holds for a node $a$ in $\LN(x_1,y_1,wt_1)$.
			\item \label{fproof:fitem7} The entry $\DMs(x_1, y_1)$ with weight $wt_1$ is updated to the current level.
		\end{enumerate}
		\textbf{Local Data Structures:} for each level $j$, let $c_j$ be the number of SPs of the form $(x_1a_1, b_1y_1)$ and weight $wt_1$ centered in $\levelgraph_j$ and let $c_j(n)$ be the new ones discovered bythe algorithm. Thus $c=\sum_{j}{c_j}$ and $c_{new}=\sum_{j}{c_j(n)}$. Then,
		\begin{enumerate}[resume]
			\item \label{fproof:flitem2} the value of $\CA_\gamma[j]$, where $\gamma$ is the triple of the form $(x_1a_1, b_1y_1)$ and weight $wt_1$ in $P(x_1,y_1)$, is $c_j$. 
			\item \label{fproof:flitem21}
			If a triple $\gamma$ of the form $(x_1a_1, b_1y_1)$ and weight $wt_1$ is present as an HT in $P^*$, then $P_j^*(x1,y1)$ represents $c_j$ paths. A link to $\gamma$ in $P$ is present in $dict_j$. Moreover, $x_1 \in L_j^*(a_1, y_1)$ (respectively $LC_j^*(a_1, y_1)$ if $x_1$ is centered in $\levelgraph_j$). A similar statement holds for $y_1 \in R_j^*(x_1, b_1)$ (respectively $RC_j^*(x_1, b_1)$ if $y_1$ is centered in $\levelgraph_j$).
		\end{enumerate}
		\item \label{fproof:fitem5} For any key $[wt_2, x_2, y_2 ] \geq [wt, x, y]$, let
		$G'$ contain $c > 0$ number of LHPs of weight  ${ wt_2}$ of the
		form $(x_2a_2, b_2y_2)$. Further, let ${ c_{new}}$ (resp. ${ c_{old }}$) denote the number of such LHPs
		that are {\em new} (resp. not {\em new}).
		Here ${ c_{new} + c_{old} = c}$. Then the tuple $(x_2 a_2, b_2 y_2) \in$ \MTend, iff 
		$c_{old} > 0$ and $c_{new}$ paths have been added to $H_f$ by some earlier iteration of the while loop.
	\end{enumerate}
	
	\noindent
	By Lemma \ref{fdfixcorr}, all the SP distances in $G'$ are placed in $H_f$ and processed by the algorithm. Hence, after Algorithm~\ref{algo:fffixup} is executed, every SP in $G'$ is in its corresponding $P^*$ by the invariant of Lemma~\ref{proof:ffflem}.   
	Since every LST of the form $(xa,by)$ in $G'$ is formed by a left extension of a set of STs of the form $(a\times,by)$ (Step \ref{ffixup:lext} -  Algorithm~\ref{algo:fffixup}), or a right extension of a set of the form $(xa,\times b)$ (analogous steps for right extensions), and all the STs are correctly maintained and extendend (by the invariant of Lemma \ref{proof:ffflem}), then all the LSTs are correctly maintained at the end of \ffullyfixupend. This completes the proof of the Lemma.
\end{proof}

%\appendixhead{PON}

% Bibliography
\bibliographystyle{plain}
\bibliography{references,refs2}

\end{document}